\documentclass[a4paper]{article}
\usepackage{a4wide}
\usepackage[UKenglish]{babel}
\usepackage{amsmath,amssymb,amsthm}
\usepackage{graphicx,xcolor,float,subcaption}
	\graphicspath{{./}{Figure/}}
\usepackage{hyperref}
\usepackage[inline]{enumitem}

\allowdisplaybreaks

\newtheorem{prpstn}[]{Proposition}
\newtheorem{rmrk}[]{Remark}
\newtheorem{lem}[]{Lemma}
\newtheorem{thrm}[]{Theorem}
\numberwithin{equation}{section}
\numberwithin{prpstn}{section}
\numberwithin{rmrk}{section}

\newcommand{\Supp}{\operatorname{Supp}}
\newcommand{\sgn}{\operatorname{sgn}}
\newcommand{\beq}{\begin{equation}}
\newcommand{\eeq}{\end{equation}}
\newcommand{\beqa}{\begin{eqnarray}}
\newcommand{\eeqa}{\end{eqnarray}}
\newcommand{\ddt}{\frac{{\rm d}}{{\rm d}t}}
\newcommand{\di}{\, {\rm d}}
\newcommand{\e}{\varepsilon}
\newcommand{\pa}{\partial}

\title{Modelling coevolutionary dynamics in heterogeneous SI epidemiological systems across scales\thanks{This work was supported by the Italian Ministry of University and Research (MUR) through the PRIN 2020 project (No. 2020JLWP23) ``Integrated Mathematical Approaches to Socio–Epidemiological Dynamics'' (CUP: E15F21005420006) and the PRIN2022-PNRR project (No. P2022Z7ZAJ) ``A Unitary Mathematical Framework for Modelling Muscular Dystrophies'' (CUP: E53D23018070001), and the INdAM group GNFM.}}

\author{Tommaso Lorenzi,\thanks{Department of Mathematical Sciences ``G. L. Lagrange'', Politecnico di Torino, 10129 Torino, Italy (tommaso.lorenzi@polito.it)}
\and
Elisa Paparelli,\thanks{Department of Mathematical Sciences ``G. L. Lagrange'', Politecnico di Torino, 10129 Torino, Italy (elisa.paparelli@polito.it)}
\and
Andrea Tosin \thanks{Department of Mathematical Sciences ``G. L. Lagrange'', Politecnico di Torino, 10129 Torino, Italy (andrea.tosin@polito.it)}
}

\begin{document}
\maketitle

\begin{abstract}
We develop a new structured compartmental model for the coevolutionary dynamics between susceptible and infectious individuals in heterogeneous SI epidemiological systems. In this model, the susceptible compartment is structured by a continuous variable that represents the level of resistance to infection of susceptible individuals, while the infectious compartment is structured by a continuous variable that represents the viral load of infectious individuals. We first formulate an individual-based model wherein the dynamics of single individuals is described through stochastic processes, which permits a fine-grain representation of individual dynamics and captures stochastic variability in evolutionary trajectories amongst individuals. Next we formally derive the mesoscopic counterpart of this model, which consists of a system of coupled integro-differential equations for the population density functions of susceptible and infectious individuals. Then we consider an appropriately rescaled version of this system and we carry out formal asymptotic analysis to derive the corresponding macroscopic model, which comprises a system of coupled ordinary differential equations for the proportions of susceptible and infectious individuals, the mean level of resistance to infection of susceptible individuals, and the mean viral load of infectious individuals. Overall, this leads to a coherent mathematical representation of the coevolutionary dynamics between susceptible and infectious individuals across scales. We provide well-posedness results for the mesoscopic and macroscopic models, and we show that there is excellent agreement between analytical results on the long-time behaviour of the components of the solution to the macroscopic model, the results of Monte Carlo simulations of the individual-based model, and numerical solutions of the macroscopic model. 
\end{abstract}

\section{Introduction}

\paragraph{Background.} Compartmental epidemiological models formulated as ordinary differential equations (ODEs), wherein each differential equation governs the dynamic of the proportion of individuals in a given compartment, have found broad application in theoretical and empirical research in epidemiology~\cite{diekmann2013mathematical,diekmann2000mathematical,2,3,iannelli2015introduction,nowak2000virus}. An implicit assumption underlying these models is homogeneity of the different compartments, in that individuals in each compartment are assumed, as a first approximation, identical. This, however, ignores the heterogeneity typically observed within populations (i.e. the members of the population express different characteristics to different degrees), which has been found to play a pivotal role in the spread of various infectious diseases -- see for instance~\cite{britton2020mathematical,chabas2018evolutionary,elie2022source,jones2021viral,kimberly2016heterogeneity,tkachenko2021time,trauer2019importance,woolhouse1997heterogeneities,yates2006pathogen,zhang2022monitoring} and references therein.

A possible way of incorporating heterogeneity in compartmental epidemiological models consists in structuring one or more compartments by continuous variables that capture variation in relevant characteristics amongst the members of the same compartment -- i.e. individuals expressing a given characteristic to different degrees are characterised by different values of the structuring variable~\cite{metz2014dynamics,perthame2006transport}. The distribution of the expression of such characteristics in a structured compartment is then described by a population density function, whose dynamics are governed by an integro-differential equation (IDE), a partial differential equation (PDE) or a partial integro-differential equation (PIDE), which replaces the ODE for the corresponding unstructured compartment in the original model. 

Continuously structured compartmental epidemiological models have been proposed in a variety of contexts, from age-structured populations~\cite{busenberg1991global,inaba1990threshold,inaba2017age,thieme2009spectral} to populations occupying spatially heterogeneous environments \cite{diekmann1978thresholds,ducasse2022threshold,inaba2012new,lou2011reaction,peng2012reaction,thieme1977model,wang2011nonlocal,wang2012basic}. Closer to the topic of our study, we mention earlier works where continuous structuring variables have been introduced in compartmental epidemiological models to capture variability in phenotypic traits~\cite{abi2021asymptotic,almeida2021final,burie2020asymptotic,burie2020slow,burie2020concentration,djidjou2017steady,karev2019trait,novozhilov2012epidemiological,veliov2005effect}, resistance to infection or immunity level~\cite{barbarossa2015immuno,gandolfi2015epidemic,iacono2012evolution,lorenzi2021evolutionary}, viral or pathogen load (i.e. the quantity of virus or pathogen carried by infectious individuals)~\cite{banerjee2020immuno,della2021sir,della2022sir,gandolfi2015epidemic,loy2021viral}, and socio-economic characteristics~\cite{bernardi2022effects,dimarco2020wealth,dimarco2021kinetic,zanella2022kinetic}. 

\paragraph{Contents of the article.} In this article, we complement earlier works on the development and study of structured compartmental epidemiological models by considering a heterogeneous SI system in which the susceptible compartment is structured by a continuous variable that captures variability in resistance to infection amongst susceptible individuals, while the infectious compartment is structured by a continuous variable that captures variability in viral load amongst infectious individuals. The level of resistance to infection of susceptible individuals and the viral load of infectious individuals may evolve in time, for example because of (epi)genetic changes~\cite{atlante2020epigenetic,gomez2012epigenetics,hill1996genetics,karlsson2014natural,kwok2021host,quintana2019human} and within-host virus dynamics~\cite{du2022within,fraser2007variation,gutierrez2012circulating,hadjichrysanthou2016understanding,hay2021estimating,puhach2023sars}, respectively. Moreover, due to contact with infectious individuals, susceptible individuals may become infectious with a probability that depends both on their level of resistance to infection and the viral load of the infectious individuals they come in contact with. 

Instead of defining a mathematical model that provides an aggregate description of the dynamics of such a heterogeneous SI system on the basis of population-scale phenomenological assumptions, here we first formulate an individual-based model wherein the dynamics of single individuals are described through stochastic processes. This permits a fine-grain representation of individual dynamics and captures stochastic variability in evolutionary trajectories amongst individuals. Next we formally derive the mesoscopic counterpart of this model, which consists of a system of coupled IDEs for the population density functions of susceptible and infectious individuals. Being formally derived from an underlying individual-based model, the terms comprised in such a mesoscopic model provide an accurate mean-field representation of the underlying individual dynamics and inter-individual interactions. Then we consider an appropriately rescaled version of this IDE system and we carry out formal asymptotic analysis to derive the corresponding macroscopic model, which comprises a system of coupled ODEs for the proportions of susceptible and infectious individuals, the mean level of resistance to infection of susceptible individuals, and the mean viral load of infectious individuals. Overall, this leads to a coherent mathematical representation of the coevolutionary dynamics between susceptible and infectious individuals across scales.

\paragraph{Outline of the article.} In Section~\ref{sec:model}, we formulate an individual-based model for the dynamics of the heterogeneous SI system under study and we formally derive the mesoscopic and macroscopic counterparts of this model. In Section~\ref{sec:analysis}, we provide well-posedness results for the mesoscopic and macroscopic models, and we study the long-time behaviour of the solution to the coupled ODEs of the macroscopic model. In Section~\ref{sec:numsim}, we report on the results of Monte Carlo simulations of the individual-based model and integrate them with the results of numerical simulations of the macroscopic model along with the analytical results on the long-time behaviour of the solution to the corresponding ODE system. In Section~\ref{sec:disc}, we conclude with a discussion and propose some future research directions. 

\section{Individual-based, mesoscopic, and macroscopic models}
\label{sec:model}
In this section, we formulate an individual-based model for the SI epidemiological system considered here (see Section~\ref{Dec:micromodel}), and we formally derive first the mesoscopic counterpart of this model (see Section~\ref{Dec:mesomodel}) and then the corresponding macroscopic model (see Section~\ref{Dec:macromodel}).

\subsection{Formulation of the individual-based model}
\label{Dec:micromodel}
Building on the modelling approach employed in~\cite{loy2020boltzmann,loy2021viral}, at time $t \in \mathbb{R}_+$ we describe the microscopic state of the individuals, regarded as being indistinguishable, by the triplet $(X_{t}, R_{t}, V_{t})$, with $X_{t} \in \mathcal{A} := \{I, S\}$, $R_{t} \in \mathcal{R} \subset \mathbb{R}_+$ and $V_{t} \in \mathcal{V} \subset \mathbb{R}_+$, where $\mathcal{R}$ and $\mathcal{V}$ are bounded intervals and $\mathbb{R}_+$ denotes the set of non-negative real numbers. The discrete random variable $X_{t}$ specifies the compartment that an individual belongs to -- i.e. susceptible, $X_{t} = S$, or infectious, $X_{t} = I$ -- while the continuous random variables $R_{t}$ and $V_{t}$ represent, respectively, the level of resistance to infection and the viral load of the individual.
\begin{rmrk}
\label{rmrk:v=0}
In our modelling framework, infectious individuals with zero viral load are infectious individuals whose viral load is below the infectivity threshold and thus cannot transmit the infection to susceptible individuals.
\end{rmrk}

We denote by $f(t,x_i,r,v)$ the probability density function of the triple $(X_{t}, R_{t}, V_{t})$, i.e.
\beq
\label{eq:f}
(X_{t}, R_{t}, V_{t}) \sim f(t,x_i,r,v),
\eeq
and, in the remainder of the article, we will use the more compact notation
\beq
\label{eq:fSfI}
f(t,x_i=S,r,v) =: f_S(t,r,v) \; \text{ and } \;  f(t,x_i=I,r,v) =: f_I(t,r,v)
\eeq
along with the fact that, since $f(x_i,r,v,t)$ is regarded as a probability density function,
\beq
\label{ass:f}
\sum_{x_i \in \mathcal{A}} \int_{\mathcal{R}} \int_{\mathcal{V}} f(t,x_i,r,v) \di v \di r = \int_{\mathcal{R}} \int_{\mathcal{V}} f_S(t,r,v) \di v \di r + \int_{\mathcal{R}} \int_{\mathcal{V}} f_I(t,r,v) \di v \di r= 1 \;\; \forall \, t \geq 0.
\eeq

We focus on a scenario where, between time $t$ and time $t + \Delta t$, with $\Delta t \in \mathbb{R}^*_+$, where $\mathbb{R}^*_+$ denotes the set of positive real numbers:
\begin{itemize}
\item[(i)] the level of resistance to infection of a susceptible individual may change in time at rate $\theta_R\in\mathbb{R}^*_+$ and switch from $r$ to $r'$ with probability $K_S(r'| r)$;
\item[(ii)] the viral load of an infectious individual may change in time at rate $\theta_V\in\mathbb{R}^*_+$ and switch from $v$ to $v'$ with probability $K_I(v'| v)$;
\item[(iii)] contacts between a susceptible individual with level of resistance to infection $r$ and an infectious individual with viral load $v^*$, which occur at rate $\theta_X\in\mathbb{R}^*_+$, may cause the susceptible individual to become infectious, with probability $q(r, v^*)$, and thus acquire the viral load $v''$, with probability $Q(v'' |r, v^*)$.
\end{itemize}
The kernels $K_S$, $K_I$, and $Q$ along with the function $q$ satisfy the following assumptions 
\beq
\label{ass:KSa}
K_S\in L^1(\mathcal{R}; C(\mathcal{R})), \quad \Supp \left(K_S\right) \subseteq \mathcal{R}\times\mathcal{R}, \;\; K_S \geq 0 \; \text{ a.e.}, \;\; \int_{\mathcal{R}} K_S(r' | r) \di r' = 1 \;\; \forall \, r \in  \mathcal{R},
\eeq

\beq
\label{ass:KSb}
\int_{\mathcal{R}} r' \, K_S(r' | r) \di r'  \in  \mathcal{R} \;\; \forall \, r \in  \mathcal{R}, \quad \int_{\mathcal{R}} (r')^2 K_{S}(r' | r) \di r' - \left(\int_{\mathcal{R}} r' K_{S}(r' | r) \di r' \right)^2 > 0 \;\; \forall \, r \in  \mathcal{ \mathring R},
\eeq

\beq
\label{ass:KIa}
\Supp \left(K_I \right) \subseteq \mathcal{V}\times\mathcal{V}, \;\; K_I \geq 0\; \text{ a.e.}, \;\;  \int_{\mathcal{V}} K_I(v' | v) \di v' = 1  \;\; \forall \, v \in  \mathcal{V},
\eeq

\beq
\label{ass:KIb}
\int_{\mathcal{V}} v' \, K_I(v' | v) \di v' \in  \mathcal{V} \;\; \forall \, v \in  \mathcal{V}, \quad \int_{\mathcal{V}} (v')^2 K_{I}(v' | v) \di v' - \left(\int_{\mathcal{V}} v' K_{I}(v' | v) \di v' \right)^2 > 0 \;\; \forall \, v \in  \mathcal{ \mathring V},
\eeq

\beq
\label{ass:Qa}
Q \in L^1(\mathcal{V}; C(\mathcal{R}\times\mathcal{V})), \quad \Supp\left(Q \right) \subseteq \mathcal{V}\times\mathcal{R}\times\mathcal{V}, \;\; Q \geq 0 \; \text{ a.e.}, \;\; \int_{\mathcal{V}} Q(v | r, v^*) \di v = 1 \;\; \forall \, (r, v^*) \in  \mathcal{R} \times \mathcal{V},
\eeq

\beq
\label{ass:Qb}
\int_{\mathcal{V}} v \,  \, Q(v | r, v^*) \di v \in \mathcal{V} \;\; \forall \, (r, v^*) \in  \mathcal{R} \times \mathcal{V},
\eeq

\beq
\label{ass:q}
q\in C(\mathcal{R}\times\mathcal{V}), \quad \Supp\left(q \right)\subseteq \mathcal{R}\times\mathcal{V}, \quad 0 \leq q \leq 1, \quad q(r,0) = 0 \;\; \forall r \in \mathcal{R},
\eeq
where $\Supp\left(\cdot\right)$ denotes the support of the function $\left(\cdot\right)$, $\mathcal{ \mathring R}$ denotes the interior of the interval $\mathcal{R}$, and $\mathcal{ \mathring V}$ denotes the interior of the interval $\mathcal{V}$.

\begin{rmrk}
\label{rmrk:q(r,0)}
The assumption~\eqref{ass:q} on $q(r,0)$ translates in mathematical terms to the biological idea that, as mentioned in Remark~\ref{rmrk:v=0}, infectious individuals with $v=0$ are infectious individuals who carry a viral load which is too low for the infection to spread from them to the susceptible individuals they come in contact with.
\end{rmrk}

Under the scenario corresponding to (i)-(iii), the evolution of a focal individual in the microscopic state $(X_{t}, R_{t}, V_{t})$ is governed by the following system
\beq
\label{eq:microdynlaws}
\begin{cases}
X_{t + \Delta t} = (1 - \Theta_X) \, X_t + \Theta_X \, X'_t,
\\\\
R_{t + \Delta t} = (1 - \Theta_R) \, R_t + \Theta_R \, R'_t,
\\\\
V_{t + \Delta t} = (1 - \Theta_X) \Big[(1 - \Theta_V) \, V_t + \Theta_V \, V'_t \Big] + \Theta_X \, V''_t,
\end{cases}
\eeq
where
\beq
\label{eq:Theta}
\Theta_{\alpha} \sim {\rm Bernoulli}\left(\theta_{\alpha} \, \Delta t \right) \; \text{ with } \; \theta_{\alpha} \in \mathbb{R}^*_+  \; \text{ for } \; \alpha \in \{X, R, V\}.
\eeq
In particular, to integrate (i) and (ii) into~\eqref{eq:microdynlaws}, we let
\beq
\label{eq:Rp}
\left(R'_t | X_t, R_t\right) \sim J_S(r'|x_i, r) \; \text{ with } \; J_S(r'|x_i, r) := \begin{cases}
K_S(r'| r), \; \text{if } x_i = S,
\\\\
\delta_{r}(r'), \;\, \quad \, \text{if } x_i = I,
\end{cases}
\eeq
and
\beq
\label{eq:Vp}
\left(V'_t | X_t, V_t\right) \sim J_I(v'|x_i, v) \; \text{ with } \; J_I(v'|x_i, v) := \begin{cases}
\delta_{v}(v'), \quad\;\, \text{if } x_i = S,
\\\\
K_I(v'| v), \; \text{if } x_i = I,
\end{cases}
\eeq
where $\delta_{b}(a)$ is the Dirac delta centred in $a=b$. The definition~\eqref{eq:Rp} for $J_S$ captures the fact that the component of the microscopic state modelling the level of resistance to infection evolves in time, according to the kernel $K_S$, for susceptible individuals only. Similarly, the definition for $J_I$ given by~\eqref{eq:Vp} models the fact that the component of the microscopic state corresponding to the viral load evolves in time, according to the kernel $K_I$, only for infectious individuals.

Moreover, to integrate (iii) into~\eqref{eq:microdynlaws} we let
\beq
\label{eq:XpVpp}
\left(X'_t, V''_t | X_t, R_t, V_t, X^*_t, V^*_t \right) \sim T(x'_i, v'' | x_i, r, v, x^*_j, v^*) 
\eeq
with 
\beq
\label{eq:T}
\resizebox{.9\textwidth}{!}{$
T(x'_i, v'' | x_i, r, v, x^*_j, v^*) := 
\begin{cases}
q(r, v^*) \, Q(v'' |r, v^*) \, \delta_{I, x'_i} + \left(1 - q(r, v^*)\right) \, \delta_{v''}(v) \, \delta_{S, x'_i}, \;\;\,\, \text{if } x_i = S, \; x^*_j = I,
\\
\delta_{v}(v'') \, \delta_{S, x'_i}, \quad \quad \quad \quad \quad \quad \quad\quad\quad \quad \quad \quad \quad \quad \quad\quad\quad \quad \, \text{if } x_i = S, \; x^*_j = S,
\\
\delta_{v}(v'') \, \delta_{I, x'_i}, \quad \quad \quad \quad \quad \quad \quad\quad\quad \quad \quad \quad \quad \quad \quad\quad\quad \quad \; \text{if } x_i = I, \; x^*_j = I,
\\
\delta_{v}(v'') \, \delta_{I, x'_i}, \quad \quad \quad \quad \quad \quad \quad\quad\quad \quad \quad \quad \quad \quad \quad\quad\quad \quad \; \text{if } x_i = I, \; x^*_j = S,
\end{cases}
$}
\eeq
where $\delta_{\beta, \alpha}$ is the Kronecker delta centred in $\alpha = \beta$. The definition for $T$ given by~\eqref{eq:T} translates in mathematical terms to the biological ideas that: a contact with an infectious individual is required for a susceptible individual to become infectious; if, upon contact with an infectious individual, a susceptible individual does not become infectious then the individual's microscopic state remains unchanged.

\begin{rmrk}
\label{rmrk:T}
The definition given by~\eqref{eq:T} is obtained by rewriting $T(x'_i, v'' | x_i, r, v, x^*_j, v^*)$ as
$$
T(x'_i, v'' | x_i, r, v, x^*_j, v^*) = T_1(v'' | x'_i, x_i, r, v, x^*_j, v^*) \, T_2(x'_i | x_i, r, v, x^*_j, v^*),
$$
and then choosing $T_1$ and $T_2$ appropriately. For instance, the definition given by~\eqref{eq:T} for $x_i = S$ and $x^*_j = I$ is obtained by choosing
$$
T_2(x'_i | x_i=S, r, v, x^*_j=I, v^*) := q(r, v^*) \, \delta_{I, x'_i} + \left(1 - q(r, v^*) \right)  \, \delta_{S, x'_i}
$$
and
$$
T_1(v'' | x'_i, x_i, r, v, x^*_j, v^*) :=
\begin{cases}
\delta_{v''}(v), \qquad\;\, \text{if } x'_i = S,
\\\\
Q(v'' |r, v^*), \;\; \text{if } x'_i = I,
\end{cases}
$$
while definitions~\eqref{eq:T} for the other values of $x_i$ and $x^*_j$ are obtained by choosing
$$
T_2(x'_i | x_i, r, v, x^*_j, v^*) := \delta_{x_i, x'_i}  \quad \text{and} \quad T_1(v'' | x'_i, x_i, r, v, x^*_j, v^*) :=  \delta_{v''}(v).
$$
\end{rmrk}

\subsection{Formal derivation of the corresponding mesoscopic model}
\label{Dec:mesomodel}
In this section, we carry out a formal derivation of the mesoscopic model corresponding to the individual-based model presented in Section~\ref{Dec:micromodel}. In summary, we show that the population density functions of susceptible and infectious individuals, i.e. the functions
\beq
\label{eq:nSnI}
n_S(t,r) := \int_{\mathcal{V}} \, f_S(t,r,v) \di v \quad \text{and} \quad n_I(t,v) := \int_{\mathcal{R}} \, f_I(t,r,v) \di r,
\eeq
which model, respectively, the proportion of susceptible individuals with level of resistance $r$ and the proportion of infectious individuals with viral load $v$ at time $t$, are weak solutions of the following IDE system
\beq
\label{eq:model}
\resizebox{.9\textwidth}{!}{$
\begin{cases}
\displaystyle{\pa_t n_S = - \theta_X \, n_S \, \int_{\mathcal{V}} q(r, v^*) \, n_I(t,v^*) \di v^* + \theta_R \int_{\mathcal{R}} K_S(r | r') \, n_S(t,r') \di r' - \theta_R \, n_S}, 
\\\\
\displaystyle{\pa_t n_I = \theta_X \, \int_{\mathcal{R}} \int_{\mathcal{V}} q(r,v^*) \, Q(v | r, v^*) \, n_S(t,r) \, n_I(t,v^*) \di v^* \di r + \theta_V \int_{\mathcal{V}} K_I(v | v') \, n_I(t,v') \di v' - \theta_V \, n_I},
\end{cases}
$}
\eeq
which we complement with an initial condition of components $n_S(r,0) = n_{S,0}(r)$ and $n_I(v,0) = n_{I,0}(v)$ such that (see also Remark~\ref{rmrk:masscons})
\begin{equation}
\label{eq:ICa}
\displaystyle{n_{S,0} \in L^1 \cap L^{\infty}(\mathcal{R}), \;\; n_{S,0} \geq 0 \; \text{ a.e.}, \quad n_{I,0} \in L^1 \cap L^{\infty}(\mathcal{V}), \;\; n_{I,0} \geq 0 \; \text{ a.e.}}, 
\end{equation}

\begin{equation}
\label{eq:ICb}
\resizebox{.9\textwidth}{!}{$
\displaystyle{\int_{\mathcal{R}} n_{S,0}(r) \di r + \int_{\mathcal{V}} n_{I,0}(v) \di v = 1, \qquad 0 \leq \dfrac{\displaystyle{\int_{\mathcal{R}} r \, n_{S,0}(r) \di r}}{\displaystyle{\int_{\mathcal{R}} n_{S,0}(r) \di r}} \leq \sup \mathcal{R}, \qquad 0 \leq \dfrac{\displaystyle{\int_{\mathcal{V}} v \, n_{I,0}(v) \di v}}{\displaystyle{\int_{\mathcal{V}} n_{I,0}(v) \di v}} \leq \sup \mathcal{V},}
$}
\end{equation}
where $\sup (\cdot)$ denotes the supremum of the interval $(\cdot)$.

\paragraph{General evolution equation for expectations of observables.} Recalling the probability density function $f(t,x_i,r,v)$ given in~\eqref{eq:f}, we note that, since the components of the triple $\left(X_{t + \Delta t}, R_{t + \Delta t}, V_{t + \Delta t}\right)$ are given by~\eqref{eq:microdynlaws}, for any observable $\Phi : \mathcal{A} \times \mathcal{R} \times \mathcal{V} \to \mathbb{R}$ the expectation 
$$
\left<\Phi\left(X_{t}, R_{t},  V_{t}\right)\right> = \sum_{x_i \in \mathcal{A}} \int_{\mathcal{R}} \int_{\mathcal{V}} \Phi(x_i,r,v) \, f(t,x_i,r,v) \di v \di r
$$
satisfies (see, for instance,~\cite{pareschi2013BOOK})
\begin{align}
\label{eq:EVPhi1}
	\left<\Phi\left(X_{t+\Delta{t}},R_{t+\Delta{t}},V_{t+\Delta{t}}\right)\right> &= \theta_X \,\Delta{t}\left<\Phi\left(X'_t,R_{t},V''_t\right)\right>
		+\theta_R\,\Delta{t}\left<\Phi\left(X_t,R'_t,V_t\right)\right> \nonumber \\
	&\phantom{=} +\theta_V\,\Delta{t}\left<\Phi\left(X_t,R_{t},V'_t\right)\right> \nonumber \\
	&\phantom{=} +\left(1-\theta_X\,\Delta{t}-\theta_R\,\Delta{t}-\theta_V\,\Delta{t}\right)\,\left<\Phi\left(X_{t},R_{t},V_{t}\right)\right>+o(\Delta{t}) 
\end{align}
with
\beq
\label{eq:distXpRVpp}
\resizebox{.9\textwidth}{!}{$
\displaystyle\left(X'_t,  R_t, V''_t \right) \sim \sum_{x_i \in \mathcal{A}} \sum_{x^*_j \in \mathcal{A}} \int_{\mathcal{R}} \int_{\mathcal{V}} \int_{\mathcal{V}} T(x'_i, v'' | x_i, r, v, x^*_j, v^*) \, f(t,x_i,r,v) \, f(t,x^*_j, r^*, v^*) \di v \di v^* \di r^*,
$}
\eeq
\beq
\label{eq:distXRpVp}
\resizebox{.9\textwidth}{!}{$
\displaystyle\left(X_t, R'_t, V_t\right) \sim \int_{\mathcal{R}} J_S(r'|x_i, r) \, f(t,x_i,r,v) \di r, \quad \left(X_t, R_t, V'_t\right) \sim \int_{\mathcal{V}} J_I(v'|x_i, v) \, f(t,x_i,r,v) \di v.
$}
\eeq
Hence, rearranging terms in~\eqref{eq:EVPhi1}, dividing through by $\Delta t$ and then letting $\Delta t \to 0^+$, we formally obtain the following evolution equation
\begin{align}
\label{eq:EVPhi}
	\ddt \left<\Phi\left(X_{t},R_{t}, V_{t}\right) \right>  &= \theta_X \, \Big[\left<\Phi\left(X'_t, R_{t}, V''_t\right) \right> - \left<\Phi\left(X_{t}, R_{t},  V_{t}\right) \right> \Big] \nonumber \\
	&\phantom{=} + \, \theta_R \, \Big[\left<\Phi\left(X_t, R'_t, V_t\right) \right> -	\left<\Phi\left(X_{t}, R_{t},  V_{t}\right) \right> \Big] \nonumber \\
	&\phantom{=} +\theta_V \, \Big[\left<\Phi\left(X_t, R_{t}, V'_t \right) \right> - \left<\Phi\left(X_{t}, R_{t},  V_{t}\right) \right> \Big],
\end{align}
which can be regarded as a weak form of the kinetic equation for the probability density function $f$ with $\Phi$ playing the role of an arbitrary test function. Explicitly, expressing the expectations $\langle\cdot\rangle$ as summations/integrations with respect to $f$ and invoking~\eqref{eq:distXpRVpp},~\eqref{eq:distXRpVp} we write:
\begin{align}
\label{eq:kinetic.f}
	& \ddt\sum_{x_i\in\mathcal{A}}\int_\mathcal{R}\int_\mathcal{V}\Phi(x_i,r,v)f(t,x_i,r,v)\di v\di r= \nonumber \\
	& \resizebox{\textwidth}{!}{$\displaystyle\quad \theta_X\sum_{x_i,\,x_j^\ast,\,x_i'\in\mathcal{A}}\iint_{\mathcal{R}^2}\iiint_{\mathcal{V}^3}\Phi(x_i',r,v'')\,T(x'_i, v'' | x_i, r, v, x^*_j, v^*) \, f(t,x_i,r,v) \, f(t,x^*_j, r^*, v^*) \di v \di v^* \di v'' \di r \di r^*$} \nonumber \\
	& \quad +\theta_R\sum_{x_i\in\mathcal{A}}\iint_{\mathcal{R}^2}\int_\mathcal{V}\Phi(x_i,r',v)\,J_S(r'|x_i, r) \, f(t,x_i,r,v) \di v \di r \di r' \nonumber \\
	& \quad +\theta_V\sum_{x_i\in\mathcal{A}}\int_\mathcal{R}\iint_{\mathcal{V}^2}\Phi(x_i,r,v')\,J_I(v'|x_i, v) \, f(t,x_i,r,v) \di v \di v' \di r \nonumber \\
	& \quad -(\theta_X+\theta_R+\theta_V)\sum_{x_i\in\mathcal{A}}\int_\mathcal{R}\int_\mathcal{V}\Phi(x_i,r,v)f(t,x_i,r,v)\di v\di r.
\end{align}

\paragraph{Derivation of the IDE~\eqref{eq:model} for~\texorpdfstring{$\boldsymbol{n_S(t,r)}$}{}.} In order to derive the IDE~\eqref{eq:model} for the population density function $n_S(t,r)$, we choose
$$
\Phi(x_i,r,v) := \delta_{S, x_i} \, \phi(r) \,  {\bf 1}_{\mathcal{V}}(v),
$$
where ${\bf 1}_{\mathcal{V}}(v)$ is the indicator function of the set $\mathcal{V}$ and $\phi \in C^{\infty}\big(\mathcal{R}\big)$. Substituting the above expression of $\Phi(x_i,r,v)$ into~\eqref{eq:kinetic.f}, using the assumptions and definitions given by~\eqref{eq:Rp}-\eqref{eq:T}, and recalling~\eqref{eq:fSfI} and~\eqref{eq:nSnI} yields
\begin{align*}
	\resizebox{\textwidth}{!}{$
	\begin{aligned}[b]
		\ddt\int_{\mathcal{R}}\phi(r)\,n_S(t,r)\di r &= \theta_X\,\int_{\mathcal{R}}\int_{\mathcal{V}}\phi(r)\,{\bf 1}_{\mathcal{V}}(v'')\int_{\mathcal{V}}\int_{\mathcal{V}}\left(1-q(r,v^*)\right) \,
			\delta_{v''}(v) \, f_S(t,r,v) \, n_I(t,v^*) \di v \di v^* \di v'' \di r \\
		&\phantom{=} +\,\theta_X\,\int_{\mathcal{R}}\int_{\mathcal{V}}\phi(r)\,{\bf 1}_{\mathcal{V}}(v'')\int_{\mathcal{R}}\int_{\mathcal{V}}\delta_{v''}(v)\,f_S(t,r,v)\,n_S(t,r^*)\di v\di r^*\di v''\di r \\
		&\phantom{=} +\,\theta_R\,\int_{\mathcal{R}}\int_{\mathcal{V}}\phi(r')\,{\bf 1}_{\mathcal{V}}(v)\int_{\mathcal{R}}K_S(r'|r)\,f_S(t,r,v)\di r\di v\di r' \\
		&\phantom{=} -\,\theta_R\int_{\mathcal{R}}\int_{\mathcal{V}}\phi(r)\,{\bf 1}_{\mathcal{V}}(v)\,f_S(t,r,v)\di v\di r \\
		&= \theta_X\,\int_{\mathcal{R}}\int_{\mathcal{V}}\phi(r)\left(1-q(r,v^*)\right)\,n_S(t,r)\,n_I(t,v^*)\di v^*\di r \\
		&\phantom{=} +\theta_X\,\int_{\mathcal{R}}\int_{\mathcal{R}}\phi(r)\,n_S(t,r)\,n_S(t,r^*)\di r^*\di r-\,\theta_X\int_{\mathcal{R}}\phi(r)\,n_S(t,r)\di r \\
		&\phantom{=} +\,\theta_R\,\int_{\mathcal{R}}\phi(r')\int_{\mathcal{R}}K_S(r'|r)\,n_S(t,r)\di r\di r'-\,\theta_R\int_{\mathcal{R}}\phi(r)\,n_S(t,r)\di r.
	\end{aligned}
	$}
\end{align*}
Moreover, using the fact that (cf. the integral identity~\eqref{ass:f})
$$
\int_{\mathcal{R}} n_S(t,r^*) \di r^* + \int_{\mathcal{V}} n_I(t,v^*) \di v^* = 1 \;\; \forall t \geq 0
$$
along with the fact that
$$
\int_{\mathcal{R}} \phi(r') \int_{\mathcal{R}} K_S(r' | r) \, n_S(t,r) \di r \di r' = \int_{\mathcal{R}} \phi(r) \int_{\mathcal{R}} K_S(r | r') \, n_S(t,r') \di r' \di r,
$$
we find
\beqa
\int_{\mathcal{R}} \phi(r) \, \pa_t n_S(t,r) \di r &=& \int_{\mathcal{R}} \phi(r) \Big[- \theta_X \, n_S(t,r) \int_{\mathcal{V}} q(r,v^*) \, n_I(t,v^*) \di v^* + \nonumber \\
&& + \, \theta_R \int_{\mathcal{R}} K_S(r | r') \, n_S(t,r') \di r'  - \, \theta_R \, n_S(t,r) \Big] \di r, \nonumber
\eeqa
which is a weak formulation of the IDE~\eqref{eq:model} for $n_S(t,r)$.

\paragraph{Derivation of the IDE~\eqref{eq:model} for~\texorpdfstring{$\boldsymbol{n_I(t,v)}$}{}.} In order to derive the IDE~\eqref{eq:model} for the population density function $n_I(t,v)$, we choose
$$
\Phi(x_i,r,v) := \delta_{I, x_i} \, {\bf 1}_{\mathcal{R}}(r) \, \phi(v),
$$
where ${\bf 1}_{\mathcal{R}}(r)$ is the indicator function of the set $\mathcal{R}$ and $\phi \in C^{\infty}\big(\mathcal{V}\big)$. Substituting the above expression of $\Phi(x_i,r,v)$ into~\eqref{eq:kinetic.f}, using assumptions and definitions~\eqref{eq:Rp}-\eqref{eq:T}, and recalling~\eqref{eq:fSfI} and~\eqref{eq:nSnI} yields
\begin{align*}
	\begin{aligned}[b]
		\ddt\int_{\mathcal{V}}\phi(v)\,n_I(t,v)\di v &=
			\theta_X\,\int_{\mathcal{R}}\int_{\mathcal{V}}{\bf 1}_{\mathcal{R}}(r)\,\phi(v'')\int_{\mathcal{V}}q(r,v^*)\,Q(v''|r,v^*)\,n_S(t,r)\,n_I(t,v^*)\di v^*\di v''\di r \\
		&\phantom{=} +\,\theta_X\,\int_{\mathcal{R}}\int_{\mathcal{V}}{\bf 1}_{\mathcal{R}}(r)\,\phi(v'')\int_{\mathcal{V}}\int_{\mathcal{V}}\delta_{v''}(v)\,f_I(t,r,v)\,n_I(t,v^*)\di v\di v^*\di v''\di r \\
		&\phantom{=} +\,\theta_X\,\int_{\mathcal{R}}\int_{\mathcal{V}}{\bf 1}_{\mathcal{R}}(r)\,\phi(v'')\int_{\mathcal{R}}\int_{\mathcal{V}}\delta_{v''}(v)\,f_I(t,r,v)\,n_S(t,r^*)\di v\di r^*\di v''\di r \\
		&\phantom{=} -\,\theta_X\int_{\mathcal{R}}\int_{\mathcal{V}}{\bf 1}_{\mathcal{R}}(r)\,\phi(v)\,f_I(t,r,v)\di v\di r \\
		&\phantom{=} +\,\theta_V\,\int_{\mathcal{V}}\int_{\mathcal{R}}{\bf 1}_{\mathcal{R}}(r)\,\phi(v')\int_{\mathcal{V}}K_I(v'|v)\,f_I(t,r,v)\di v\di r\di v' \\
		&\phantom{=} -\,\theta_V\int_{\mathcal{R}}\int_{\mathcal{V}}{\bf 1}_{\mathcal{R}}(r)\,\phi(v)\,f_I(t,r,v)\di v\di r \\
		&= \theta_X\,\int_{\mathcal{R}}\int_{\mathcal{V}}\phi(v)\int_{\mathcal{V}}q(r,v^*)\,Q(v|r,v^*)\,n_S(t,r)\,n_I(t,v^*)\di v^*\di v\di r \\
		&\phantom{=} +\,\theta_X\,\int_{\mathcal{V}}\int_{\mathcal{V}}\phi(v)\,n_I(t,v)\,n_I(t,v^*)\di v\di v^* \\
		&\phantom{=} +\,\theta_X\,\int_{\mathcal{R}}\int_{\mathcal{V}}\phi(v)\,n_I(t,v)\,n_S(t,r^*)\di v\di r^* \\
		&\phantom{=} -\,\theta_X\int_{\mathcal{V}}\phi(v)\,n_I(t,v)\di v \\
		&\phantom{=} +\,\theta_V\,\int_{\mathcal{V}}\phi(v')\int_{\mathcal{V}}K_I(v'|v)\,n_I(t,v)\di v\di v'-\theta_V\int_{\mathcal{V}}\phi(v)\,n_I(t,v)\di v.
	\end{aligned}
\end{align*}

Moreover, using the fact that (cf. the integral identity~\eqref{ass:f})
$$
\int_{\mathcal{R}} n_S(t,r^*) \di r^* + \int_{\mathcal{V}} n_I(t,v^*) \di v^* = 1 \;\; \forall t \geq 0
$$
along with the fact that
$$
\int_{\mathcal{V}} \phi(v') \int_{\mathcal{V}} K_I(v' | v) \, n_I(t,v) \di v \di v' = \int_{\mathcal{V}} \phi(v) \int_{\mathcal{V}} K_I(v | v') \, n_I(t,v') \di v' \di v,
$$
we find
\beqa
\int_{\mathcal{V}} \phi(v) \, \pa_t n_I(t,v) \di v &=& \int_{\mathcal{V}} \phi(v) \Big[\theta_X \, \int_{\mathcal{R}} \int_{\mathcal{V}} q(r,v^*) \, Q(v | r, v^*) \, n_S(t,r) \, n_I(t,v^*) \di v^* \di r + \nonumber \\
&& + \, \theta_V \int_{\mathcal{V}} K_I(v | v') \, n_I(t,v') \di v'  - \, \theta_V \, n_I(t,v) \Big] \di v, \nonumber
\eeqa
which is a weak formulation of the IDE~\eqref{eq:model} for $n_I(t,v)$.
\begin{rmrk}
\label{rmrk:masscons}
Note that, under assumptions~\eqref{ass:KSa}, \eqref{ass:KIa}, and \eqref{ass:Qa}, adding together the differential equations obtained by integrating the IDE~\eqref{eq:model} for $n_S(t,r)$ over $\mathcal{R}$ and the IDE~\eqref{eq:model} for $n_I(t,v)$ over $\mathcal{V}$ gives
$$
\ddt \left(\int_{\mathcal{R}} n_S(t,r) \di r + \int_{\mathcal{V}} n_I(t,v) \di v \right) = 0,
$$
which implies that, for any initial condition that satisfies assumptions~\eqref{eq:ICa} and~\eqref{eq:ICb}, the integral identity~\eqref{ass:f} holds.
\end{rmrk}

\subsection{Formal derivation of the corresponding macroscopic model}
\label{Dec:macromodel}
Starting from an appropriately rescaled version of the mesoscopic model given by the IDE system~\eqref{eq:model} and introducing appropriate assumptions on the kernels $K_S$ and $K_I$, in this section we formally derive a possible macroscopic counterpart of the individual-based model presented in Section~\ref{Dec:micromodel}, which comprises the system of ordinary differential equations~\eqref{eq:ODEsNSNIMSMIgen} for the proportions $N_S$, $N_I$ of susceptible and infectious individuals, respectively, the mean level of resistance to infection $M_S$ of susceptible individuals, and the mean viral load $M_I$ of infectious individuals, i.e. the quantities
\beq
\label{def:NSI}
N_S(t) := \int_{\mathcal{R}} n_S(t,r) \di r, \quad N_I(t) := \int_{\mathcal{V}} n_I(t,v) \di v,
\eeq
\beq
\label{def:MSI}
N_S(t)M_S(t) := \int_{\mathcal{R}} r \, n_S(t,r) \di r, \quad N_I(t)M_I(t) := \int_{\mathcal{V}} v \, n_I(t,v) \di v.
\eeq

\paragraph{Key underlying assumptions on~\texorpdfstring{$\boldsymbol{K_S}$}{} and~\texorpdfstring{$\boldsymbol{K_I}$}{}.} Note that, under assumptions~\eqref{ass:KSa}-\eqref{ass:KIb}, one can easily prove that, for all functions $\Lambda_S \in C(\mathbb{R}_+, \mathcal{R})$ and $\Lambda_I \in C(\mathbb{R}_+, \mathcal{V})$, the integral operators
\beq\label{def:intopK}
\mathcal{K}_S[\psi](t,r) = \int_{\mathcal{R}} K_S(r | r') \, \psi(t,r') \di r', \quad \mathcal{K}_I[\psi](t,v) = \int_{\mathcal{V}} K_I(v | v') \, \psi(t,v') \di v'
\eeq
admit, respectively, a non-negative eigenfuction $\psi_S \in C\left(\mathbb{R}_+; L^1(\mathcal{R})\right)$ and a non-negative eigenfuction $\psi_I \in C\left(\mathbb{R}_+; L^1(\mathcal{V})\right)$ associated with the eigenvalue $1$ and satisfying the following conditions
\beq\label{def:intopKconst}
\int_{\mathcal{R}} \psi_S(t,r)  \di r = 1, \quad \int_{\mathcal{R}} r \, \psi_S(t,r) \di r = \Lambda_S(t), \quad \int_{\mathcal{V}} \psi_I(t,v)  \di v = 1, \quad \int_{\mathcal{V}} v \, \psi_I(t,v) \di v = \Lambda_I(t).
\eeq
Moreover, we let the kernels $K_S$ and $K_I$ satisfy the following additional assumptions
\beq
\label{ass:Krv}
\int_{\mathcal{R}} r \, K_S(r | r') \di r = r' \;\; \forall \, r' \in  \mathcal{R}, \quad \int_{\mathcal{V}} v \, K_I(v | v') \di v = v' \;\; \forall \, v' \in  \mathcal{V}.
\eeq

\begin{rmrk}
\label{rem:Krv}
Assumptions~\eqref{ass:Krv} correspond to a biological scenario where the level of resistance to infection of susceptible individuals and the viral load of infectious individuals remain, on average, unaltered. 
\end{rmrk}

\paragraph{Rescaled mesoscopic model.} We introduce a small parameter $\e \in \mathbb{R}^*_+$ and let 
\beq
\label{ass:thetaseps}
\theta_X = \e, \quad \theta_R = O(1) \; \text{ and } \; \theta_V = O(1) \;\; \text{for } \e \to 0^+.
\eeq

\begin{rmrk}
\label{rem:theta}
The assumptions given by~\eqref{ass:thetaseps} correspond to a biological scenario where contacts between susceptible and infectious individuals that lead to disease transmission occur on a slower time scale than changes in the level of resistance to infection of susceptible individuals and in the viral load of infectious individuals.
\end{rmrk}

Furthermore, under assumptions~\eqref{ass:thetaseps}, in order to capture the effect of disease transmission driven by contacts between susceptible and infectious individuals (cf. Remark~\ref{rem:theta}), we use the time scaling $t \to t / \e$ in the IDE system~\eqref{eq:model} and, in so doing, we obtain the following system of IDEs for the population density functions $n^{\e}_S(t,r) \equiv n_S(\frac{t}{\e},r)$ and $n^{\e}_I(t,v) \equiv n_I(\frac{t}{\e},v)$
\beq
\label{eq:modeleps}
\resizebox{.9\textwidth}{!}{$
\begin{cases}
\displaystyle{\e \, \pa_t n^{\e}_S = - \e \, n^{\e}_S \, \int_{\mathcal{V}} q(r, v^*) \, n^{\e}_I(t,v^*) \di v^* + \theta_R \, \int_{\mathcal{R}} K_S(r | r') \, n^{\e}_S(t,r') \di r' - \theta_R \, n^{\e}_S}, 
\\\\
\displaystyle{\e \, \pa_t n^{\e}_I = \e \, \int_{\mathcal{R}} \int_{\mathcal{V}} q(r,v^*) \, Q(v | r, v^*) \, n^{\e}_S(t,r) \, n^{\e}_I(t,v^*) \di v^* \di r + \theta_V \int_{\mathcal{V}} K_I(v | v') \, n^{\e}_I(t,v') \di v' -  \theta_V \, n^{\e}_I},
\end{cases}
$}
\eeq
which we complement with the initial conditions 
\beq
\label{eq:ICeps}
n^{\e}_S(r,0)=n^0_S(r), \quad n^{\e}_I(v,0)=n^0_I(v),
\eeq
where $n^0_S$ and $n^0_I$ satisfy assumptions~\eqref{eq:ICa} and~\eqref{eq:ICb}. 

Integrating both sides of the IDE~\eqref{eq:modeleps}$_1$ over $\mathcal{R}$ and both sides of the IDE~\eqref{eq:modeleps}$_2$ over $\mathcal{V}$, and then using assumptions~\eqref{ass:KSa}, \eqref{ass:KIa}, and \eqref{ass:Qa} produces the following system of differential equations 
\beq
\label{eq:ODEsNSNIeps}
\begin{cases}
\displaystyle{\ddt \int_{\mathcal{R}} n^{\e}_S(t,r) \di r = - \int_{\mathcal{R}} \int_{\mathcal{V}} q(r, v^*) \, n^{\e}_S(t,r) \, n^{\e}_I(t,v^*) \, \di v^* \di r},
\\\\
\displaystyle{\ddt \int_{\mathcal{R}} n^{\e}_I(t,v) \di v = \int_{\mathcal{R}} \int_{\mathcal{V}} q(r,v^*) \, n^{\e}_S(t,r) \, n^{\e}_I(t,v^*) \, \di v^* \di r}.
\end{cases}
\eeq
Similarly, first multiplying both sides of the IDE~\eqref{eq:modeleps}$_1$ by $r$ and integrating over $\mathcal{R}$ the resulting IDE, then multiplying both sides of the IDE~\eqref{eq:modeleps}$_2$ by $v$ and integrating over $\mathcal{V}$ the resulting IDE, and finally using the additional assumptions~\eqref{ass:Krv} produces the following system of differential equations 
\beq
\label{eq:ODEsMSMIeps}
\begin{cases}
\displaystyle{\ddt \int_{\mathcal{R}} r \, n^{\e}_S(t,r) \di r = - \int_{\mathcal{V}} \left(\int_{\mathcal{R}} r \, q(r, v^*) \, n^{\e}_S(t,r) \, \di r \right) \, n^{\e}_I(t,v^*) \di v^*},
\\\\
\displaystyle{\ddt \int_{\mathcal{R}} v \, n^{\e}_I(t,v) \di v = \int_{\mathcal{R}} \int_{\mathcal{V}} \overline{Q}(r, v^*) \, q(r,v^*) \, n^{\e}_S(t,r) \, n^{\e}_I(t,v^*) \di v^* \di r},
\end{cases}
\eeq
where
\beq
\label{def:overlineQ}
\overline{Q}(r, v^*) := \int_{\mathcal{V}} v \,  \, Q(v | r, v^*) \di v.
\eeq

\paragraph{Formal asymptotic analysis for~\texorpdfstring{$\boldsymbol{\e\to 0^+}$}{}.} Under assumptions~\eqref{ass:Qa}, \eqref{ass:Qb}, and \eqref{ass:Krv}, denoting by $n_S(t,r)$ and $n_I(t,v)$ the leading-order terms of the asymptotic expansions for $n^{\e}_S(t,r)$ and $n^{\e}_I(t,v)$, letting $\e \to 0^+$ in the IDE system~\eqref{eq:modeleps} and then using the definitions given by~\eqref{def:NSI} and \eqref{def:MSI} formally gives
\beq
\label{eq:nSnIfact}
\begin{cases}
\displaystyle{\int_{\mathcal{R}} K_S(r | r') \, n_S(t,r') \di r' - n_S(t,r) = 0}
\\\\
\displaystyle{\int_{\mathcal{V}} K_I(v | v') \, n_I(t,v') \di v' -  n_I(t,v) = 0}
\end{cases}
\quad \Longrightarrow \quad
\begin{cases}
\displaystyle{n_S(t,r) = \psi_S(t,r) \, N_S(t)}
\\\\
\displaystyle{n_I(t,v) = \psi_I(t,v) \, N_I(t),}
\end{cases}
\eeq
where the non-negative functions $\psi_S \in C\left(\mathbb{R}_+; L^1(\mathcal{R})\right)$ and $\psi_I \in C\left(\mathbb{R}_+; L^1(\mathcal{V})\right)$ are eigenfunctions of the operators $\mathcal{K}_S$ and $\mathcal{K}_I$, defined via~\eqref{def:intopK}, associated with the eigenvalue $1$ and satisfying conditions~\eqref{def:intopKconst} with $\Lambda_S \equiv M_S$ and $\Lambda_I \equiv M_I$, that is,
\beq
\label{ass:KSKIaddMSMI}
\resizebox{.9\textwidth}{!}{$
\begin{cases}
 \displaystyle{\int_{\mathcal{R}} K_S(r | r') \, \psi_S(t,r') \di r' = \psi_S(t,r)},
 \\\\
 \displaystyle{\int_{\mathcal{R}} \psi_S(t,r)  \di r = 1, \; \int_{\mathcal{R}} r \, \psi_S(t,r) \di r = M_S(t),}
\end{cases}
\begin{cases}
 \displaystyle{\int_{\mathcal{V}} K_I(v | v') \, \psi_I(t,v') \di v'  = \psi_I(t,v)},
 \\\\
 \displaystyle{\int_{\mathcal{V}} \psi_I(t,v)  \di v = 1, \; \int_{\mathcal{V}} v \, \psi_I(t,v) \di v = M_I(t).}
\end{cases}
$}
\eeq
 
Moreover, letting $\e \to 0^+$ in~\eqref{eq:ODEsNSNIeps}, substituting the expressions of $n_S(t,r)$ and $n_I(t,v)$ given by~\eqref{eq:nSnIfact} in the resulting system of differential equations, and then using the definitions given by~\eqref{def:NSI}, we formally obtain the following system of differential equations 
\beq
\label{eq:ODEsNSNI}
\begin{cases}
\displaystyle{\ddt N_S = - N_S \, N_I \int_{\mathcal{R}} \int_{\mathcal{V}} q(r, v^*) \, \psi_S(t,r) \, \psi_I(t,v^*) \di v^* \di r},
\\\\
\displaystyle{\ddt N_I = N_I \, N_S \int_{\mathcal{R}} \int_{\mathcal{V}} q(r,v^*) \, \psi_S(t,r) \, \psi_I(t,v^*) \di v^* \di r}.
\end{cases}
\eeq
Similarly, letting $\e \to 0^+$ in~\eqref{eq:ODEsMSMIeps}, substituting the expressions of $n_S(t,r)$ and $n_I(t,v)$ given by~\eqref{eq:nSnIfact} in the resulting system of differential equations, and then using the definitions given by~\eqref{def:NSI} and \eqref{def:MSI}, formally yields the following system of differential equations 
\beq
\label{eq:ODEsMSMI}
\begin{cases}
\displaystyle{\ddt \left(N_S \, M_S \right) = - N_S \, N_I \int_{\mathcal{V}} \left(\int_{\mathcal{R}} r \, q(r, v^*) \, \psi_S(t,r) \di r \right) \, \psi_I(t,v^*) \di v^*},
\\\\
\displaystyle{\ddt \left(N_I \, M_I \right) = N_I \, N_S \int_{\mathcal{R}} \int_{\mathcal{V}} \overline{Q}(r, v^*) \, q(r,v^*) \, \psi_S(t,r) \, \psi_I(t,v^*) \di v^* \di r}.
\end{cases}
\eeq
Finally, substituting the differential equation~\eqref{eq:ODEsNSNI}$_1$ into the differential equation~\eqref{eq:ODEsMSMI}$_1$ and the differential equation~\eqref{eq:ODEsNSNI}$_2$ into the differential equation~\eqref{eq:ODEsMSMI}$_2$, after a little algebra we find the following ODE system
\beq
\label{eq:ODEsNSNIMSMIgen}
\begin{cases}
\displaystyle{\ddt N_S = - N_S \, N_I \int_{\mathcal{R}} \int_{\mathcal{V}} q(r, v^*) \, \psi_S(t,r) \, \psi_I(t,v^*) \di v^* \di r},
\\\\
\displaystyle{\ddt N_I = N_I \, N_S \int_{\mathcal{R}} \int_{\mathcal{V}} q(r,v^*) \, \psi_S(t,r) \, \psi_I(t,v^*) \di v^* \di r},
\\\\
\displaystyle{\ddt M_S = N_I \int_{\mathcal{V}} \left(\int_{\mathcal{R}} \left(M_S - r\right) \, q(r, v^*) \, \psi_S(t,r) \di r \right) \, \psi_I(t,v^*) \di v^*},
\\\\
\displaystyle{\ddt M_I = N_S \int_{\mathcal{R}} \int_{\mathcal{V}} \left(\overline{Q}(r, v^*) - M_I\right) \, q(r,v^*) \, \psi_S(t,r) \, \psi_I(t,v^*) \di v^* \di r},
\end{cases}
\eeq
which is complemented with~\eqref{def:overlineQ} and~\eqref{ass:KSKIaddMSMI}, and is subject to the following initial conditions (cf. assumptions~\eqref{eq:ICa} and~\eqref{eq:ICb} and the definitions given by~\eqref{def:NSI} and \eqref{def:MSI}) 
\beq
\label{eq:ODEsNSNIMSMIICgen}
N_S(0) = N_{S,0} \in (0,1), \quad N_I(0) = 1- N_{S,0}, \quad M_S(0) = M_{S,0} \in \mathcal{R}, \quad M_I(0) = M_{I,0} \in \mathcal{V}.
\eeq

\begin{rmrk}
A simplified version of the macroscopic model is obtained in the case when, instead of satisfying the additional assumptions~\eqref{ass:Krv}, the kernels $K_S$ and $K_I$ satisfy the following additional assumptions
\beq
\label{ass:KSKI_analysis2}
K_S(r | r') \equiv K_S(r), \quad K_I(v | v') \equiv K_I(v),
\eeq
which correspond to a biological scenario where the new level of resistance to infection of susceptible individuals and the new viral load of infectious individuals that are acquired as a result of a change are independent of the original ones. This case is covered in detail in Appendix~\ref{sec:appD}.
\end{rmrk}

\section{Analytical results}
\label{sec:analysis}
In this section, we establish well-posedness of the mesoscopic and the macroscopic models derived in the previous section (see Section~\ref{sec:analysisWP}), and we study the long-time behaviour of the solution to the macroscopic model.

\subsection{Well-posedness of the mesoscopic and the macroscopic models}
\label{sec:analysisWP}
\paragraph{Well-posedness of the mesoscopic model.} Using a method similar to those employed in~\cite{borra2013asymptotic,delitala2012asymptotic} (see Appendix~\ref{sec:appA}), which relies on the Banach fixed-point theorem, one can prove the following well-posedness result for the mesoscopic model~\eqref{eq:model}:
\begin{thrm}
\label{theo:wpide}
Under assumptions~\eqref{ass:KSa}, \eqref{ass:KIa}, \eqref{ass:Qa}, \eqref{ass:q} and~\eqref{eq:Theta}, the Cauchy problem~\eqref{eq:model}-\eqref{eq:ICb} admits a unique solution of non-negative components $n_S \in C\left(\mathbb{R}_+; L^1(\mathcal{R})\right)$ and $n_I \in C\left(\mathbb{R}_+; L^1(\mathcal{V})\right)$, which satisfy the following a priori estimate
$$
\| n_S(t,\cdot)\|_{L^1(\mathcal{R})} + \| n_I(t,\cdot)\|_{L^1(\mathcal{V})} = 1 \;\; \forall \, t \geq 0.
$$
\end{thrm}

Results analogous to those established by Theorem~\ref{theo:wpide} hold, for any $\e > 0$, also for the Cauchy problem defined by complementing the rescaled mesoscopic model~\eqref{eq:modeleps} with initial conditions~\eqref{eq:ICeps}.

\paragraph{Well-posedness of the macroscopic model.} The following well-posedness result for the macroscopic model defined by the ODE system~\eqref{eq:ODEsNSNIMSMIgen} coupled with the eigenproblems~\eqref{ass:KSKIaddMSMI} can be proved by using standard Cauchy-Lipschitz theory and, therefore, its proof is omitted here. We omit also the proofs of the uniform bounds~\eqref{eq:unifestNSNIMSMI} on $N_S(t)$, $N_I(t)$ and $M_I(t)$, which can be obtained through simple estimates, while we provide the proof of the upper bound~\eqref{eq:unifestNSNIMSMI} on $M_S(t)$ (see  Appendix~\ref{sec:appB}), since it requires a little bit more work.
\begin{thrm}
\label{theo:wpode}
Let the non-negative functions $\psi_S\in C\left(\mathbb{R}_+; L^1(\mathcal{R})\right)$ and $\psi_I\in C\left(\mathbb{R}_+; L^1(\mathcal{V})\right)$ satisfy~\eqref{ass:KSKIaddMSMI}. Under assumptions~\eqref{ass:Qa}-\eqref{ass:q}, the Cauchy problem~\eqref{eq:ODEsNSNIMSMIgen}-\eqref{eq:ODEsNSNIMSMIICgen} complemented with~\eqref{def:overlineQ} admits a unique solution of components $N_S, N_I, M_S, M_I \in C\left(\mathbb{R}_+\right)$, which satisfy the following a priori estimates
\beq
\label{eq:unifestNSNIMSMI}
\resizebox{.9\textwidth}{!}{$
0 \leq N_S(t) \leq 1, \;\;  0 \leq N_I(t) \leq 1, \;\; 0 \leq M_S(t) \leq \sup \mathcal{R}, \;\; 0 \leq M_I(t) \leq \sup \mathcal{V}, \;\; N_S(t)+N_I(t)=1
$}
\eeq
for all $t\geq 0$.
\end{thrm}

\subsection{Long-time behaviour of the macroscopic model}
\label{sec:analysisLTB}
In the remainder of the article, without loss of generality, we will be letting
\beq
\label{def:domainsRV}
\mathcal{R} := [0,1], \quad \mathcal{V} := [0,1].
\eeq
Note that, under the definitions given by~\eqref{def:domainsRV}, the uniform bounds~\eqref{eq:unifestNSNIMSMI} specify to
\beq
\label{eq:aprioriest}
0 \leq N_S(t) \leq 1, \quad 0 \leq N_I(t) \leq 1, \quad 0 \leq M_S(t) \leq 1, \quad 0 \leq M_I(t) \leq 1, \quad N_S(t)+N_I(t)=1
\eeq
for all $t\geq 0$.

Under assumptions~\eqref{ass:KSa}-\eqref{ass:KIb} and \eqref{ass:Krv}, the eigenproblems~\eqref{ass:KSKIaddMSMI} can be solved explicitly, as established by Lemma~\ref{lemma:psiSpsiI}, the proof of which is provided in Appendix~\ref{sec:appC}. 

\begin{lem}
\label{lemma:psiSpsiI}
Let the kernels $K_S$ and $K_I$ satisfy assumptions~\eqref{ass:KSa}-\eqref{ass:KIb} and \eqref{ass:Krv}. Then, under assumptions~\eqref{def:domainsRV}, the unique non-negative solutions $\psi_S,\,\psi_I\in C\left(\mathbb{R}_+; L^1([0,1])\right)$ of the eigenproblems~\eqref{ass:KSKIaddMSMI} are
\beq
\label{eq:psiSpsiI}
\psi_S(t,r) = (1-M_S(t)) \, \delta_0(r) + M_S(t) \, \delta_1(r), \quad \psi_I(t,v) = (1-M_I(t)) \, \delta_0(v) + M_I(t) \, \delta_1(v)
\eeq
for all $t\geq 0$.
\end{lem}
Substituting expressions~\eqref{eq:psiSpsiI} into the ODE system~\eqref{eq:ODEsNSNI}-\eqref{eq:ODEsMSMI} and using the fact that $q(r,0) = 0$ for all $r \in [0,1]$ (cf. assumptions~\eqref{ass:q}) yields
\beq
\label{eq:ODEsNSNIMSMI0}
\begin{cases}
\displaystyle{\ddt N_S = - N_S N_I M_I \Big[q(0,1) (1-M_S) + q(1,1) M_S \Big]},
\\\\
\displaystyle{\ddt N_I = N_S N_I M_I \Big[q(0,1) (1-M_S) + q(1,1) M_S \Big]},
\\\\
\displaystyle{\ddt \left(N_S \, M_S \right) = - N_S M_S N_I M_I q(1,1)},
\\\\
\displaystyle{\ddt \left(N_I \, M_I \right) = N_S N_I M_I \Big[\overline{Q}(0,1) q(0,1) (1-M_S) + \overline{Q}(1,1) q(1,1) M_S \Big]}.
\end{cases}
\eeq
Similarly, inserting~\eqref{eq:psiSpsiI} in the ODE system~\eqref{eq:ODEsNSNIMSMIgen}, after a little algebra we find the following system of ODEs
\beq
\label{eq:ODEsNSNIMSMI1}
\begin{cases}
\displaystyle{\ddt N_S = - N_S N_I M_I \Big[q(0,1) (1-M_S) + q(1,1) M_S \Big]},
\\\\
\displaystyle{\ddt N_I = N_S N_I M_I \Big[q(0,1) (1-M_S) + q(1,1) M_S \Big]},
\\\\
\displaystyle{\ddt M_S = N_I M_I \left(q(0,1) - q(1,1)\right) (1-M_S) M_S},
\\\\
\displaystyle{\ddt M_I = N_S \Big[q(0,1) (1-M_S) + q(1,1) M_S \Big]} \times
\\
\phantom{\ddt M_I = N_S} \quad \times \displaystyle{\left[\dfrac{\overline{Q}(0,1) q(0,1) (1-M_S) + \overline{Q}(1,1) q(1,1) M_S}{q(0,1) (1-M_S) + q(1,1) M_S } - M_I \right] M_I},
\end{cases}
\eeq
which we complement with the following initial condition (cf. assumptions~\eqref{eq:ODEsNSNIMSMIICgen})
\beq
\label{eq:ODEsNSNIMSMIIC}
N_S(0) = N_{S,0} \in (0,1), \quad N_I(0) = 1- N_{S,0}, \quad M_S(0) = M_{S,0} \in (0,1), \quad M_I(0) = M_{I,0} \in (0,1).
\eeq
Note that, under assumptions~\eqref{ass:Qa}, \eqref{ass:Qb}, \eqref{ass:q}, and \eqref{def:domainsRV}, the following conditions hold
\beq
\label{ass:qQbounds}
0 \leq q(0,1) \leq 1, \quad 0 \leq q(1,1) \leq 1, \quad 0 \leq \overline{Q}(0,1) \leq 1, \quad 0 \leq \overline{Q}(1,1) \leq 1.
\eeq
In particular, from now on we will also be assuming that
\beq
\label{ass:qQ0111pos}
\overline{Q}(0,1) > 0, \quad \overline{Q}(1,1) > 0.
\eeq

A complete characterisation of the long-term limit of the components of the solution to the Cauchy problem~\eqref{eq:ODEsNSNIMSMI1}-\eqref{eq:ODEsNSNIMSMIIC}, under alternative biological scenarios corresponding to different assumptions on the probability for a susceptible individual with the minimum level of resistance to infection $r=0$ or the maximum level of resistance to infection $r=1$ to become infectious due to contact with an infectious individual with the maximum viral load $v^*=1$ (i.e. different values of $q(0,1)$ and $q(1,1)$, respectively) is provided by Propositions~\ref{prop1}-\ref{prop3}. 

Proposition~\ref{prop1} shows that if $q(0,1) > q(1,1) > 0$ then all individuals in the system will eventually become infectious.

\begin{prpstn}
\label{prop1}
Let conditions~\eqref{ass:qQbounds}-\eqref{ass:qQ0111pos} hold and assume $0 < q(1,1) < q(0,1) \leq 1$. Then the components of the solution to the Cauchy problem~\eqref{eq:ODEsNSNIMSMI1}-\eqref{eq:ODEsNSNIMSMIIC} are such that
\beq
\label{eq:asylim1a}
\lim_{t \to \infty} N_S(t) = 0, \quad \lim_{t \to \infty} N_I(t) = 1, \quad  \lim_{t \to \infty} M_S(t) = 1,
\eeq
and
\beq
\label{eq:asylim1b}
\lim_{t \to \infty} M_I(t) = N_{I,0} M_{I,0} + N_{S,0} \Big[\overline{Q}(0,1) (1-M_{S,0}) + \overline{Q}(1,1) M_{S,0} \Big].
\eeq
\end{prpstn}

\begin{proof}
We start by noting that the system of differential equations~\eqref{eq:ODEsNSNIMSMI1} can be rewritten as
\beq
\label{eq:ODEsNSNIMSMI1rev}
\begin{cases}
\displaystyle{\ddt N_S = - F_1(N_I,M_S,M_I) N_S},
\\\\
\displaystyle{\ddt N_I = F_2(N_S,M_S,M_I) N_I},
\\\\
\displaystyle{\ddt M_S = F_3(N_I,M_I) (1-M_S) M_S},
\\\\
\displaystyle{\ddt M_I = F_4(N_S,M_S) \left[G(M_S) - M_I \right] M_I},
\end{cases}
\eeq
where
\begin{align*}
 & F_1(N_I,M_S,M_I) := N_I M_I \Big[q(0,1) (1-M_S) + q(1,1) M_S \Big], \\[3mm]
 & F_2(N_S,M_S,M_I) := N_S M_I \Big[q(0,1) (1-M_S) + q(1,1) M_S \Big], \\[3mm]
 & F_3(N_I,M_I) := N_I M_I \left(q(0,1) - q(1,1)\right), \\[3mm]
 & F_4(N_S,M_S) := N_S \Big[q(0,1) (1-M_S) + q(1,1) M_S\Big], \\[3mm]
 & G(M_S) := \dfrac{\overline{Q}(0,1) q(0,1) (1-M_S) + \overline{Q}(1,1) q(1,1) M_S}{q(0,1) (1-M_S) + q(1,1) M_S }.
\end{align*}

{\it Positivity of $N_I(t)$ and $M_I(t)$.} When conditions~\eqref{ass:qQbounds}-\eqref{ass:qQ0111pos} hold, the a priori estimates~\eqref{eq:aprioriest} ensure that
$$
F_2[N_S,M_S,M_I](t) \geq 0, \quad F_4[N_S,M_S](t) \geq 0, \quad G[M_S](t)>0 \quad \forall \, t \geq 0
$$
and, therefore, the solutions to the differential equations~\eqref{eq:ODEsNSNIMSMI1rev}$_2$ and~\eqref{eq:ODEsNSNIMSMI1rev}$_4$ subject, respectively, to the initial conditions $N_I(0)$ and $M_I(0)$ given by~\eqref{eq:ODEsNSNIMSMIIC} are such that  
\beq
\label{eq:NIMIpositive}
N_I(t) > 0, \quad M_I(t) > 0 \quad \forall \, t \geq 0.
\eeq   
\\
{\it Asymptotic behaviour of $N_S(t)$, $N_I(t)$ and $M_S(t)$ for $t \to \infty$.} The positivity results~\eqref{eq:NIMIpositive} along with the a priori estimates~\eqref{eq:aprioriest} ensure that 
$$
F_1[N_I,M_S,M_I](t) > 0 \quad \forall \, t \geq 0.
$$
Moreover, when $0 < q(1,1) < q(0,1) \leq 1$, the positivity results~\eqref{eq:NIMIpositive} also ensure that 
$$
F_3[N_I,M_I](t) > 0 \quad \forall \, t \geq 0.
$$
Hence, the solutions to the differential equations~\eqref{eq:ODEsNSNIMSMI1rev}$_1$ and~\eqref{eq:ODEsNSNIMSMI1rev}$_3$ subject, respectively, to the initial conditions $N_S(0)$ and $M_S(0)$ given by~\eqref{eq:ODEsNSNIMSMIIC} are such that the asymptotic results~\eqref{eq:asylim1a} on $N_S(t)$ and $M_S(t)$ hold. Furthermore, the asymptotic result~\eqref{eq:asylim1a} on $N_I(t)$ can easily be obtained by using the asymptotic result~\eqref{eq:asylim1a} on $N_S(t)$ along with the fact that $N_I(t)=1-N_S(t)$ for all $t \geq 0$ (cf. the a priori estimates~\eqref{eq:aprioriest}).
\\\\
{\it Asymptotic behaviour of $M_I(t)$ for $t \to \infty$.} 
Dividing the differential equation \eqref{eq:ODEsNSNIMSMI1}$_3$ by the differential equation \eqref{eq:ODEsNSNIMSMI1}$_1$ yields
\beq \nonumber
\frac{\di M_S}{\di N_S}= \frac{N_IM_IM_S(1-M_S)(q(0,1)-q(1,1))}{-N_SN_IM_I[(1-M_S)q(0,1)+M_Sq(1,1)]}=\frac{M_S(1-M_S)(q(0,1)-q(1,1))}{-N_S[(1-M_S)q(0,1)+M_Sq(1,1)]}.
\eeq
Hence, under the initial conditions $N_S(0)$ and $M_S(0)$ given by~\eqref{eq:ODEsNSNIMSMIIC}, we have
\begin{eqnarray*}
&& \int_{M_{S,0}}^{M_S} \Biggr( \frac{ q(0,1)}{q(0,1)-q(1,1)}\frac{1}{M_S}+\frac{q(1,1)}{M_S(q(1,1)-q(0,1))+q(0,1)-q(1,1)} \Biggl )\di M_S=-\int_{N_{S,0}}^{N_S} \frac{\di N_S}{N_S}
\\\\
&& 
\resizebox{.95\textwidth}{!}{$
\displaystyle\Longrightarrow \frac{q(0,1)}{q(0,1)-q(1,1)}\log{\frac{M_S}{M_{S,0}}}+\frac{q(1,1)}{q(1,1)-q(0,1)}\log{\frac{|M_S(q(1,1)-q(0,1))+q(0,1)-q(1,1)|}{|M_{S,0}(q(1,1)-q(0,1))+q(0,1)-q(1,1)|}}=\log{\frac{N_{S,0}}{N_S}},
$}
\end{eqnarray*}
from which, introducing the notation
\beq \label{p}
\frac{q(0,1)}{q(0,1)-q(1,1)} =: p \quad \text{and} \quad \frac{q(1,1)}{q(1,1)-q(0,1)} =: 1-p,
\eeq
we obtain
$$ \Biggr( \frac{M_S}{M_{S,0}} \Biggl) ^p \Biggr( \frac{|(1-M_S)(q(0,1)-q(1,1))|}{|(1-M_{S,0})(q(0,1)-q(1,1))|} \Biggl) ^{1-p}= \frac{N_{S,0}}{N_S}. $$
From the latter equation, using the fact that $1-M_{S,0}>0$ and $1-M_S \geq 0$ (cf. assumptions~\eqref{eq:ODEsNSNIMSMIIC} and the a priori estimates~\eqref{eq:aprioriest}), we obtain
\beq \label{NSfuncMS}
N_S=\frac{k}{M_S^p(1-M_S)^{1-p}}, \qquad k := M_{S,0}^p(1-M_{S,0})^{1-p}N_{S,0}.
\eeq

Moreover, introducing the notation $P_I := N_I M_I$ and combining the differential equations~\eqref{eq:ODEsNSNIMSMI1}$_1$, \eqref{eq:ODEsNSNIMSMI1}$_3$ and \eqref{eq:ODEsNSNIMSMI0}$_4$ yields the following system of differential equations 
\beq
\label{eq:ODEsNSMSPI_generale}
\begin{cases}
\displaystyle{\ddt N_S = -N_SP_I\biggl[q(0,1)(1-M_S)+q(1,1)M_S\biggr]},
\\\\
\displaystyle{\ddt  M_S  =P_I(q(0,1)-q(1,1))(1-M_S)M_S },
\\\\
\displaystyle{\ddt P_I  = P_IN_S \biggl[ q(0,1)\overline{Q}(0,1)(1-M_S)+q(1,1)\overline{Q}(1,1)M_S\biggr]}.
\end{cases}
\eeq
Dividing the differential equation~\eqref{eq:ODEsNSMSPI_generale}$_3$ by the differential equation~\eqref{eq:ODEsNSMSPI_generale}$_2$ and substituting the expression for $N_S$ given by~\eqref{NSfuncMS} into the resulting differential equation gives 
$$
\dfrac{\di P_I}{\di M_S} = \frac{k[ q(0,1)\overline{Q}(0,1)(1-M_S)+q(1,1)\overline{Q}(1,1)M_S]}{M_S^p(1-M_S)^{1-p}M_S(1-M_S)(q(0,1)-q(1,1))}.
$$
Hence, under the initial conditions $P_I(0) := N_I(0) M_I(0)$ and $M_S(0)$ given by~\eqref{eq:ODEsNSNIMSMIIC}, we have
\begin{eqnarray*}
&& \int_{P_{I,0}}^{P_I} \di P_I= \frac{k}{q(0,1)-q(1,1)}\int_{M_{S,0}}^{M_S} \frac{q(0,1)\overline{Q}(0,1)(1-M_S)+q(1,1)\overline{Q}(1,1)M_S}{M_S^{p+1}(1-M_S)^{2-p}} \di M_S
\\\\
&& 
\resizebox{.9\textwidth}{!}{$
\displaystyle\Longrightarrow P_I-P_{I,0}=\frac{k}{q(0,1)-q(1,1)} \dfrac{(1-M_{S,0})^{p-1}}{M_{S,0}^{p}} \frac{\overline Q(0,1)q(0,1)(M_{S,0}(1-p)+p-1) + p \, \overline Q(1,1)q(1,1)M_{S,0}}{(p-1)p} + 
$}
\\\\
&&
\resizebox{.9\textwidth}{!}{$
\displaystyle\phantom{\Longrightarrow P_I-P_{I,0}=} - \frac{k}{q(0,1)-q(1,1)} \dfrac{(1-M_{S})^{p-1}}{M_{S}^{p}} \frac{\overline Q(0,1)q(0,1)(M_S(1-p)+p-1) + p \, \overline Q(1,1)q(1,1)M_S}{(p-1)p}.
$}
\end{eqnarray*}
Since $P_I:=N_I M_I$ and $P_{I,0}:=N_{I,0} M_{I,0}$, recalling the definition of $k$ given by~\eqref{NSfuncMS} alongside the definitions of $p$ and $1-p$ given by~\eqref{p}, and using the fact that if $0 < q(1,1) < q(0,1) \leq 1$ then $p>1$ along with the asymptotic results~\eqref{eq:asylim1a} on $M_S(t)$ and $N_I(t)$, from the latter equation we obtain the asymptotic result~\eqref{eq:asylim1a} on $M_I(t)$.
\end{proof}

Moreover, Proposition~\ref{prop2} shows that if $q(0,1) > q(1,1) = 0$ then an endemic equilibrium is eventually established.

\begin{prpstn}
\label{prop2}
Let conditions~\eqref{ass:qQbounds}-\eqref{ass:qQ0111pos} hold and assume $q(1,1)=0$ and $0 < q(0,1) \leq 1$. Then the components of the solution to the Cauchy problem~\eqref{eq:ODEsNSNIMSMI1}-\eqref{eq:ODEsNSNIMSMIIC} are such that
\beq
\label{eq:asylim2a}
\lim_{t \to \infty} N_S(t) = N_{S,0} M_{S,0}, \quad \lim_{t \to \infty} N_I(t) = 1-N_{S,0} M_{S,0}, \quad  \lim_{t \to \infty} M_S(t) = 1,
\eeq
and
\beq
\label{eq:asylim2b}
\lim_{t \to \infty} M_I(t) = \overline{Q}(0,1) - \dfrac{\left(\overline{Q}(0,1)-M_{I,0}\right) \left(1-N_{S,0}\right)}{1-N_{S,0} M_{S,0}}.
\eeq
\end{prpstn}

\begin{proof} {\it Asymptotic behaviour of $M_S(t)$ for $t \to \infty$.} The asymptotic result~\eqref{eq:asylim2a} on $M_S(t)$ can be proved using a method similar to that employed in the proof of Proposition~\ref{prop1}.
\\\\
{\it Asymptotic behaviour of $N_S(t)$ and $N_I(t)$ for $t \to \infty$.} When $q(1,1)=0$ and $0 < q(0,1) \leq 1$, under the initial condition~\eqref{eq:ODEsNSNIMSMIIC}, using the differential equations~\eqref{eq:ODEsNSNIMSMI1}$_1$ and~\eqref{eq:ODEsNSNIMSMI1}$_3$ one finds 
\beq \label{eq:NSMS}
\dfrac{{\rm d} N_S}{{\rm d} M_S} = - \dfrac{N_S}{M_S} \quad \Longrightarrow \quad N_S(t) = N_{S,0} \dfrac{M_{S,0}}{M_S(t)} \quad \forall \, t \geq 0.
\eeq
This along with the asymptotic result~\eqref{eq:asylim2a} on $M_S(t)$ leads to the asymptotic result~\eqref{eq:asylim2a} on $N_S(t)$. Furthermore, the asymptotic result~\eqref{eq:asylim1a} on $N_I(t)$ can easily be obtained by using the asymptotic result~\eqref{eq:asylim2a} on $N_S(t)$ along with the fact that $N_I(t)=1-N_S(t)$ for all $t \geq 0$ (cf. the a priori estimates~\eqref{eq:aprioriest}).
\\\\
{\it Asymptotic behaviour of $M_I(t)$ for $t \to \infty$.} 
When $q(1,1)=0$ and $0 < q(0,1) \leq 1$, dividing the differential equation~\eqref{eq:ODEsNSNIMSMI1}$_4$ by the differential equation~\eqref{eq:ODEsNSNIMSMI1}$_3$ and substituting the expression for $N_S$ given by \eqref{eq:NSMS} into the resulting differential equation yields
\beq \nonumber
 \frac{\di M_I}{\di M_S}=\frac{M_{S,0}N_{S,0}(\overline{Q}(0,1)-M_I)}{M_S(M_S-M_{S_{0}}N_{S_{0}})}.
 \eeq
 Hence, under the initial conditions $M_S(0)$ and $M_I(0)$ given by~\eqref{eq:ODEsNSNIMSMIIC}, we have
\begin{eqnarray}
 \label{curvaMSMI}
&&   \int_{M_{I_{0}}}^{M_I} \frac{\di M_I}{\overline{Q}(0,1)-M_I}=\int_{M_{S_{0}}}^{M_S} \Biggr( -\frac{1}{M_S}+\frac{1}{M_S-M_{S,0}N_{S,0}} \Biggl ) \di M_S \nonumber 
\\
\nonumber
\\
&& \Longrightarrow \frac{|\overline{Q}(0,1)-M_{I,0}|}{|\overline{Q}(0,1)-M_I|}= \frac{M_{S,0}}{M_S}    \frac{|M_S-M_{S_{0}}N_{S_{0}}|}{|M_{S_{0}}(1-N_{S_{0}})|} \nonumber
\\
\nonumber
\\
&& \Longrightarrow \frac{|\overline{Q}(0,1)-M_I||M_S-M_{S_{0}}N_{S_{0}}|}{M_S}=\frac{|\overline{Q}(0,1)-M_{I,0}||M_{S,0}(1-N_{S_{0}})|}{M_{S,0}}.
\end{eqnarray}
One can easily show that when $q(1,1)=0$ and $0 < q(0,1) \leq 1$ the differential equation~\eqref{eq:ODEsNSNIMSMI1rev}$_4$ subject to the initial condition $M_I(0)$ given by~\eqref{eq:ODEsNSNIMSMIIC} is such that 
$$
\sgn\left(\overline{Q}(0,1)-M_I(t) \right) = \sgn\left(\overline{Q}(0,1)-M_{I,0} \right) \quad \forall \, t \geq 0
$$
and, therefore, equation~\eqref{curvaMSMI} simplifies to 
$$
 \frac{(\overline{Q}(0,1)-M_I) |M_S-M_{S_{0}}N_{S_{0}}|}{M_S}=\frac{(\overline{Q}(0,1)-M_{I,0}) |M_{S,0}(1-N_{S_{0}})|}{M_{S,0}}.
$$
Using assumptions~\eqref{eq:ODEsNSNIMSMIIC} on $N_{S,0}$ and $M_{S,0}$ along with the asymptotic results~\eqref{eq:asylim2a} on $M_S(t)$ and $N_S(t)$, from the latter equation we obtain the asymptotic result~\eqref{eq:asylim2b} on $M_I(t)$.
\end{proof} 

Furthermore, Proposition~\ref{prop3} shows that also in the case where $q(0,1) = q(1,1) > 0$ all individuals in the system will eventually become infectious.

\begin{prpstn}
\label{prop3}
Let conditions~\eqref{ass:qQbounds}-\eqref{ass:qQ0111pos} hold and assume $q(0,1)=q(1,1)=q^*$ with $0 < q^* \leq 1$. Then the components of the solution to the Cauchy problem~\eqref{eq:ODEsNSNIMSMI1}-\eqref{eq:ODEsNSNIMSMIIC} are such that
\beq
\label{eq:asylim3a}
\lim_{t \to \infty} N_S(t) = 0, \quad \lim_{t \to \infty} N_I(t) = 1, \quad  M_S(t) = M_{S,0} \; \forall \, t \geq 0,
\eeq
and
\beq
\label{eq:asylim3b}
\lim_{t \to \infty} M_I(t) = N_{I,0} M_{I,0} + N_{S,0} \Big[\overline{Q}(0,1) (1-M_{S,0}) + \overline{Q}(1,1) M_{S,0} \Big].
\eeq
\end{prpstn}

\begin{proof} {\it Behaviour of $M_S(t)$.} When $q(0,1)=q(1,1)=q^*$ with $0 < q^* \leq 1$, one has $F_3(N_I,M_I)=0$ in the differential equation~\eqref{eq:ODEsNSNIMSMI1rev}$_3$ and, therefore, under the initial condition $M_S(0)$ given by~\eqref{eq:ODEsNSNIMSMIIC}, the result~\eqref{eq:asylim3a} on $M_S$ holds.
\\\\
{\it Asymptotic behaviour of $N_S(t)$ and $N_I(t)$ for $t \to \infty$.} Exploiting the result~\eqref{eq:asylim3a} on $M_S$, the asymptotic results~\eqref{eq:asylim3a} on $N_S(t)$ and $N_I(t)$ can be proved using a method similar to that employed in the proof of Proposition~\ref{prop1}.
\\\\
{\it Asymptotic behaviour of $M_I(t)$ for $t \to \infty$.} 
When $q(0,1)=q(1,1)=q^*$ with $0 < q^* \leq 1$, dividing the differential equation~\eqref{eq:ODEsNSMSPI_generale}$_3$ by the differential equation~\eqref{eq:ODEsNSMSPI_generale}$_1$ and substituting the expression for $M_S$ given by \eqref{eq:asylim3a} into the resulting differential equation yields
\beq \nonumber
\frac{\di P_I}{\di N_S}=-\frac{P_IN_S [\overline Q(0,1)(1-M_{S,0})+\overline Q(1,1)M_{S,0}]}{P_IN_S}.
\eeq
Hence, under the initial conditions $P_I(0) := N_I(0) M_I(0)$ and $N_S(0)$ given by~\eqref{eq:ODEsNSNIMSMIIC}, we have
$$
P_I-P_{I,0}=-\bigl(\overline Q(0,1)(1-M_{S,0})+\overline Q(1,1)M_{S,0}\bigr)\bigl(N_S-N_{S,0}\bigr).
$$
Since $P_I:=N_I M_I$ and $P_{I,0}:=N_{I,0} M_{I,0}$, using the asymptotic results~\eqref{eq:asylim3a} on $N_S(t)$ and $N_I(t)$, from the latter equation we obtain the asymptotic result~\eqref{eq:asylim3b} on $M_I(t)$.
\end{proof}

\section{Numerical simulations}
\label{sec:numsim}
In this section, we present a sample of results of numerical simulations. In Section~\ref{sec:num:setup}, we describe the set-up of numerical simulations and the numerical methods employed to carry them out. In Section~\ref{sec:num:verif}, we compare the numerical solutions of the macroscopic model defined via the ODE system~\eqref{eq:ODEsNSNIMSMI1} with the analytical results established by Propositions~\ref{prop1}-\ref{prop3} and the results of Monte Carlo simulations of the individual-based model, obtained under assumptions~\eqref{ass:Krv}-\eqref{ass:thetaseps} (i.e. the assumptions under which the macroscopic model is formally obtained from the mesoscopic model). 

\subsection{Set-up of numerical simulations and numerical methods}
\label{sec:num:setup}
\paragraph{Monte Carlo simulations of the individual-based model.} Monte Carlo simulations of the individual-based model are performed in {\sc Matlab}. We choose a uniform discretisation of step $10^{-3}$ of the interval $[0,t_f]$ as the computational domain of the independent variable $t$ and, under assumptions~\eqref{def:domainsRV}, we choose a uniform discretisation of step $10^{-4}$ of the interval $[0,1]$ as the computational domain of the independent variables $r$ and $v$. We focus on a system of $10^6$ agents that are initially distributed amongst the susceptible and infectious compartments according to the following population density functions, which satisfy assumptions~\eqref{eq:ICb} under the definitions given by~\eqref{def:domainsRV}:
\beq
\label{def:ICIDE}
n^0_S(r) := 0.9 \, \delta_{M_{S,0}}(r), \quad n^0_I(r) := 0.1 \, \delta_{M_{I,0}}(v), \quad M_{S,0}, M_{I,0} \in (0,1).
\eeq
Hence, at the initial time of simulations $t=0$, the system comprises $9 \times 10^5$ susceptible individuals, all with level of resistance $M_{S,0}$, and $10^5$ infectious individuals, all with viral load $M_{I,0}$. We consider different values of the parameters $M_{S,0}$ and $M_{I,0}$, in order to explore a variety of epidemiological scenarios.

We let assumptions~\eqref{ass:thetaseps} hold, with $\theta_X = \e=10^{-3}$ and $\theta_R=\theta_V=1$, and we focus on the case where the kernels $K_S$ and $K_I$ are defined as
\beq
\label{def:KSKI}
K_S(r | r') := (1-r') \, \delta_0(r) + r' \, \delta_1(r), \quad K_I(v | v') := (1-v') \, \delta_0(v) + v' \, \delta_1(v).
\eeq
The definitions given by~\eqref{def:KSKI} satisfy assumptions~\eqref{ass:KSa}-\eqref{ass:KIb} and~\eqref{ass:Krv} and correspond to a stylised scenario where: changes in the level of resistance to infection lead susceptible individuals with level of resistance $r'$ to acquire either the minimum level of resistance $r=0$ or the maximum level of resistance $r=1$; the probability of acquiring the level of resistance $r=0$ increases as $r' \to 0$ while the probability of acquiring the level of resistance $r=1$ increases as $r' \to 1$. Similarly: changes in the viral load lead infectious individuals with viral load $v'$ to acquire either the minimum viral load $v=0$ or the maximum viral load $v=1$; the probability of acquiring the viral load $v=0$ increases as $v' \to 0$ while the probability of acquiring the viral load $v=1$ increases as $v' \to 1$.

Moreover, we define the function $q$ either as
\beq
\label{def:qProp12}
q(r,v^*):=\frac{\kappa_1 - r}{\kappa_2 + r}v^*, \quad \kappa_1, \kappa_2 \in \mathbb{R}, \quad \kappa_2 \geq \kappa_1 \geq 1,
\eeq
with $\kappa_1$ and $\kappa_2$ such that the assumptions on the function $q$ which underlie Proposition~\ref{prop1} or Proposition~\ref{prop2} are satisfied, or as
\beq
\label{def:qProp3}
q(r,v^*) \equiv q(v^*):= \kappa_3 \, v^*, \quad \kappa_3 \in \mathbb{R}, \quad \kappa_3 > 0
\eeq
so that the assumptions on the function $q$ which underlie Proposition~\ref{prop3} are met. The definitions given by~\eqref{def:qProp12} and~\eqref{def:qProp3} translate in mathematical terms to the biological idea that the probability for a susceptible individual with level of resistance $r$ to become infected due to contact with an infectious individual with viral load $v^*$ increases as the viral load of the infectious individual increases. In particular, the definition given by~\eqref{def:qProp12} corresponds to the (perhaps more realistic) scenario where this probability decreases as the level of resistance of the susceptible individual increases.  

\begin{rmrk}
Additional numerical simulations were carried out under the following definition of the function $q$
$$
q(r,v^*) := (1 - r) \, v^*
$$
and/or the following definitions of the kernels $K_S$ and $K_I$
$$
K_S(r|r') := (r'-1)^2\delta_0(r)+2r'(1-r')\delta_{\frac{1}{2}}(r)+(r')^2\delta_1(r), 
$$
$$
K_I(v|v') := (v'-1)^2\delta_0(v)+2v'(1-v')\delta_{\frac{1}{2}}(v)+(v')^2\delta_1(v).
$$
The results obtained are omitted here since they are analogous to those displayed in Figures~\ref{fig:prop1}-\ref{fig:prop3}.  
\end{rmrk}

Finally, we define the kernel $Q(v|r,v^*)$ either as a uniform probability distribution on the interval $[0,1]$ for all $(r,v^*) \in [0,1] \times [0,1]$, i.e.
\beq
\label{def:Qunif}
Q(v|r,v^*) \equiv 1,
\eeq
or as a truncated normal probability distribution on the interval $[0,1]$ for all $(r,v^*) \in [0,1] \times [0,1]$, 
i.e.
\beq
\label{def:Qtrunc}
Q(v|r,v^*) := \frac{1}{\sqrt {2 \pi} \sigma} \frac{\exp\Bigl(-\dfrac{1}{2} \left(\frac{v-\mu(r,v^*)}{\sigma}\right)^2\Bigr)}{\Psi\left(\frac{1-\mu(r,v^*)}{\sigma}\right)-\Psi\left(\frac{-\mu(r,v^*)}{\sigma}\right)}
\eeq
with
\beq
\label{def:phimusigma}
\Psi(x) := \frac{1}{2}\left(1+ \frac{2}{\sqrt \pi} \int^{\frac{x}{\sqrt 2}}_0 e^{-t^2} \di t\right), \quad \mu(r,v^*) := \Biggl( \frac{3}{4}-\frac{1}{2}r \Biggr )\Biggl(\frac{3}{4}v^*+\frac{1}{4} \Biggr), \quad \sigma = 0.1.
\eeq
In particular, the definition given by~\eqref{def:Qtrunc}-\eqref{def:phimusigma} corresponds to a biological scenario where the mean value of the viral load $v$ acquired by a susceptible individual with original level of resistance $r$ that becomes infected due to contact with an infectious individual with viral load $v^*$ (i.e. the value of $\overline{Q}(r, v^*)$ defined via~\eqref{def:overlineQ}) decreases with the level of resistance of the susceptible individual and increases with the viral load of the infectious individual. Note that under the definition given by~\eqref{def:Qunif} we have
$$
\overline{Q}(r, v^*) \equiv \dfrac{1}{2},
$$
whereas under the definition given by~\eqref{def:Qtrunc}-\eqref{def:phimusigma} we have
$$
\overline{Q}(r, v^*) = \mu(r,v^*)+\frac{1}{\sqrt {2 \pi} } \frac{\exp\Bigl(-\dfrac{1}{2} \left(\frac{1-\mu(r,v^*)}{\sigma}\right)^2 \Bigr)-\exp\Bigl(-\dfrac{1}{2} \left(\frac{-\mu(r,v^*)}{\sigma}\right)^2 \Bigr)}{\Psi\left(\frac{1-\mu(r,v^*)}{\sigma}\right)-\Psi\left(\frac{-\mu(r,v^*)}{\sigma}\right)}\sigma
$$
and, in particular,
$$
\overline Q(0,1)=0.7482 \quad \text{and} \quad \overline Q(1,1)=0.2518.
$$
Hence, both the definition given by~\eqref{def:Qunif} and definition given by~\eqref{def:Qtrunc}-\eqref{def:phimusigma} meet conditions~\eqref{ass:qQ0111pos}.

\paragraph{Numerical simulations of the corresponding macroscopic model.} The corresponding macroscopic model defined via the ODE system~\eqref{eq:ODEsNSNIMSMI1} posed on the interval $(0,t_f]$ and subject to initial conditions corresponding to the population density functions~\eqref{def:ICIDE}, which are such that assumptions~\eqref{eq:ODEsNSNIMSMIIC} are satisfied, i.e.
 \beq
\label{def:ICODE}
N_S(0) = 0.9, \quad N_I(0) = 0.1, \quad M_S(0) = M_{S,0} \in (0,1), \quad M_I(0) = M_{I,0}\in (0,1),
\eeq
is solved numerically in {\sc Matlab} using the built-in function {\sc ode45}, which is based on a fourth order Runge-Kutta method.

\subsection{Main results of numerical simulations}
\label{sec:num:verif}
\subsubsection*{Agreement between analytical and numerical results under the assumptions of Proposition~\ref{prop1}}
Figure~\ref{fig:prop1} displays a sample of numerical results obtained in the case where the probability of infection $q$ is defined via~\eqref{def:qProp12} with the values of the parameters $\kappa_1$ and $\kappa_2$ such that the assumptions of Proposition~\ref{prop1} are satisfied (i.e. $0<q(1,1)<q(0,1)\le1$). Figures~\ref{q1R1V1}-\ref{q1R1V9} refer to the case where the kernel $Q$ is a defined via~\eqref{def:Qunif}, while Figures~\ref{gauR1V1}-\ref{gauR1V9} refer to the case where the kernel $Q$ is defined via~\eqref{def:Qtrunc}-\eqref{def:phimusigma}. In agreement with the analytical results established by Proposition~\ref{prop1}, these numerical results show that the components of the solution to the ODE system~\eqref{eq:ODEsNSNIMSMI1} converge to the corresponding long-term limits~\eqref{eq:asylim1a} and~\eqref{eq:asylim1b} as $t$ increases. 

\subsubsection*{Agreement between analytical and numerical results under the assumptions of Proposition~\ref{prop2}}
Figure~\ref{fig:prop2} displays a sample of numerical results obtained in the case where the probability of infection $q$ is defined via~\eqref{def:qProp12} with the values of the parameters $\kappa_1$ and $\kappa_2$ such that the assumptions of Proposition~\ref{prop2} are satisfied (i.e. $q(1,1)=0$ and $0 < q(0,1) \leq 1$). Figures~\ref{UniProp2R9V1}-\ref{UniProp2R1V9} refer to the case where the kernel $Q$ is defined via~\eqref{def:Qunif}, while Figures~\ref{prop2R9V1}-\ref{prop2R1V9} refer to the case where the kernel $Q$ is defined via~\eqref{def:Qtrunc}-\eqref{def:phimusigma}. In agreement with the analytical results established by Proposition~\ref{prop2}, these numerical results show that the components of the solution to the ODE system~\eqref{eq:ODEsNSNIMSMI1} converge to the corresponding long-term limits~\eqref{eq:asylim2a} and~\eqref{eq:asylim2b} as $t$ increases. 

\subsubsection*{Agreement between analytical and numerical results under the assumptions of Proposition~\ref{prop3}}
A sample of numerical results for the case where the function $q$ is defined via~\eqref{def:qProp3} with $\kappa_3=0.5$ so that $q(0,1)=q(1,1)=0.5$, and thus the assumptions of Proposition~\ref{prop3} are satisfied, is displayed in Figure~\ref{fig:prop3}. Figures~\ref{UniProp3R9V1}-\ref{UniProp3R1V9} refer to the case where the kernel $Q$ is a defined as a uniform probability distribution via~\eqref{def:Qunif}, while Figures~\ref{prop3R9V1}-\ref{prop3R1V9} refer to the case where the kernel $Q$ is defined as a truncated normal probability distribution via~\eqref{def:Qtrunc}-\eqref{def:phimusigma}. In agreement with the analytical results established by Proposition~\ref{prop3}, these numerical results show that the components of the solution to the ODE system~\eqref{eq:ODEsNSNIMSMI1} converge to the corresponding long-term limits~\eqref{eq:asylim3a} and~\eqref{eq:asylim3b} as $t$ increases. 

\subsubsection*{Agreement between the results of Monte Carlo simulations of the individual-based model and numerical solutions of the macroscopic model}
Taken together, the numerical results presented in Figures~\ref{fig:prop1}-\ref{fig:prop3} indicate that there is an excellent quantitative agreement between the dynamics of the proportion of susceptible individuals, $N_S(t)$, the proportion of infectious individuals, $N_I(t)$, the mean level of resistance to infection, $M_S(t)$, and the mean viral load, $M_I(t)$, obtained from Monte Carlo simulations of the individual-based model and the dynamics of the corresponding quantities obtained by solving numerically the macroscopic model defined via the ODE system~\eqref{eq:ODEsNSNIMSMI1}, both in the case where the kernel $Q$ is a defined as a uniform probability distribution via~\eqref{def:Qunif} (cf. Figures~\ref{q1R1V1}-\ref{q1R1V9}, Figures~\ref{UniProp2R9V1}-\ref{UniProp2R1V9}, and Figures~\ref{UniProp3R9V1}-\ref{UniProp3R1V9}) and in the case where the kernel $Q$ is defined as a truncated normal probability distribution via~\eqref{def:Qtrunc}-\eqref{def:phimusigma} (cf. Figures~\ref{gauR1V1}-\ref{gauR1V9}, Figures~\ref{prop2R9V1}-\ref{prop2R1V9}, and Figures~\ref{prop3R9V1}-\ref{prop3R1V9}). 

Note that in the cases where $N_S(t) \to 0$ as $t \to \infty$, $M_S(t)$ obtained from Monte Carlo simulations of the individual-based model may undergo oscillations and then deviates from the asymptotic trend predicted by the macroscopic model when $N_S(t)$ attains values sufficiently close to zero. This is to be expected as the formula used to compute the mean level of resistance to infection from the results of simulations is an empirical average over a number of susceptible individuals which decreases to zero. This is apparent in Figure~\ref{fig:prop3}.

\begin{figure}[H]
\centering
\subcaptionbox{\label{q1R1V1}}
{\includegraphics[trim={0 7cm 0 7cm},clip, width=.49\textwidth]{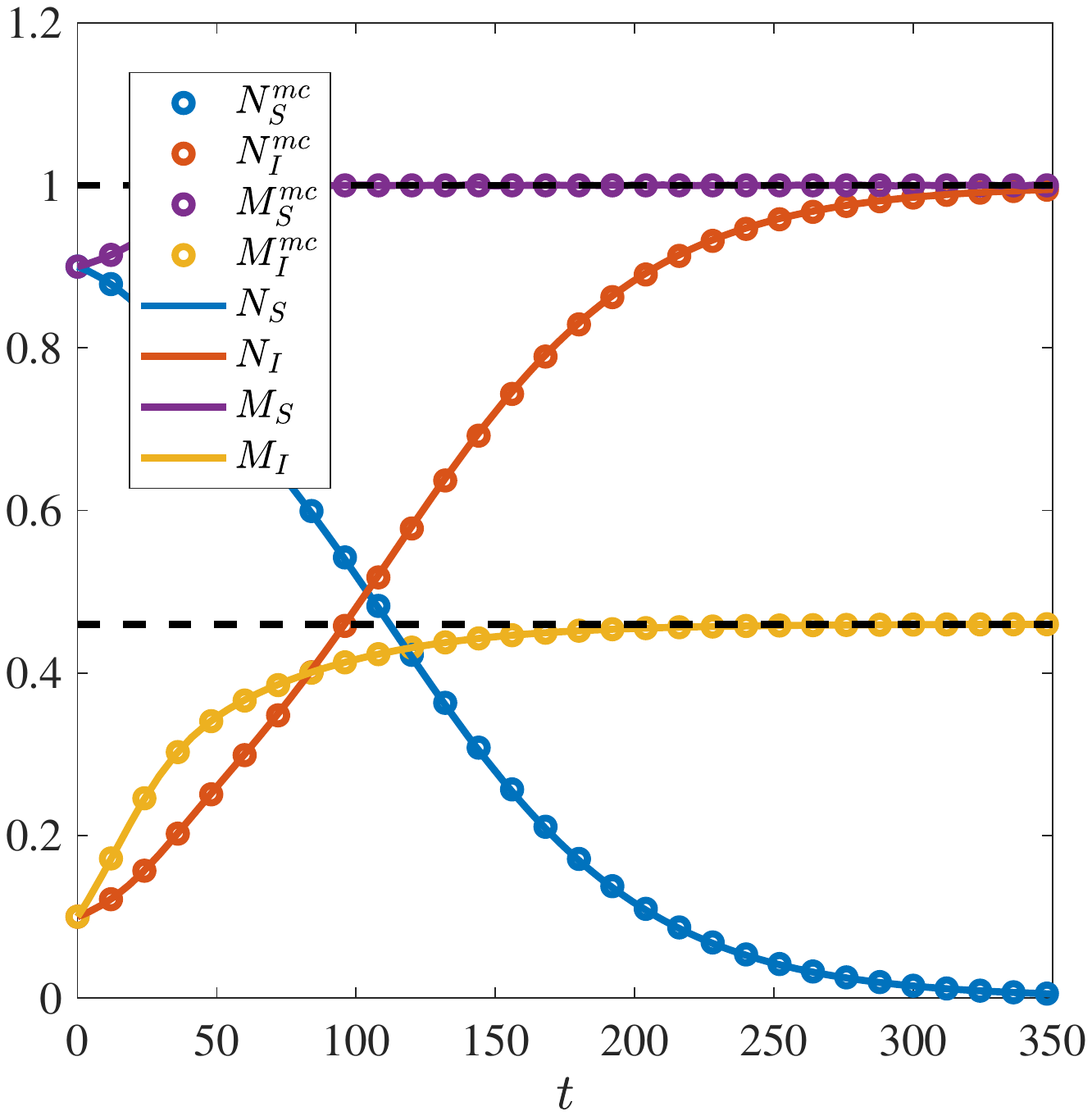}}
\subcaptionbox{\label{q1R1V9}} 
{\includegraphics[trim={0cm 7cm 0 7cm},clip, width=.49\textwidth]{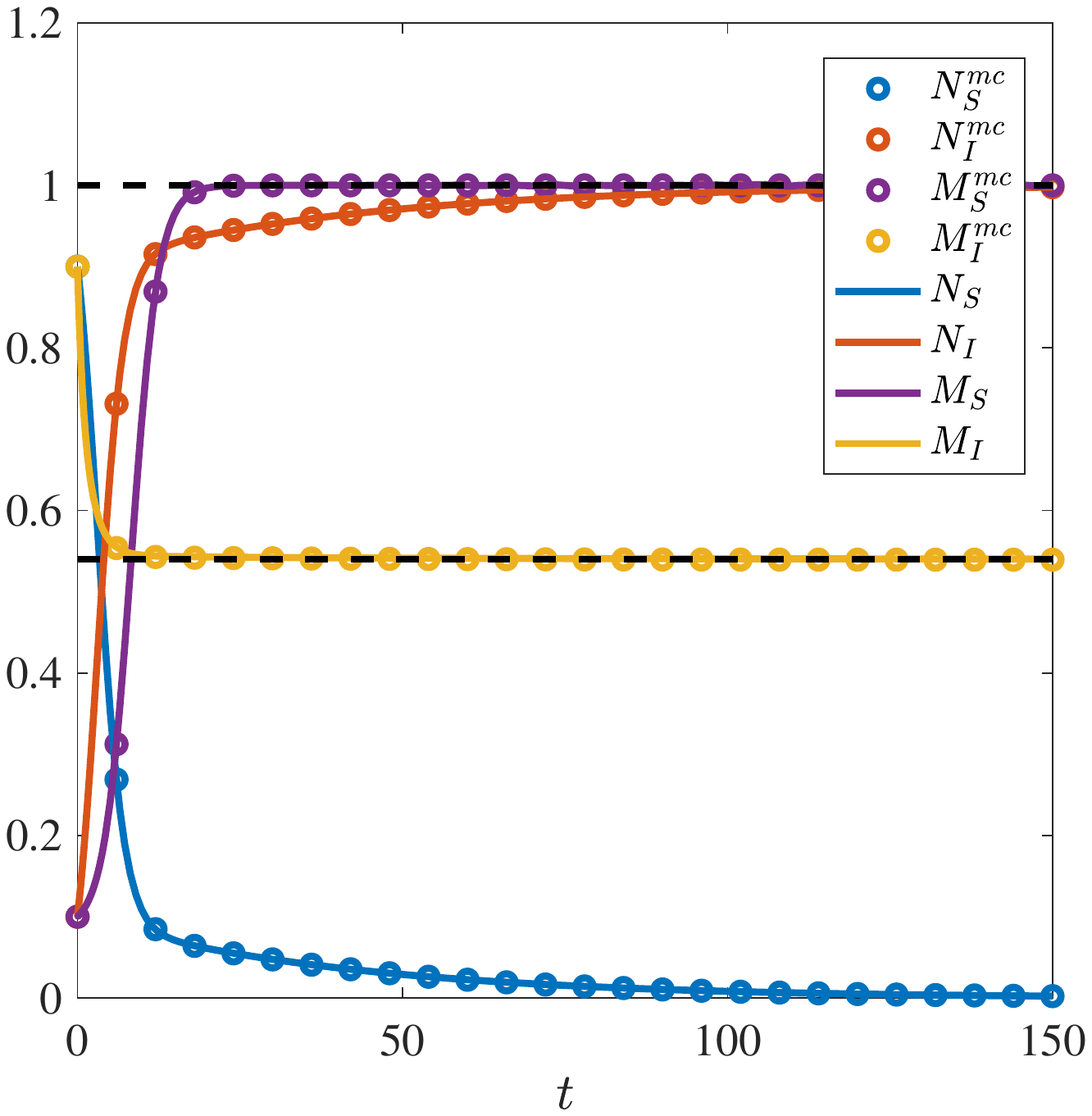}}
\subcaptionbox{\label{gauR1V1}}
{\includegraphics[trim={0cm 7cm 0 7cm},clip, width=.49\textwidth]{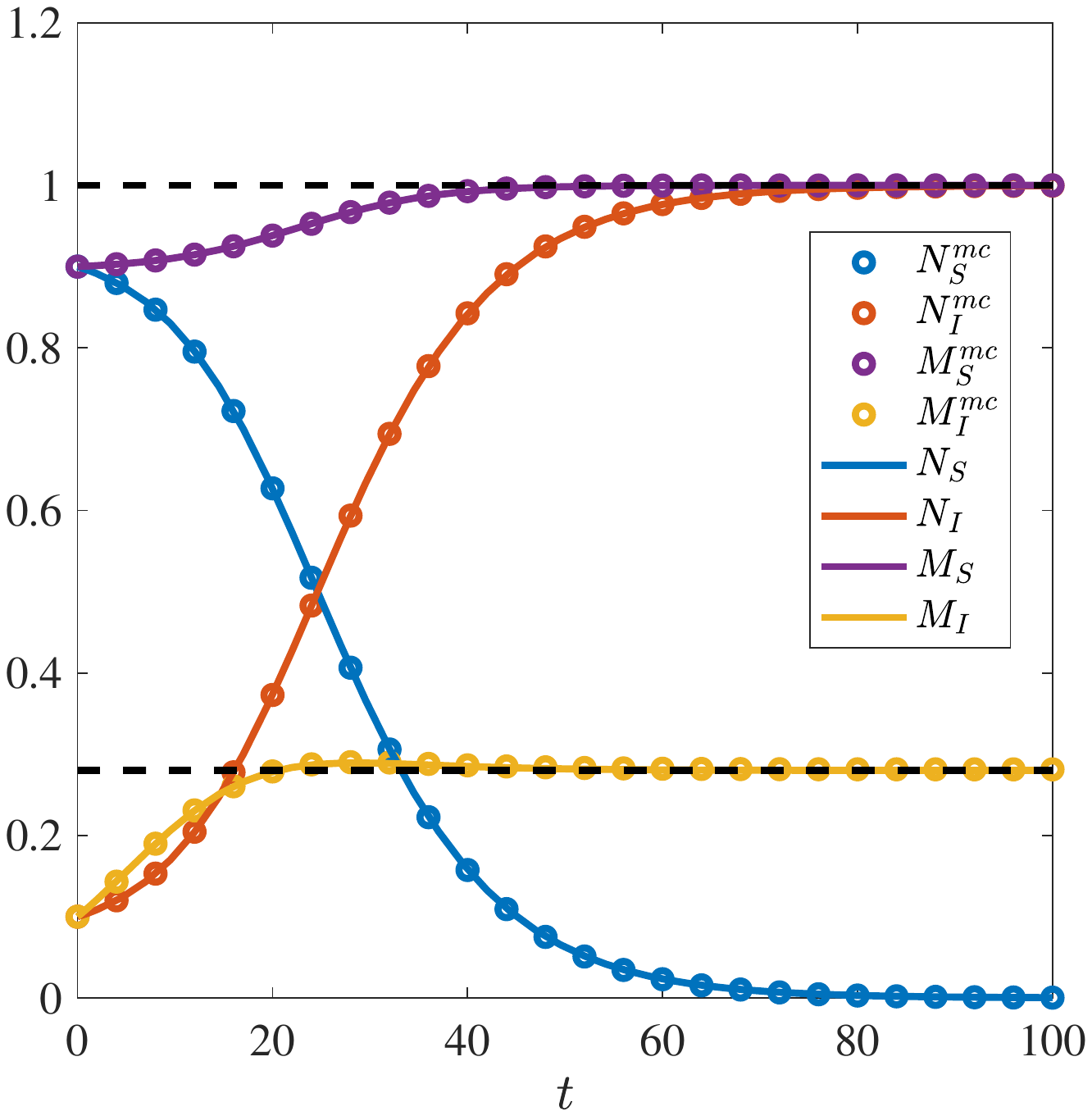}}
\subcaptionbox{\label{gauR1V9}} 
{\includegraphics[trim={0cm 7cm 0 7cm},clip,width=.49\textwidth]{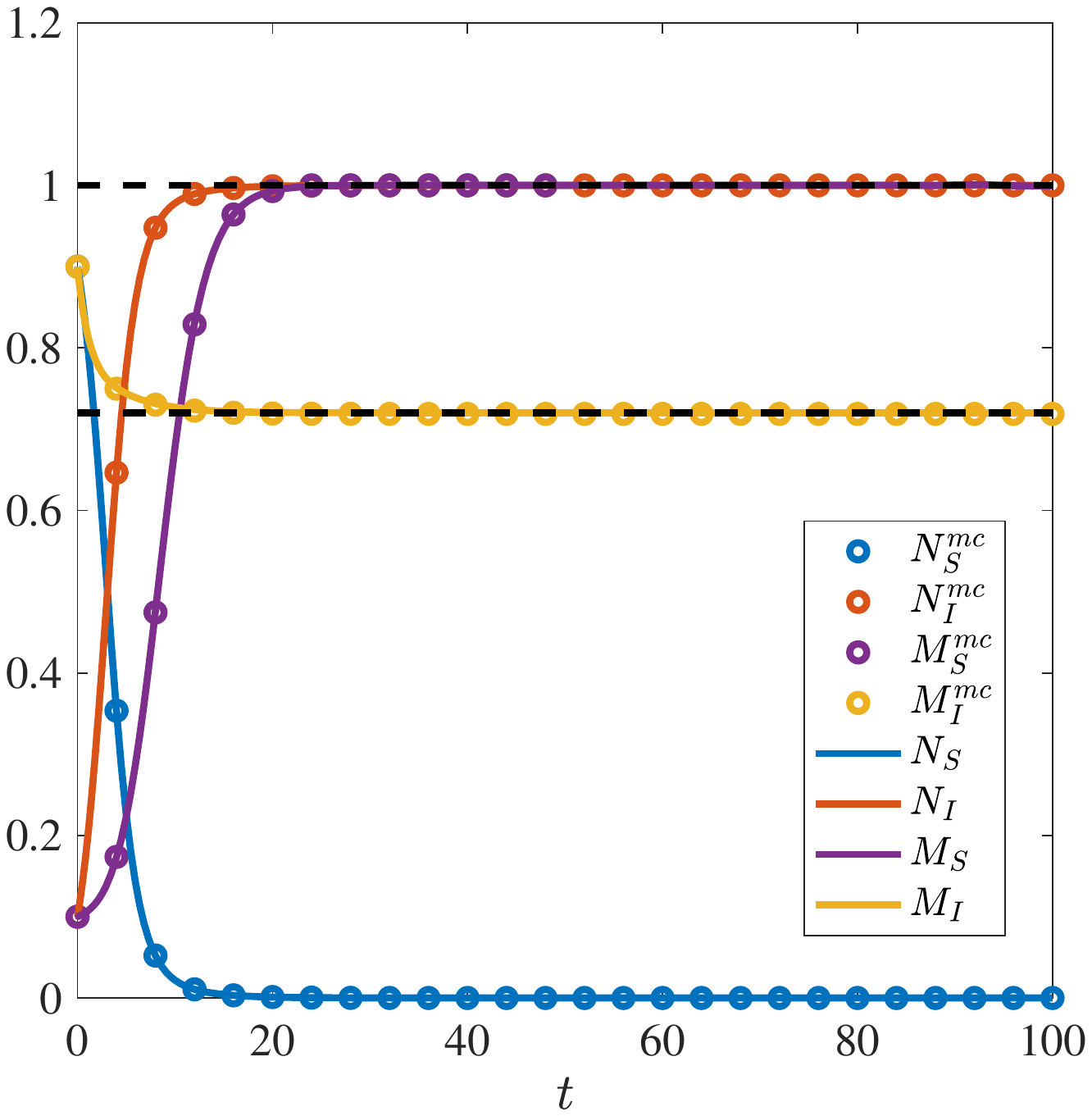}}
\caption{{\bf Numerical results under the assumptions of Proposition~\ref{prop1}.} Dynamics of $N_S(t)$ (blue lines), $N_I(t)$ (orange lines), $M_S(t)$ (purple lines) and $M_I(t)$ (yellow lines) obtained by solving numerically the macroscopic model defined via the ODE system~\eqref{eq:ODEsNSNIMSMI1} subject to initial conditions~\eqref{def:ICODE} with either $M_{S,0}=0.9$ and $M_{I,0}=0.1$ (panels (a) and (c)) or $M_{S,0}=0.1$ and $M_{I,0}=0.9$ (panels (b) and (d)). The dynamics of the corresponding quantities obtained through Monte Carlo simulations of the individual-based model are highlighted by circular markers, while the black lines highlight the asymptotic values given by~\eqref{eq:asylim1a} and~\eqref{eq:asylim1b}. These results are obtained in the case where the function $q$ is defined via~\eqref{def:qProp12} with the values of the parameters $\kappa_1$ and $\kappa_2$ such that the assumptions of Proposition~\ref{prop1} are satisfied -- i.e. $\kappa_1=1.1$ and $\kappa_2=1.1$ (panels (a) and (b)) or $\kappa_1=2.2$ and $\kappa_2=2.3$ (panels (c) and (d)) --, and the kernel $Q$ is either defined as a uniform probability distribution via~\eqref{def:Qunif} (panels (a) and (b)) or as a truncated normal probability distribution via~\eqref{def:Qtrunc}-\eqref{def:phimusigma} (panels (c) and (d)).}
\label{fig:prop1}
\end{figure}

\begin{figure}[H]
\centering
\subcaptionbox{\label{UniProp2R9V1}}
{\includegraphics[trim={0 7cm 0 7cm},clip, width=.49\textwidth]{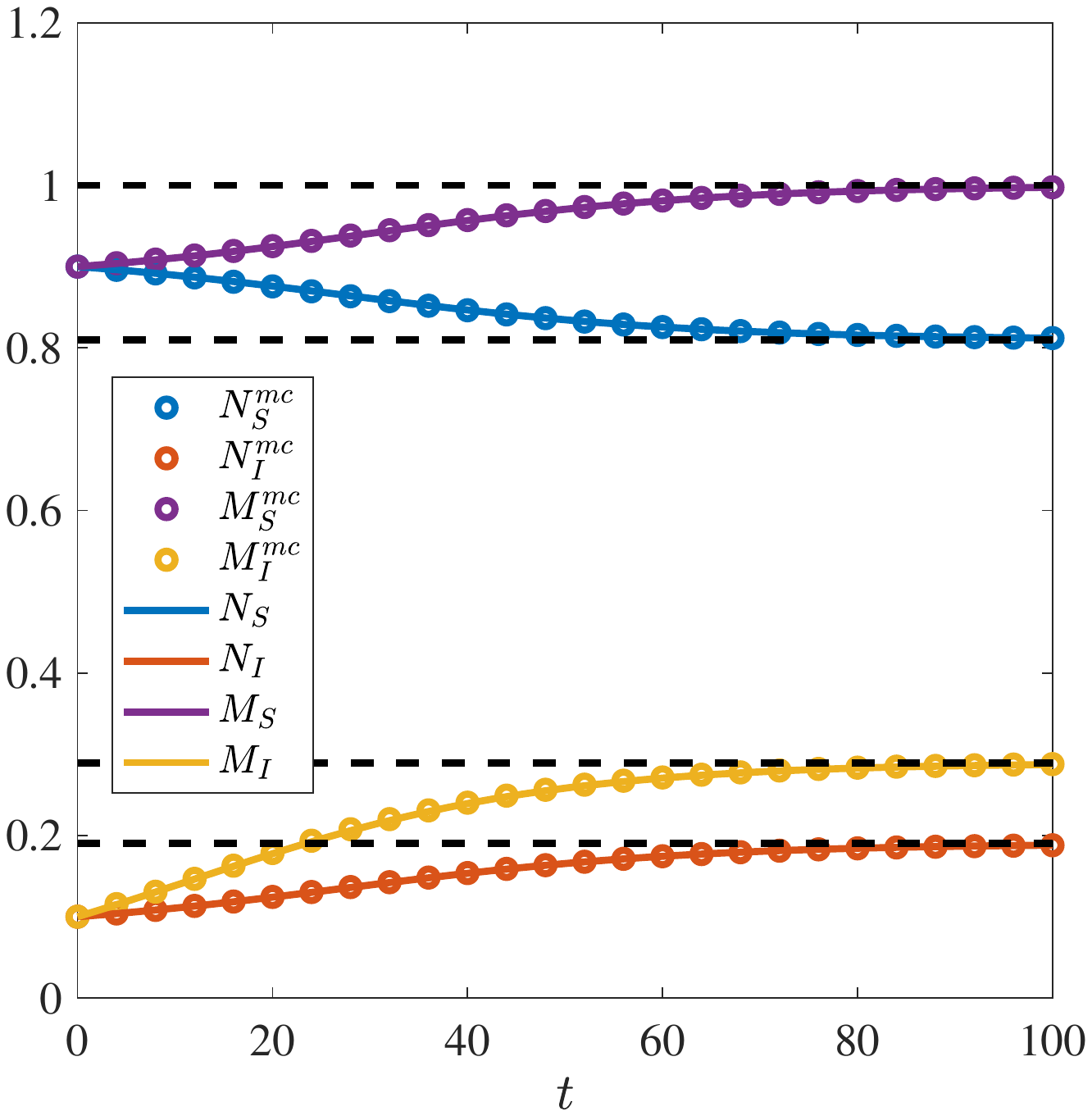}}
\subcaptionbox{\label{UniProp2R1V9}} 
{\includegraphics[trim={0 7cm 0 7cm},clip, width=.49\textwidth]{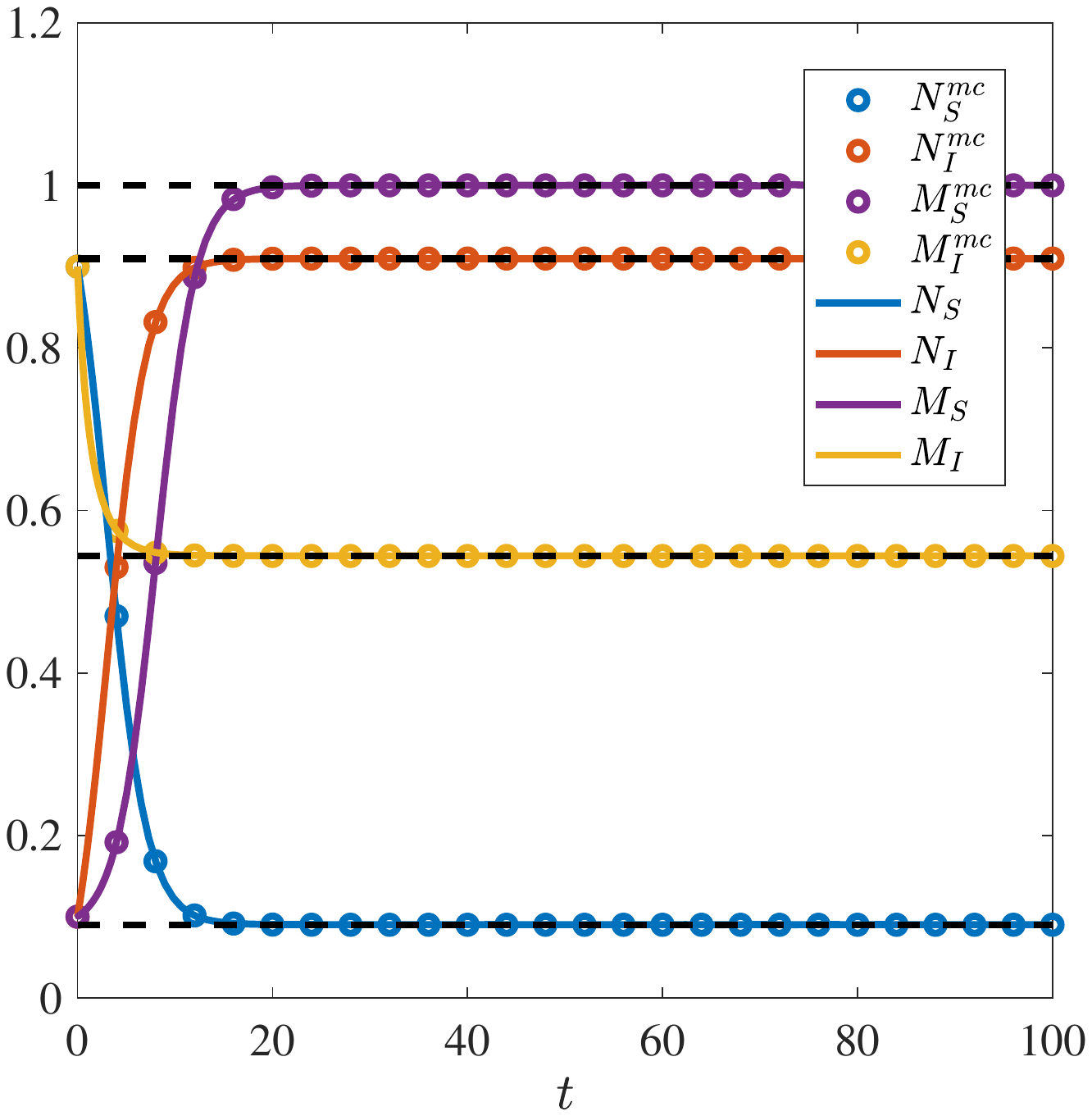}}
\subcaptionbox{\label{prop2R9V1}}
{\includegraphics[trim={0 7cm 0 7cm},clip, width=.49\textwidth]{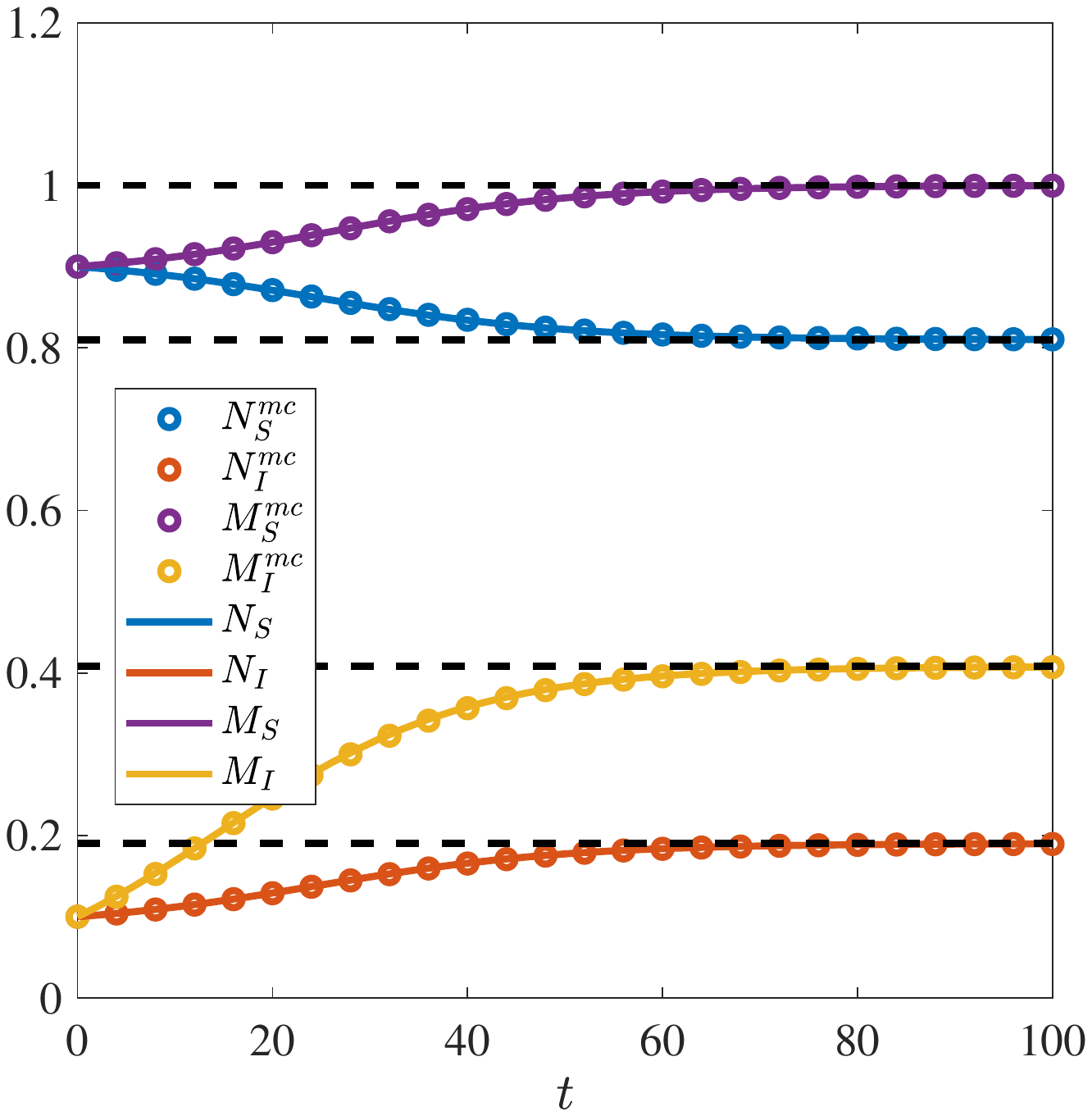}}
\subcaptionbox{\label{prop2R1V9}} 
{\includegraphics[trim={0 7cm 0 7cm},clip, width=.49\textwidth]{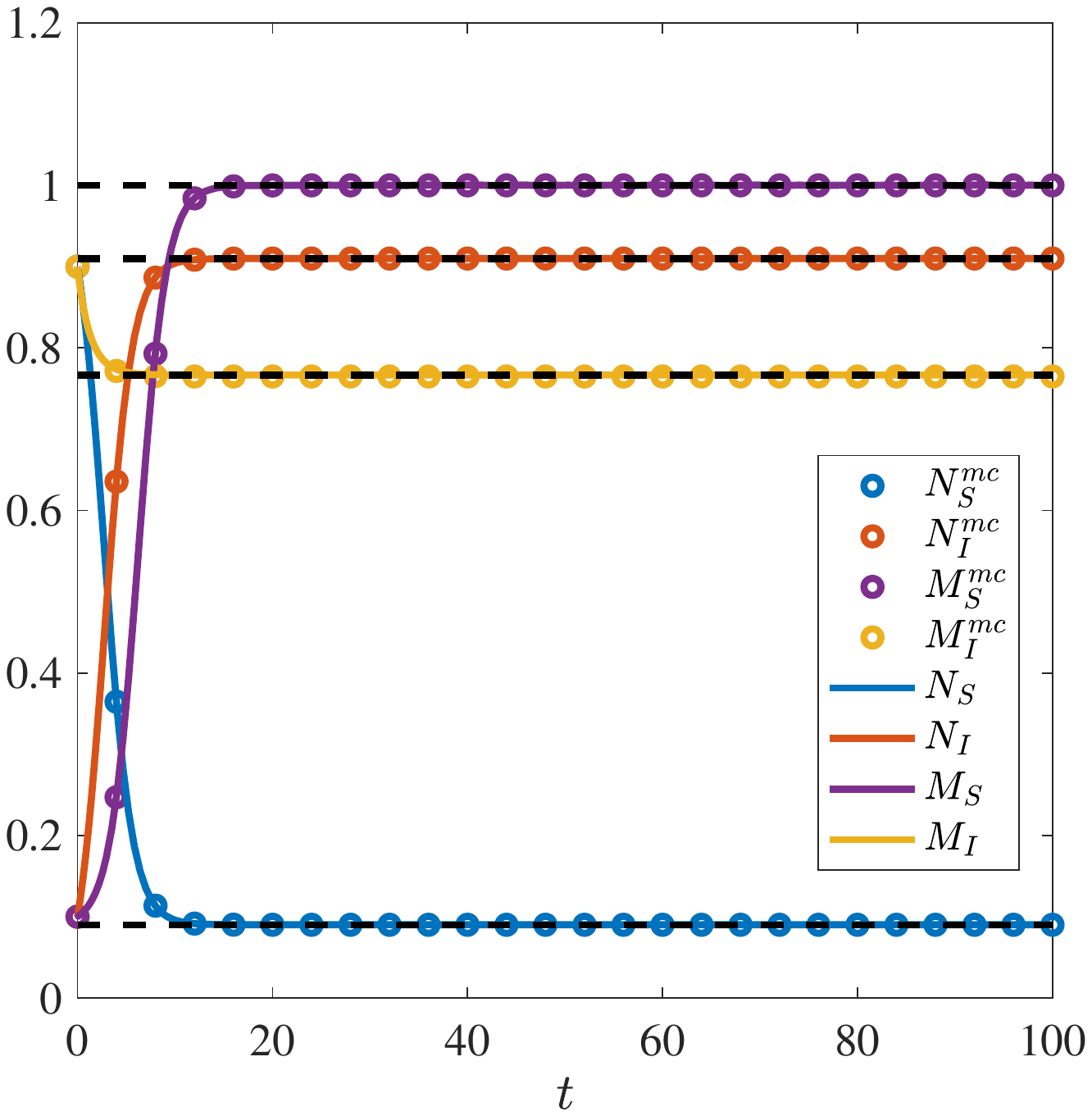}}
\caption{{\bf Numerical results under the assumptions of Proposition~\ref{prop2}.} Dynamics of $N_S(t)$ (blue lines), $N_I(t)$ (orange lines), $M_S(t)$ (purple lines) and $M_I(t)$ (yellow lines) obtained by solving numerically the macroscopic model defined via the ODE system~\eqref{eq:ODEsNSNIMSMI1} subject to initial conditions~\eqref{def:ICODE} with either $M_{S,0}=0.9$ and $M_{I,0}=0.1$ (panels (a) and (c)) or $M_{S,0}=0.1$ and $M_{I,0}=0.9$ (panels (b) and (d)). The dynamics of the corresponding quantities obtained through Monte Carlo simulations of the individual-based model are highlighted by circular markers, while the black lines highlight the asymptotic values given by~\eqref{eq:asylim2a} and~\eqref{eq:asylim2b}. These results are obtained in the case where the function $q$ is defined via~\eqref{def:qProp12} with $\kappa_1=1$ and $\kappa_2=1$ so that the assumptions of Proposition~\ref{prop2} are satisfied, and the kernel $Q$ is either defined as a uniform probability distribution via~\eqref{def:Qunif} (panels (a) and (b)) or as a truncated normal probability distribution via~\eqref{def:Qtrunc}-\eqref{def:phimusigma} (panels (c) and (d)).}
\label{fig:prop2}
\end{figure}

\begin{figure}[H]
\centering
\subcaptionbox{\label{UniProp3R9V1}}
{\includegraphics[trim={0 7cm 0 7cm},clip, width=.49\textwidth]{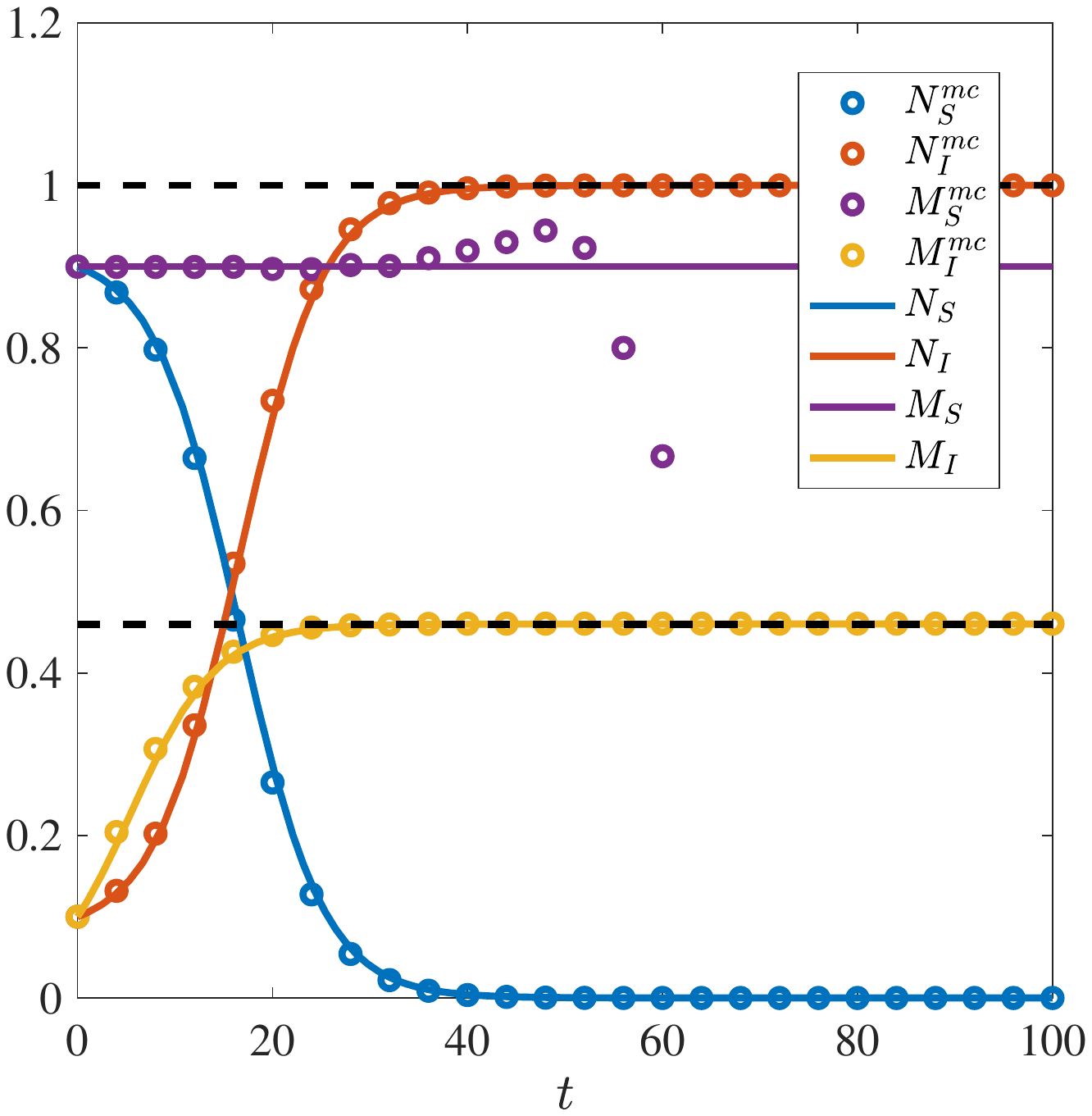}}
\subcaptionbox{\label{UniProp3R1V9}} 
{\includegraphics[trim={0 7cm 0 7cm},clip, width=.49\textwidth]{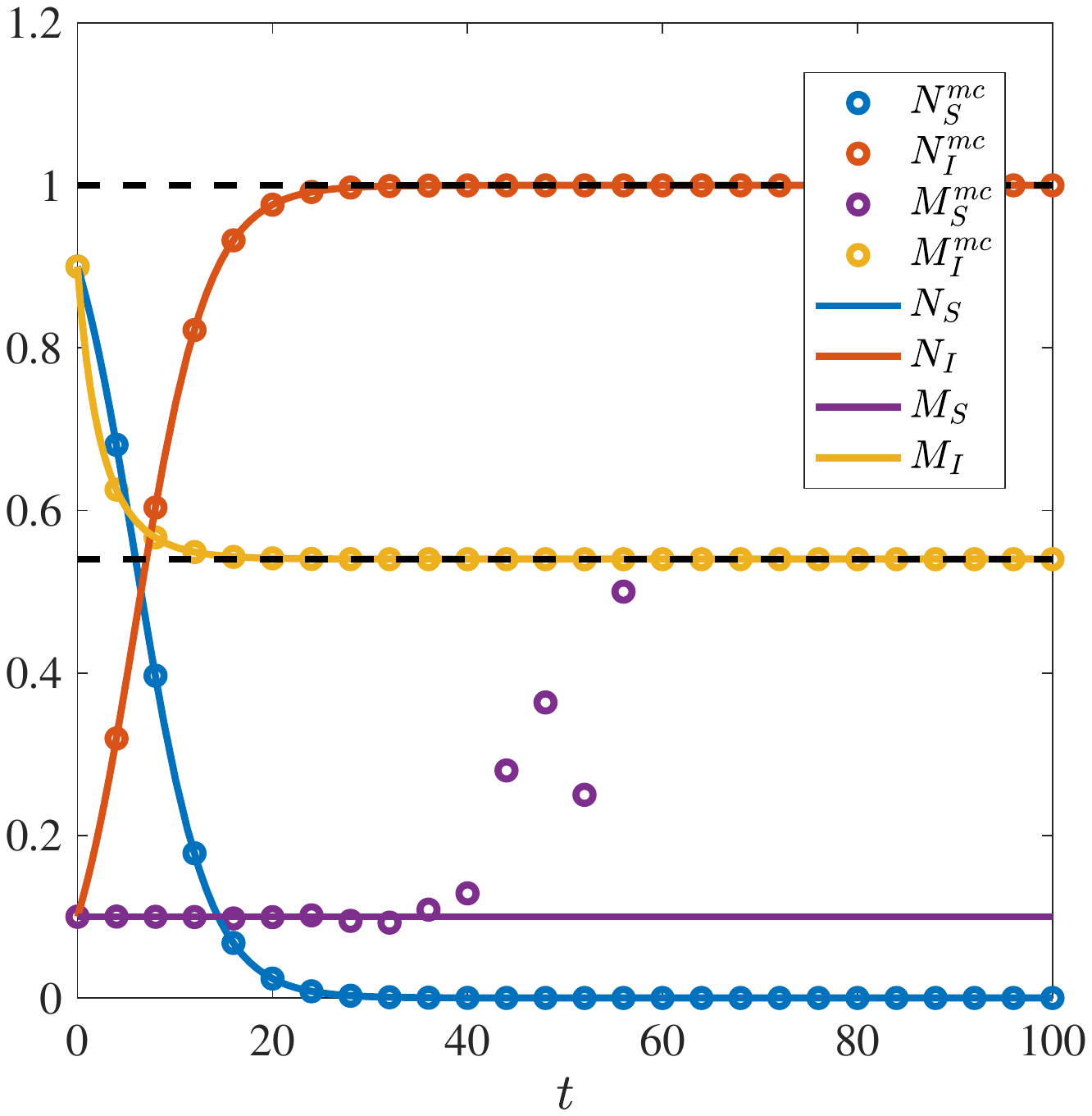}}
\subcaptionbox{\label{prop3R9V1}}
{\includegraphics[trim={0 7cm 0 7cm},clip, width=.49\textwidth]{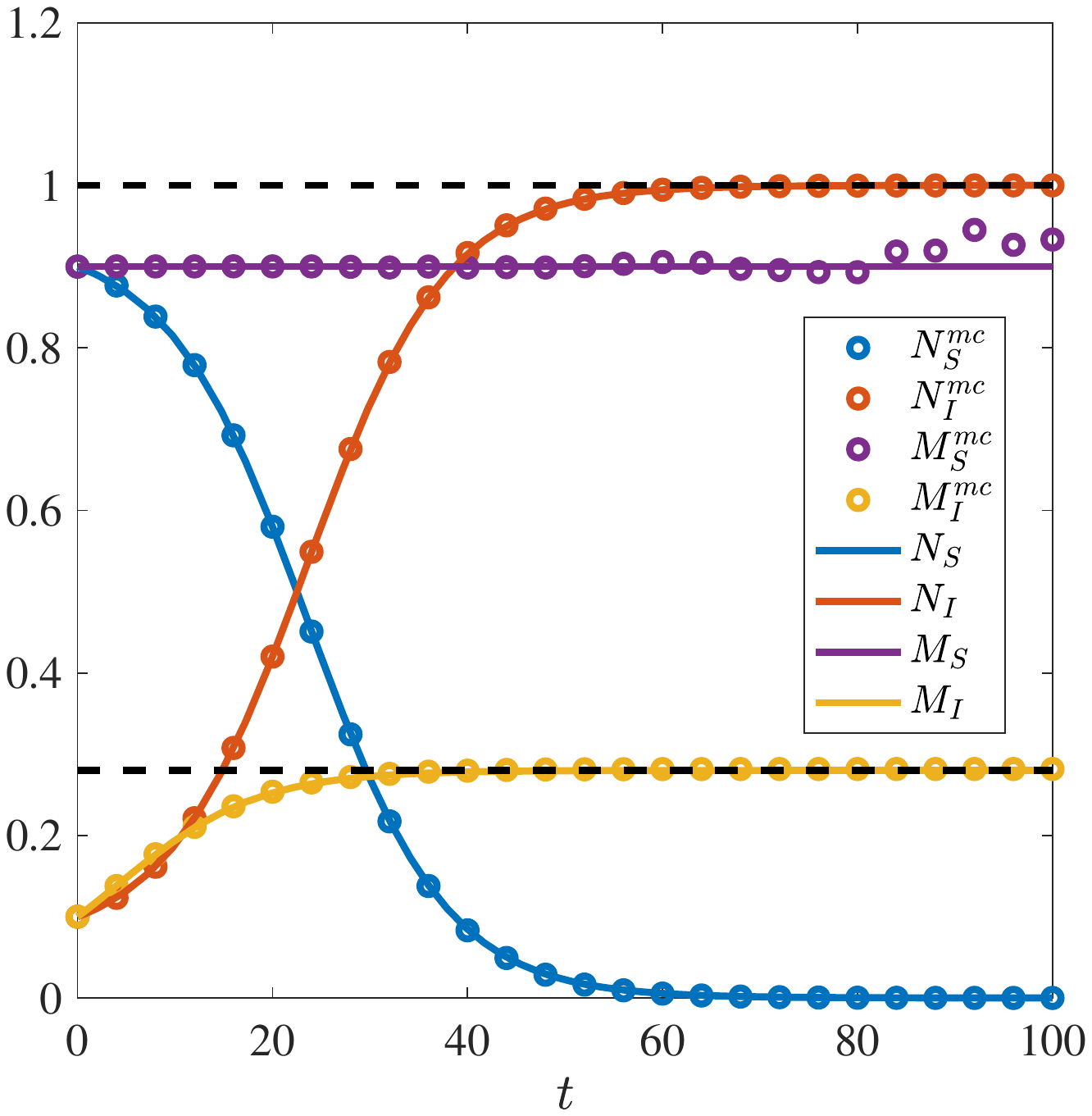}}
\subcaptionbox{\label{prop3R1V9}} 
{\includegraphics[trim={0 7cm 0 7cm},clip, width=.49\textwidth]{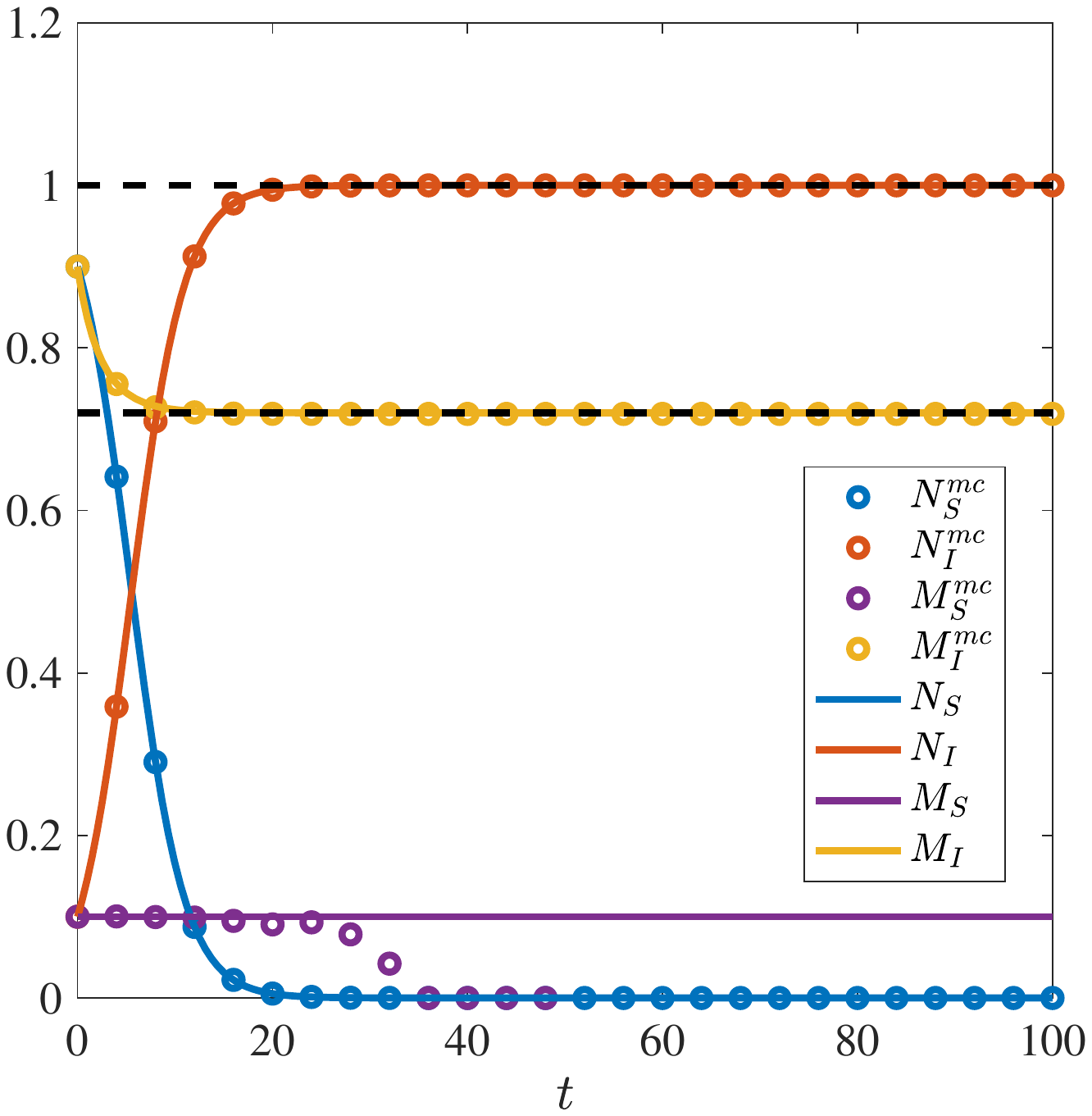}}
\caption{{\bf Numerical results under the assumptions of Proposition~\ref{prop3}.} Dynamics of $N_S(t)$ (blue lines), $N_I(t)$ (orange lines), $M_S(t)$ (purple lines) and $M_I(t)$ (yellow lines) obtained by solving numerically the macroscopic model defined via the ODE system~\eqref{eq:ODEsNSNIMSMI1} subject to initial conditions~\eqref{def:ICODE} with either $M_{S,0}=0.9$ and $M_{I,0}=0.1$ (panels (a) and (c)) or $M_{S,0}=0.1$ and $M_{I,0}=0.9$ (panels (b) and (d)). The dynamics of the corresponding quantities obtained through Monte Carlo simulations of the individual-based model are highlighted by circular markers, while the black lines highlight the asymptotic values given by~\eqref{eq:asylim3a} and~\eqref{eq:asylim3b}. These results are obtained in the case where the function $q$ is defined via~\eqref{def:qProp3} with $\kappa_3=0.5$ so that the assumptions of Proposition~\ref{prop3} are satisfied, and the kernel $Q$ is either defined as a uniform probability distribution via~\eqref{def:Qunif} (panels (a) and (b)) or as a truncated normal probability distribution via~\eqref{def:Qtrunc}-\eqref{def:phimusigma} (panels (c) and (d)).}
\label{fig:prop3}
\end{figure}

\newpage
\section{Discussion and research perspectives}
\label{sec:disc}
\paragraph{Summary of the main results.} 
We developed a new structured compartmental model for the coevolutionary dynamics between susceptible and infectious individuals in a heterogeneous SI epidemiological system. In this model, the susceptible compartment is structured by a continuous variable, $r \in \mathcal{R} \subset \mathbb{R}_+$, that represents the level of resistance to infection of susceptible individuals, while the infectious compartment is structured by a continuous variable, $v \in \mathcal{V} \subset \mathbb{R}_+$, that represents the viral load of infectious individuals. The model takes into account the fact that the level of resistance to infection and the viral load of the individuals may evolve in time, along with the fact that 
the probability of infection of susceptible individuals depends both on their level of resistance to infection and the viral load of the infectious individuals they come in contact with. 

We first formulated a stochastic individual-based model that tracks the dynamics of single individuals, from which we formally derived the corresponding mesoscopic model, which consists of the IDE system~\eqref{eq:model} for the population density functions of susceptible and infectious individuals, $n_S(t,r)$ and $n_I(t,v)$. We then considered an appropriately rescaled version of this model, which is given by the IDE system~\eqref{eq:modeleps}, and we carried out formal asymptotic analysis to derive the corresponding macroscopic model, which comprises the ODE system~\eqref{eq:ODEsNSNIMSMIgen} for the proportions of susceptible and infectious individuals, $N_S(t)$ and $N_I(t)$, the mean level of resistance to infection of susceptible individuals, $M_S(t)$, and the mean viral load of infectious individuals, $M_I(t)$.

We established well-posedness of the mesoscopic model (see Theorem~\ref{theo:wpide}) and the macroscopic model~\eqref{eq:ODEsNSNIMSMIgen} (see Theorem~\ref{theo:wpode}), and we studied the long-time behaviour of the components of the solution to the macroscopic model (see Propositions~\ref{prop1}-\ref{prop3}). Moreover, we presented a sample of numerical results (see Figures~\ref{fig:prop1}-\ref{fig:prop3}) which show that, in the framework of assumptions~\eqref{ass:Krv}-\eqref{ass:thetaseps} (i.e. the assumptions under which the macroscopic model is formally obtained from the mesoscopic model), there is an excellent quantitative agreement between the results of Monte Carlo simulations of the individual-based model, numerical solutions of the macroscopic model, and the analytical results established by Propositions~\ref{prop1}-\ref{prop3}. This validates the formal limiting procedure employed to obtain the mesoscopic model alongside the formal asymptotic analysis carried out to derive the macroscopic model.

\paragraph{Summary of biological insights provided by the main results.} Under biological scenarios corresponding to assumptions~\eqref{ass:Krv}-\eqref{ass:thetaseps} (see Remark~\ref{rem:Krv} and Remark~\ref{rem:theta}), the asymptotic results established by Propositions~\ref{prop1}-\ref{prop3} shed light on the way in which the probability of infection, $q$, affects the co-evolutionary dynamics between susceptible and infectious individuals. These results also demonstrate how the long-term behaviour of the heterogeneous SI epidemiological system considered here depends on:
\begin{enumerate*}[label=(\roman*)]
\item the mean value of the viral load acquired by a susceptible individual upon infection, $\overline Q$;
\item the initial proportions of susceptible and infectious individuals, $N_{S,0}$ and $N_{I,0}$;
\item the initial values of the mean level of resistance to infection of susceptible individuals and the mean viral load of infectious individuals, $M_{S,0}$ and $M_{I,0}$.
\end{enumerate*}
In summary:
\begin{itemize}
\item[(i)] If $q(0,1) \geq q(1,1) > 0$ then all individuals eventually become infectious. Moreover, the equilibrium value of the mean viral load increases with $N_{S,0}$, $N_{I,0}$, $M_{I,0}$, and $\overline Q(0,1)$, and increases or decreases with $M_{S,0}$ depending on the fact that $\overline{Q}(1,1)>\overline{Q}(0,1)$ or $\overline{Q}(1,1)<\overline{Q}(0,1)$, respectively.
\item[(ii)] If $q(0,1) > 0$ and $q(1,1) = 0$ then the mean level of resistance to infection eventually converges to the maximum value $1$ and an endemic equilibrium is attained. The fraction of susceptible individuals at equilibrium is proportional to $N_{S,0}$ with constant of proportionality $M_{S,0}$ (i.e. the larger the value of $M_{S,0}$ the smaller the proportion of infectious individuals at equilibrium). Moreover, the equilibrium value of the mean viral load of infectious individuals increases with $\overline{Q}(0,1)$ and $M_{I,0}$,  decreases with $M_{S,0}$, and increases or decreases with $N_{S,0}$ depending on the fact that $M_{I,0} >\overline{Q}(0,1)$ or $M_{I,0} <\overline{Q}(0,1)$, respectively.
\end{itemize}

\paragraph{Research perspectives.} We conclude with an outlook on possible research perspectives. First of all, while the current study has eschewed specific mechanisms driving changes in the level of resistance to infection of susceptible individuals and the viral load of infectious individuals, it would be relevant to explore how considering different forms of these mechanisms may result in different assumptions on the kernels $K_S$ and $K_I$ and, in turn, how such different assumptions may lead to different macroscopic models. In this regard, another avenue for future research would be explore the possibility to estimate the forms of $K_S$ and $K_I$ based on data; for this, techniques similar to those employed in~\cite{albi2022kinetic,albi2021control,dimarco2021kinetic} may prove useful. Moreover, although in this work we focused on heterogenous SI systems, the individual-based modelling approach presented here, and the formal methods to derive the corresponding mesoscopic and macroscopic models, could easily be extended to heterogenous SIS and SIR epidemiological systems. Furthermore, while we did not incorporate into the model the effect of measures for preventing and containing outbreak of infection, it would certainly be interesting to generalise the underlying individual-based modelling approach, as well as the limiting procedures to formally derive the mesoscopic and macroscopic counterparts of the model, to cases where pharmaceutical and non-pharmaceutical interventions to control the spread of infection are taken into account. In this regard, another track to follow might be to address optimal control of the macroscopic model to identify optimal strategies to prevent and contain outbreak of infection and then investigate whether such strategies would remain optimal also for the corresponding individual-based model. As a further generalisation, in the vein of~\cite{loy2021viral}, we could also consider the case where individuals are distributed over a network whereby each node would correspond to a different spatial region, in order to assess how the spread of infectious diseases is shaped by the interplay between phenotypic and spatial heterogeneities across scales. 

\section*{Appendix}

\appendix

\section{Proof of Theorem~\ref{theo:wpide}}
\label{sec:appA}
\paragraph{Notation.} Let $x = (r,v) \in \mathcal{R} \times \mathcal{V} =: \mathcal{X}$ and $n = (n_S, n_I)$, and consider the following Banach spaces 
$$
(A, \| \cdot \|_{A}), \quad A := \left\{n = (n_S, n_I) : n_1 \in L^1(\mathcal{R}), n_2 \in L^1(\mathcal{V}) \right\}, \quad \| n(t,\cdot) \|_{A} = \| n_S(t,\cdot) \|_{L^1(\mathcal{R})} + \| n_I(t,\cdot) \|_{L^1(\mathcal{V})}
$$
$$
(B, \| \cdot \|_{B}), \quad B := C([0,T], A), \quad  \| n \|_{B} = \sup_{t \in [0,T]} \| n(t,\cdot) \|_{A},
$$
with $T \in \mathbb{R}_+^*$. Let $F[n](t,x) = (F_S[n](t,x), F_I[n](t,x))$, with $F_S$ and $F_I$ being defined by the right-hand sides of the IDEs~\eqref{eq:model}, i.e.
$$
F_S[n](t,x) =  -\theta_X \, n_S \, \int_{\mathcal{V}} q(r, v^*) \, n_I(t,v^*) \di v^* + \theta_R \int_{\mathcal{R}} K_S(r | r') \, n_S(t,r') \di r' - \theta_R \, n_S
$$
and
$$
F_I[n](t,x) = \theta_X \, \int_{\mathcal{R}} \int_{\mathcal{V}} q(r,v^*) \, Q(v | r, v^*) \, n_S(t,r) \, n_I(t,v^*) \di v^* \di r + \theta_V \int_{\mathcal{V}} K_I(v | v') \, n_I(t,v') \di v' - \theta_V \, n_I.
$$

\paragraph{Preliminaries.} Considering $t \in [0,T]$ with $T \in \mathbb{R}^*_+$, we rewrite the Cauchy problem~\eqref{eq:model}-\eqref{eq:ICb} as
\beq \label{eq:modelvect}
\begin{cases}
\pa_t n(t,x) = F[n](t,x),
\\\\
n(0,x) = n_0(x) := (n_{S,0}, n_{I,0}),
\end{cases}
\quad
(t,x) \in (0,T] \times \mathcal{X}
\eeq
and we also note that
\beq
\label{eq:defP}
n(t,x) = P[n](t,x) \;\; \text{ with } \;\; P[n](t,x):= n(0,x) + \int_0^t  F[n](s,x) \di s. 
\eeq
\paragraph{Step 1: Local well-posedness.} We begin by proving the following lemma:
\begin{lem}
\label{lem:Pprop}
Let assumptions~\eqref{ass:KSa}, \eqref{ass:KIa}, \eqref{ass:Qa}, \eqref{ass:q}, and~\eqref{eq:Theta} hold. Then, for any $n, m \in B$ with $n(0,x)=m(0,x) = n_0(x)$, the operator $P$ defined via~\eqref{eq:defP} satisfies the following estimates
\beq
\label{eq:estP1}
\|P[n]\|_{B} \leq \|n_0\|_{A} + 2 \, T \Big(\max(\theta_R,\theta_V) + \theta_X \|n\|_{B} \Big) \|n\|_{B},
\eeq  
\beq
\label{eq:estP2}
\|P[n] - P[m]\|_{B} \leq 2 \, T \Big[\max(\theta_R,\theta_V) + \theta_X \Big(\|n\|_{B} + \|m\|_{B}\Big) \Big] \|n - m\|_{B}.
\eeq
\end{lem}

\begin{proof}
The estimate~\eqref{eq:estP2} is a direct consequence of the estimate~\eqref{eq:estP1}. Therefore, we provide the proof of the estimate~\eqref{eq:estP1} only. Under assumptions~\eqref{ass:KSa}, \eqref{ass:KIa}, \eqref{ass:Qa}, \eqref{ass:q}, and~\eqref{eq:Theta}, we have
$$
\|F_S[n]\|_{L^1(\mathcal{R})} \leq \theta_X \|n_S\|_{L^1(\mathcal{R})} \|n_I\|_{L^1(\mathcal{V})} + 2 \theta_R \|n_S\|_{L^1(\mathcal{R})} 
$$
and
$$
\|F_I[n]\|_{L^1(\mathcal{V})} \leq \theta_X \|n_S\|_{L^1(\mathcal{R})} \|n_I\|_{L^1(\mathcal{V})} + 2 \theta_V \|n_I\|_{L^1(\mathcal{V})}.
$$
Hence,
\beqa
\|F[n]\|_A &=& \|F_S[n]\|_{L^1(\mathcal{R})} + \|F_I[n]\|_{L^1(\mathcal{V})} \nonumber
\\
&\leq& 2 \Big(\theta_X \|n_S\|_{L^1(\mathcal{R})} \|n_I\|_{L^1(\mathcal{V})} + \theta_R \|n_S\|_{L^1(\mathcal{R})} + \theta_V \|n_I\|_{L^1(\mathcal{V})}\Big) \nonumber
\\
&\leq& 2 \Big(\theta_X \|n\|_{A}^2 + \max(\theta_R,\theta_V)  \|n\|_{A}\Big). \nonumber
\eeqa
Recalling the definition for $P$ given by~\eqref{eq:defP}, the above estimate allows us to conclude that
$$
\|P[n]\|_A \leq \|n_0\|_{A} + 2 \int_0^t \Big(\max(\theta_R,\theta_V) + \theta_X \|n(\cdot, s)\|_{A}\Big) \|n(\cdot, s)\|_{A} \, \di s,
$$
from which, recalling that $t \in [0,T]$, the estimate~\eqref{eq:estP1} can easily be obtained.
\end{proof}
The estimate~\eqref{eq:estP1} ensures that $P$ maps $B$ into itself. Moreover, the estimate~\eqref{eq:estP2} ensures that there exists $T^*>0$ such that if $T<T^*$ then $P$ is a contraction on $B$. Hence, the Banach fixed point theorem allows us to conclude that if $T<T^*$ then $P: B \to B$ admits a unique fixed point $n \in B$ and, therefore (cf. the relation given by~\eqref{eq:defP}), the Cauchy problem~\eqref{eq:modelvect} (i.e. the Cauchy problem~\eqref{eq:model}-\eqref{eq:ICb}) admits a unique solution of components $n_S \in C([0,T]; L^1(\mathcal{R}))$ and $n_I \in C([0,T]; L^1(\mathcal{V}))$.

\paragraph{Step 2: Non-negativity of~\texorpdfstring{$\boldsymbol{n_S}$}{} and~\texorpdfstring{$\boldsymbol{n_I}$}{}.} Solving the IDE~\eqref{eq:model}$_1$ for $n_S$ subject to the initial condition $n_{S,0}$ yields the semi-explicit formula
\begin{align*}
	\resizebox{\textwidth}{!}{$
	\begin{aligned}
		n_S(t,r) &= n_{S,0}(r)\exp\left[-\int_0^t\left(\theta_X\,\int_{\mathcal{V}}q(r,v^*)\,n_I(s,v^*)\di v^*+\theta_R\right)\di s\right] \\
		&\phantom{=} +\theta_R\,\int_0^t\left(\int_{\mathcal{R}}K_S(r|r')\,n_S(s,r')\di r'\right)\exp\left[-\int_s^t\left(\theta_X\,\int_{\mathcal{V}}q(r,v^*)\,n_I(\tau,v^*)\di v^*
			+\theta_R\right)\di\tau\right]\di s.
	\end{aligned}
	$}
\end{align*}
Since the kernel $K_S$ and the functions $n_{S,0}$ and $q$ are non-negative (cf. assumptions~\eqref{ass:KSa}, \eqref{ass:q} and \eqref{eq:ICa}) and the parameter $\theta_R$ is positive (cf. assumption~\eqref{eq:Theta}), the above formula implies that $n_S$ is non-negative. Similarly, solving the IDE~\eqref{eq:model}$_2$ for $n_I$ subject to the initial condition $n_{I,0}$ we obtain the semi-explicit formula
\begin{align*}
	\begin{aligned}
		n_I(t,v) &= n_{I,0}(v)\exp\left(-\theta_V\,t\right) \\
		&\phantom{=} +\theta_X\int_0^t\left(\int_{\mathcal{R}}\int_{\mathcal{V}}q(r,v^*)\,Q(v|r,v^*)\,n_S(s,r)\,n_I(s,v^*)\di v^*\di r\right)\,\exp\big[-\theta_V\,(t-s)\big]\di s \\
		&\phantom{=} +\theta_V\int_0^t\left(\int_{\mathcal{V}}K_I(v|v')\,n_I(s,v')\di v'\right)\,\exp\big[-\theta_V\,(t-s)\big]\di s,
	\end{aligned}
\end{align*}
from which we conclude that $n_I$ is non-negative due to the fact that the kernels $K_I$ and $Q$ and the functions $n_{I,0}$ and $q$ are non-negative (cf. assumptions~\eqref{ass:KIa}, \eqref{ass:Qa}, \eqref{ass:q}, and \eqref{eq:ICa}), the parameters $\theta_X$ and $\theta_V$ are positive (cf. assumptions~\eqref{eq:Theta}), and, as we have just proved, $n_S$ is non-negative.

\paragraph{Step 3: Uniform bounds on~\texorpdfstring{$\boldsymbol{\|n(t,\cdot)\|_{A}}$}{} and~\texorpdfstring{$\boldsymbol{\|n\|_{B}}$}{}.} Adding together the differential equations obtained by integrating the IDE~\eqref{eq:model} for $n_S(t,r)$ over $\mathcal{R}$ and the IDE~\eqref{eq:model} for $n_I(t,v)$ over $\mathcal{V}$ and using the non-negativity of $n_S$ and $n_I$, one can easily prove that, under assumptions~~\eqref{ass:KSa}, \eqref{ass:KIa}, \eqref{ass:Qa}, \eqref{ass:q}, the solution to the Cauchy problem~\eqref{eq:modelvect} is such that
\beq
\label{eq:estim}
\| n(t,\cdot) \|_{A} = \| n(\cdot, 0) \|_{A} \; \forall t \in \mathbb{R}_+, \quad  \| n \|_{B} = \| n(\cdot, 0) \|_{A}
\eeq
and, in particular, the components of the solution to the Cauchy problem~\eqref{eq:model}-\eqref{eq:ICb} are such that
$$
\| n_S(t,\cdot) \|_{L^1(\mathcal{R})} + \| n_I(\cdot,t) \|_{L^1(\mathcal{V})} = 1 \quad \forall t \geq 0.
$$

\paragraph{Step 4: Global well-posedness.} The uniform bounds~\eqref{eq:estim} allow one to iterate the process used in Step~1 on consecutive time intervals of the form $[i \,T, (i + 1) \, T]$ with $i \in \mathbb{N}$ and, in so doing, prove that the Cauchy problem~\eqref{eq:modelvect} (i.e. the Cauchy problem problem~\eqref{eq:model}-\eqref{eq:ICb}) admits a unique solution of components $n_S \in C(\mathbb{R}_+; L^1(\mathcal{R}))$ and $n_I \in C(\mathbb{R}_+; L^1(\mathcal{V}))$, whose non-negativity is ensured by the results established in Step 2. This concludes the proof of Theorem~\ref{theo:wpide}. $\qed$

\section{Proof of the upper bound~\texorpdfstring{\eqref{eq:unifestNSNIMSMI}}{} on~\texorpdfstring{$\boldsymbol{M_S}$}{}}
\label{sec:appB}
As mentioned in the main body of the paper, we omit the proofs of the uniform bounds~\eqref{eq:unifestNSNIMSMI} on $N_S(t)$, $N_I(t)$ and $M_I(t)$, which can be obtained through simple estimates, while we prove here the upper bound~\eqref{eq:unifestNSNIMSMI} on $M_S(t)$. 

Using the differential equation~\eqref{eq:ODEsNSNIMSMIgen}$_3$ along with the fact that the function $M_S$ is non-negative (cf. estimates~\eqref{eq:unifestNSNIMSMI}), the functions $q$ and $N_I$ are non-negative and bounded above by $1$ (cf. assumptions~\eqref{ass:q} and the estimates~\eqref{eq:unifestNSNIMSMI}), the functions $\psi_I$ and $\psi_S$ are non-negative and such that
$$
\int_{\mathcal{V}} \psi_I(t,v^*) \di v^* = 1, \quad \int_{\mathcal{R}} \psi_S(t,r) \di r = 1, \quad \int_{\mathcal{R}} r \, \psi_S(t,r) \di r = M_S(t)
$$ 
(cf. the relations given by~\eqref{ass:KSKIaddMSMI}) and $\mathcal{R} \subset \mathbb{R}_+$, we find that, in general, 
\begin{eqnarray}
\ddt M_S(t) &=& N_I(t) \, \int_{\mathcal{V}} \left(\int_{\mathcal{R}} \left(M_S(t) - r\right) \, q(r, v^*) \, \psi_S(t,r) \di r \right) \, \psi_I(t,v^*) \di v^* \nonumber
\\
&\leq& N_I(t) \, M_S(t), \nonumber
\\
&\leq& M_S(t),
\label{eq:diffineq1}
\end{eqnarray}
and, in particular, if $M_S(t) \geq \sup \mathcal{R}$ then
\begin{eqnarray}
\ddt M_S(t) &=& N_I(t) \, \int_{\mathcal{V}} \left(\int_{\mathcal{R}} \left(M_S(t) - r\right) \, q(r, v^*) \, \psi_S(t,r) \di r \right) \, \psi_I(t,v^*) \di v^* \nonumber
\\
&\leq& N_I(t) \, \int_{\mathcal{V}} \left(\int_{\mathcal{R}} \left(M_S(t) - r\right) \, \psi_S(t,r) \di r \right) \, \psi_I(t,v^*) \di v^* \nonumber
\\
&\leq& M_S(t) - M_S(t) = 0.
\label{eq:diffineq2}
\end{eqnarray}
Since $M_S(0) \leq \sup \mathcal{R}$ (cf. assumptions~\eqref{eq:ODEsNSNIMSMIICgen}), the upper bound~\eqref{eq:unifestNSNIMSMI} on $M_S(t)$ follows from the differential inequalities~\eqref{eq:diffineq1} and~\eqref{eq:diffineq2}. $\qed$

\section{Proof of Lemma~\ref{lemma:psiSpsiI}}
\label{sec:appC}
We provide the proof for the result established by Lemma~\ref{lemma:psiSpsiI} on $\psi_S$ only, since the result on $\psi_I$ can be proved using an analogous method.

When the set $\mathcal{R}$ is defined via~\eqref{def:domainsRV}, the eigenproblem~\eqref{ass:KSKIaddMSMI}$_1$ reads as
\beq
\label{lemma1:eq0}
\begin{cases}
 \displaystyle{\int_{0}^1 K_{S}(r | r') \, \psi_{S}(t,r') \di r' = \psi_{S}(t,r)},
 \\\\
 \displaystyle{\int_{0}^1 \psi_S(t,r)  \di r = 1, \; \int_{0}^1 r \, \psi_S(t,r) \di r = M_S(t).}
\end{cases}
\eeq
Multiplying both sides of~\eqref{lemma1:eq0}$_1$ by $r^2$, integrating with respect to $r$ and rearranging terms yields
\beq
\label{lemma1:eq1}
\int_{0}^1 \left(\int_{0}^1 r^2 \, K_{S}(r | r') \di r \right) \psi_{S}(t,r') \di r' = \int_{0}^1 r^2 \, \psi_{S}(t,r) \di r.
\eeq
Introducing the notation
$$
\sigma^2_S(r') := \int_{0}^1 r^2 K_{S}(r | r') \di r - \left(\int_{0}^1 r K_{S}(r | r') \di r \right)^2,
$$
which, under assumption~\eqref{ass:Krv} on $K_S$, reduces to
$$
\sigma^2_S(r') = \int_{0}^1 r^2 K_{S}(r | r') \di r - \left(r'\right)^2,
$$
from~\eqref{lemma1:eq1} we obtain
$$
\int_{0}^1 \left(\sigma^2_S(r') + \left(r'\right)^2 \right) \psi_{S}(t,r') \di r' = \int_{0}^1 r^2 \psi_{S}(t,r) \di r \quad \Longrightarrow \quad \int_{0}^1 \sigma^2_S(r') \psi_{S}(t,r') \di r' = 0.
$$
This implies that
$$
\int_{0}^{\e} \sigma^2_S(r') \psi_{S}(t,r') \di r' + \int_{\e}^{1-\e} \sigma^2_S(r') \psi_{S}(t,r') \di r'  + \int_{1-\e}^{1} \sigma^2_S(r') \psi_{S}(t,r') \di r' = 0, \quad \forall \, \e \in \mathbb{R}_+^*
$$
and, therefore, since $\sigma^2_S$ and $\psi_{S}$ are non-negative,
\beq
\label{lemma1:eq2}
\int_{0}^{\e} \sigma^2_S(r') \psi_{S}(t,r') \di r' =0, \quad \int_{\e}^{1-\e} \sigma^2_S(r') \psi_{S}(t,r') \di r' = 0, \quad \int_{1-\e}^{1} \sigma^2_S(r') \psi_{S}(t,r') \di r' = 0, \quad \forall \, \e \in \mathbb{R}_+^*.
\eeq
Then we note that, under assumptions~\eqref{ass:KSb} and~\eqref{ass:Krv}, the following relations hold
\beq
\label{lemma1:eq3}
\sigma^2_S(r') > 0 \quad \forall r' \in (0,1), \qquad \sigma^2_S(0)=\sigma^2_S(1)=0.
\eeq
Since $\e \in \mathbb{R}_+^*$ in~\eqref{lemma1:eq2} can be chosen arbitrarily, the relations given by~\eqref{lemma1:eq2} along with the relation given by~\eqref{lemma1:eq3} and the normalisation condition given by~\eqref{lemma1:eq0} imply that ${\rm Supp} \left(\psi_{S}(t,\cdot)\right) = \{0,1\}$ for all $t \geq 0$, that is,
\beq
\label{lemma1:eq4}
\psi_{S}(t,r) = w_0(t) \delta_0(r) + w_1(t) \delta_1(r), \quad \forall t \geq 0.
\eeq
Due to the non-negativity of $\psi_{S}$, we have $w_0 : \mathbb{R}_+ \to \mathbb{R}^*_+$ and $w_1 : \mathbb{R}_+ \to \mathbb{R}^*_+$. Moreover, since
$$
\int_{0}^1 \big(w_0(t) \delta_0(r) + w_1(t) \delta_1(r)\big) \di r = w_0(t)+ w_1(t), \quad  \int_{0}^1 r \,  \big(w_0(t) \delta_0(r) + w_1(t) \delta_1(r)\big) \di r = w_1(t),
$$
inserting~\eqref{lemma1:eq4} in~\eqref{lemma1:eq0}$_2$ we find 
$$
\begin{cases}
w_0(t)+ w_1(t) = 1,
\\\\
w_1(t) = M_S(t),
\end{cases}
\quad 
\Longrightarrow
\quad
\begin{cases}
w_0(t) = 1 - M_S(t),
\\\\
w_1(t) = M_S(t).
\end{cases}
$$
Finally, substituting the above expressions of $w_0(t)$ and $w_1(t)$ into~\eqref{lemma1:eq4} we obtain the expression of $\psi_S$ given by~\eqref{eq:psiSpsiI}. $\qed$

\section{Additional results under assumptions~\texorpdfstring{\eqref{ass:KSKI_analysis2}}{}}
\label{sec:appD}
Without loss of generality, we focus on the case where $\mathcal{R}$ and $\mathcal{V}$ are defined via~\eqref{def:domainsRV}. Under assumptions~\eqref{ass:KSKI_analysis2}, the relations given by~\eqref{ass:KSKIaddMSMI} imply that
$$
\psi_S(t,r) = K_S(r), \quad \psi_I(t,v) = K_I(v) \quad \forall \, t \geq 0,
$$
and also that
\beq
\label{eq:MSMI_analysis2a}
M_S(t) = \int_0^1 r \, \psi_S(t,r) \, \di r = \int_0^1 r \, K_S(r) \, \di r =: \mu_{K_S} \in [0,1], \quad \forall \, t > 0,
\eeq
and
\beq
\label{eq:MSMI_analysis2b}
M_I(t) = \int_0^1 v \, \psi_I(t,v) \, \di v = \int_0^1 v \, K_I(v) \, \di v =: \mu_{K_I} \in [0,1] \quad \forall \, t > 0.
\eeq
Relations~\eqref{eq:MSMI_analysis2a} and~\eqref{eq:MSMI_analysis2b} lead to the following compatibility conditions
\beq
\label{eq:MSMI_analysis2comp}
M_S(0) = \mu_{K_S}, \quad M_I(0) = \mu_{K_I}.
\eeq
Hence, the macroscopic model reduces to relations~\eqref{eq:MSMI_analysis2a}-\eqref{eq:MSMI_analysis2comp} and the ODE system~\eqref{eq:ODEsNSNI}, which specifies to 
\beq
\label{eq:ODEsNSNI_analysis2}
\begin{cases}
\displaystyle{\ddt N_S = - N_S \, N_I \int_{0}^1 \int_{0}^1 q(r, v^*) \, K_S(r) \, K_I(v^*) \di v^* \di r},
\\\\
\displaystyle{\ddt N_I = N_I \, N_S \int_{0}^1 \int_{0}^1 q(r,v^*) \, K_S(r) \, K_I(v^*) \di v^* \di r}.
\end{cases}
\eeq
Moreover, if the function $q$ is defined as
\beq
\label{def:q_analysis1}
q(r,v^*) := (1 - r) \, v^*,
\eeq
the ODE system~\eqref{eq:ODEsNSNI_analysis2} reduces to the following system of ODEs
\beq
\label{eq:ODEsNSNIMSMI_analysis2red}
\begin{cases}
\displaystyle{\ddt N_S = - \mu_{K_I} \, N_I \, \left(1 - \mu_{K_S}\right) N_S},
\\\\
\displaystyle{\ddt N_I = \left(1 - \mu_{K_S} \right) \, N_S \, \mu_{K_I} \, N_I},
\end{cases}
\eeq
which we complement with the following initial condition (cf. assumptions~\eqref{eq:ODEsNSNIMSMIIC})
\beq
\label{eq:ODEsNSNIMSMIIC_analysis2}
N_S(0) = N_{S,0} \in (0,1), \quad N_I(0) = 1- N_{S,0}.
\eeq
The behaviour of the components of the solution to the Cauchy problem~\eqref{eq:ODEsNSNIMSMI_analysis2red}-\eqref{eq:ODEsNSNIMSMIIC_analysis2} is characterised by Proposition~\ref{prop1_analysis2}.
\begin{prpstn}
\label{prop1_analysis2}
The components of the solution to the Cauchy problem~\eqref{eq:ODEsNSNIMSMI_analysis2red}-\eqref{eq:ODEsNSNIMSMIIC_analysis2} are such that: 
\begin{itemize}
\item[] if 
\beq
\label{ass:appd1}
0 \leq \mu_{K_S} < 1 \quad \text{and} \quad 0 < \mu_{K_I} \leq 1
\eeq
then 
\beq
\label{eq:asylim1a_analysis2a}
\lim_{t \to \infty} N_S(t) = 0 \quad \text{and} \quad \lim_{t \to \infty} N_I(t) = 1;
\eeq
\item[] if 
\beq
\label{ass:appd2}
\mu_{K_S}=1 \quad \text{and/or} \quad \mu_{K_I}=0
\eeq
then 
\beq
\label{eq:asylim1a_analysis2b}
N_S(t) \equiv N_{S,0} \quad \text{and} \quad N_I(t) \equiv 1- N_{S,0}.
\eeq
\end{itemize}
\end{prpstn}

\begin{proof}
Using the fact that $N_I(t) = 1 - N_S(t)$ for all $t \geq 0$ (cf. the a priori estimates~\eqref{eq:aprioriest}), one can rewrite the ODE~\eqref{eq:ODEsNSNIMSMI_analysis2red} for $N_S$ as
$$
\ddt N_S(t) = - \mu_{K_I} \, \left(1 - \mu_{K_S}\right) \, (1 - N_S(t)) \, N_S(t).
$$
Hence, under the initial condition~\eqref{eq:ODEsNSNIMSMIIC_analysis2}, the results~\eqref{eq:asylim1a_analysis2a} and \eqref{eq:asylim1a_analysis2b} can be proved through very simple calculations, which are omitted here.
\end{proof}

The results of numerical simulations presented in Figures~\ref{V6R2_prop3.2a}-\ref{V4R7_prop3.2b} and Figures~\ref{V6R2_prop3.2c}-\ref{V4R7_prop3.2d} indicate that there is an excellent quantitative agreement between numerical solutions of the macroscopic model defined via the ODE system~\eqref{eq:ODEsNSNIMSMI_analysis2red} complemented with relations~\eqref{eq:MSMI_analysis2a}-\eqref{eq:MSMI_analysis2comp}, the results of Monte Carlo simulations of the corresponding individual-based model, subject to the compatibility conditions~\eqref{eq:MSMI_analysis2comp}, and the analytical results established by Proposition~\ref{prop1_analysis2}. These numerical results were obtained carrying out simulations using the set-up and numerical methods detailed in Section~\ref{sec:num:setup}, with the function $q$ defined via~\eqref{def:q_analysis1}, the kernel $Q$ defined via~\eqref{def:Qunif},  and the kernels $K_S$ and $K_I$ satisfying assumptions~\eqref{ass:KSKI_analysis2} and being defined either as the following generalised Pareto distributions
\beq \label{GPDs}
K_S(r):=\frac{1}{\sigma_S} \Biggr (1+\frac{k_S(r-\theta_S)}{\sigma_S}\Biggl)^{\bigr(-\frac{1}{k_S}-1\bigl)}, \quad \theta_S=0, \quad k_S=-\frac{\mu_{K_S}}{1-\mu_{K_S}}, \quad \sigma_S=\frac{\mu_{K_S}}{1-\mu_{K_S}}, \quad \mu_{K_S}\in (0,1),
\eeq
\beq \label{GPDi}
K_I(v):=\frac{1}{\sigma_I} \Biggr (1+\frac{k_I(v-\theta_I)}{\sigma_I}\Biggl)^{\bigr(-\frac{1}{k_I}-1\bigl)}, \quad \theta_I=0, \quad k_I=-\frac{\mu_{K_I}}{1-\mu_{K_I}}, \quad \sigma_I=\frac{\mu_{K_I}}{1-\mu_{K_I}}, \quad \mu_{K_I}\in (0,1), 
\eeq
which are such that
$$
\int_0^1 r \, K_S(r) \, \di r = \mu_{K_S}, \quad \int_0^1 v \, K_I(v) \, \di v = \mu_{K_I},
$$
or as 
\beq \label{GPDs2}
K_S(r):= \delta_{\mu_{K_S}}(r), \quad \mu_{K_S}=1,
\eeq
\beq \label{GPDi2}
K_I(v):= \delta_{\mu_{K_I}}(v), \quad \mu_{K_I}=0,
\eeq
so that either assumptions~\eqref{ass:appd1} or  assumptions~\eqref{ass:appd2} are satisfied.

\begin{figure}[H]
\centering
\subcaptionbox{\label{V6R2_prop3.2a}}
{\includegraphics[trim={0 7cm 0 7cm},clip, width=.49\textwidth]{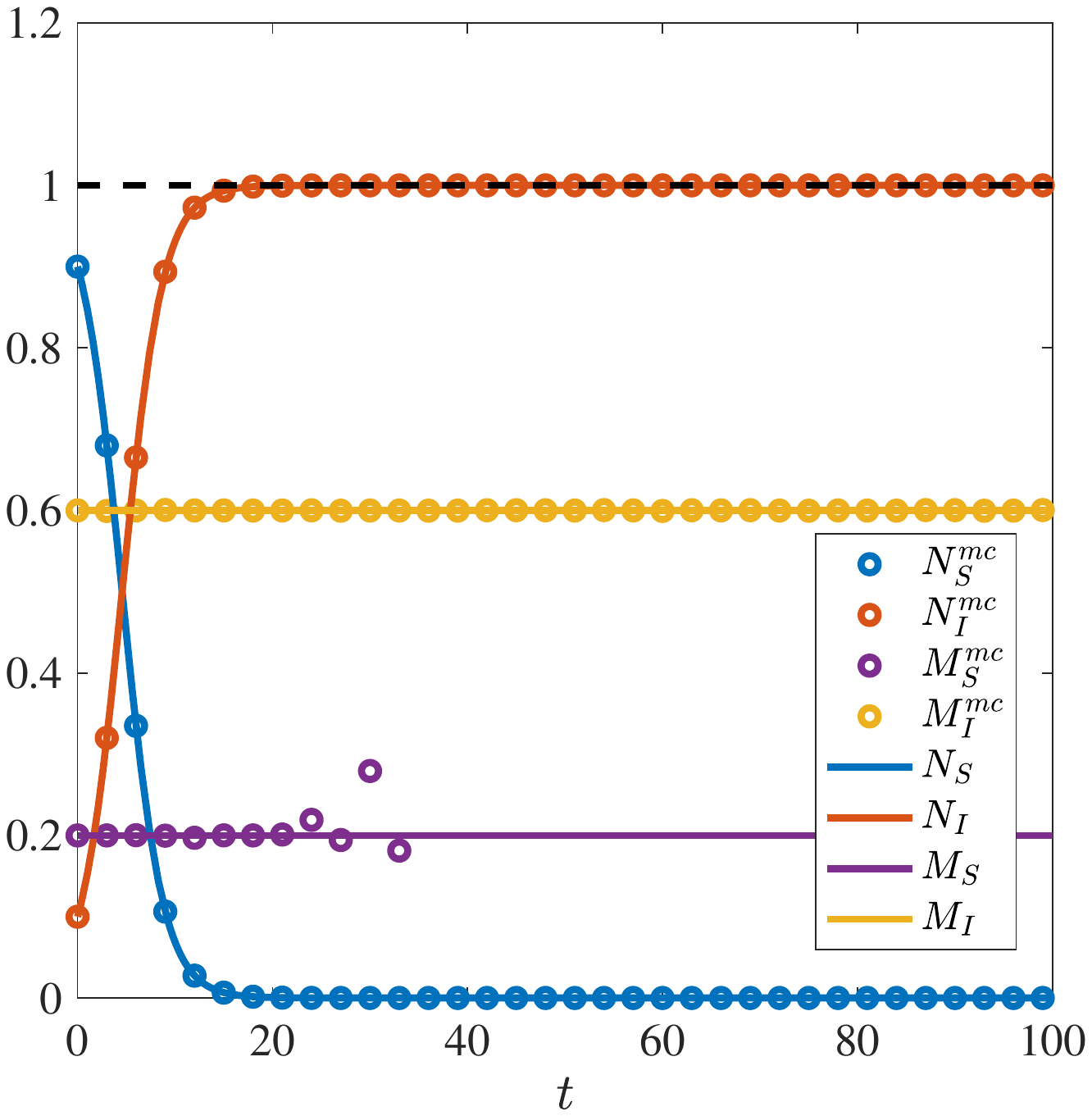}}
\subcaptionbox{\label{V4R7_prop3.2b}} 
{\includegraphics[trim={0 7cm 0 7cm},clip, width=.49\textwidth]{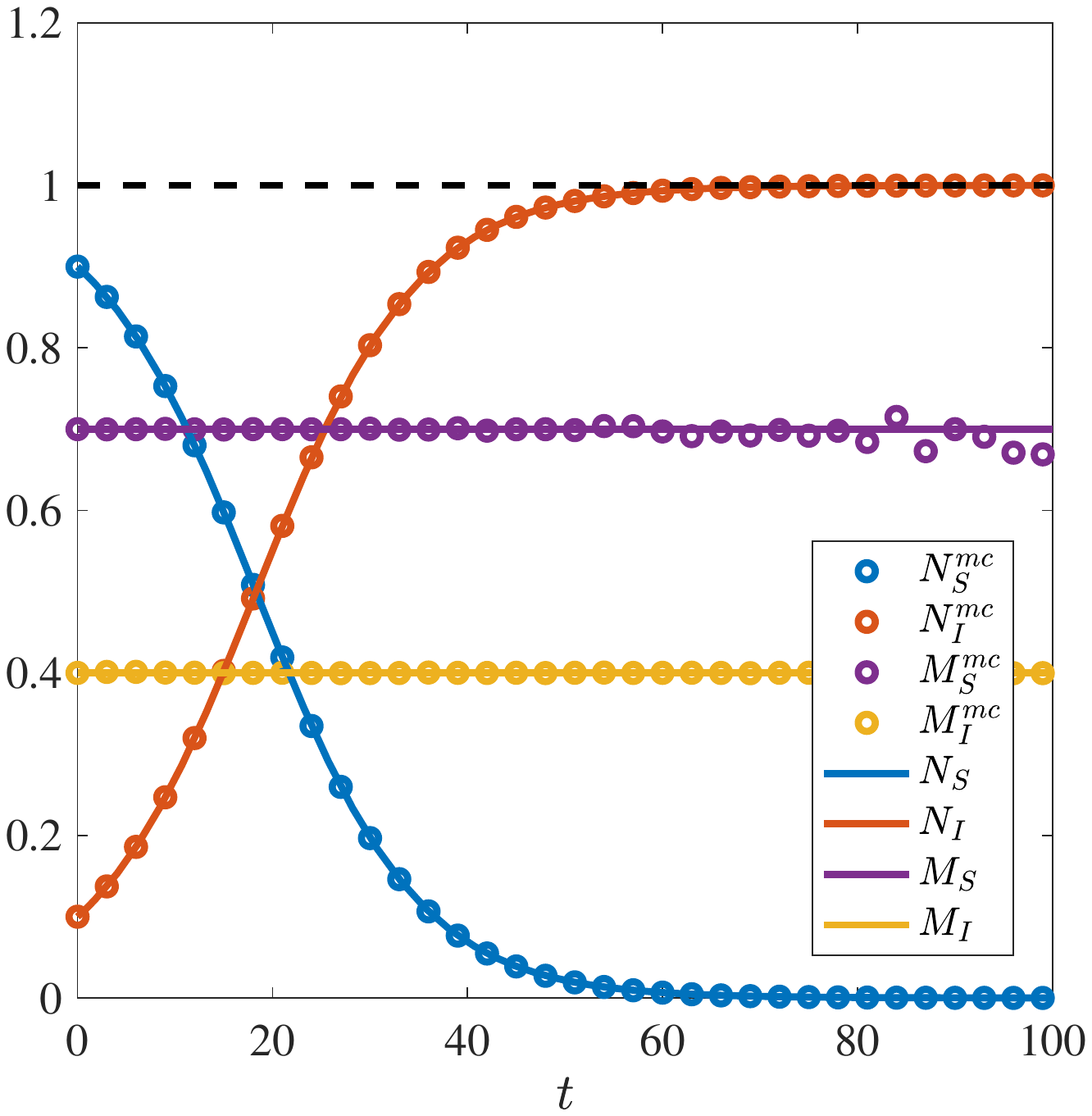}}
\subcaptionbox{\label{V6R2_prop3.2c}}
{\includegraphics[trim={0 7cm 0 7cm},clip, width=.49\textwidth]{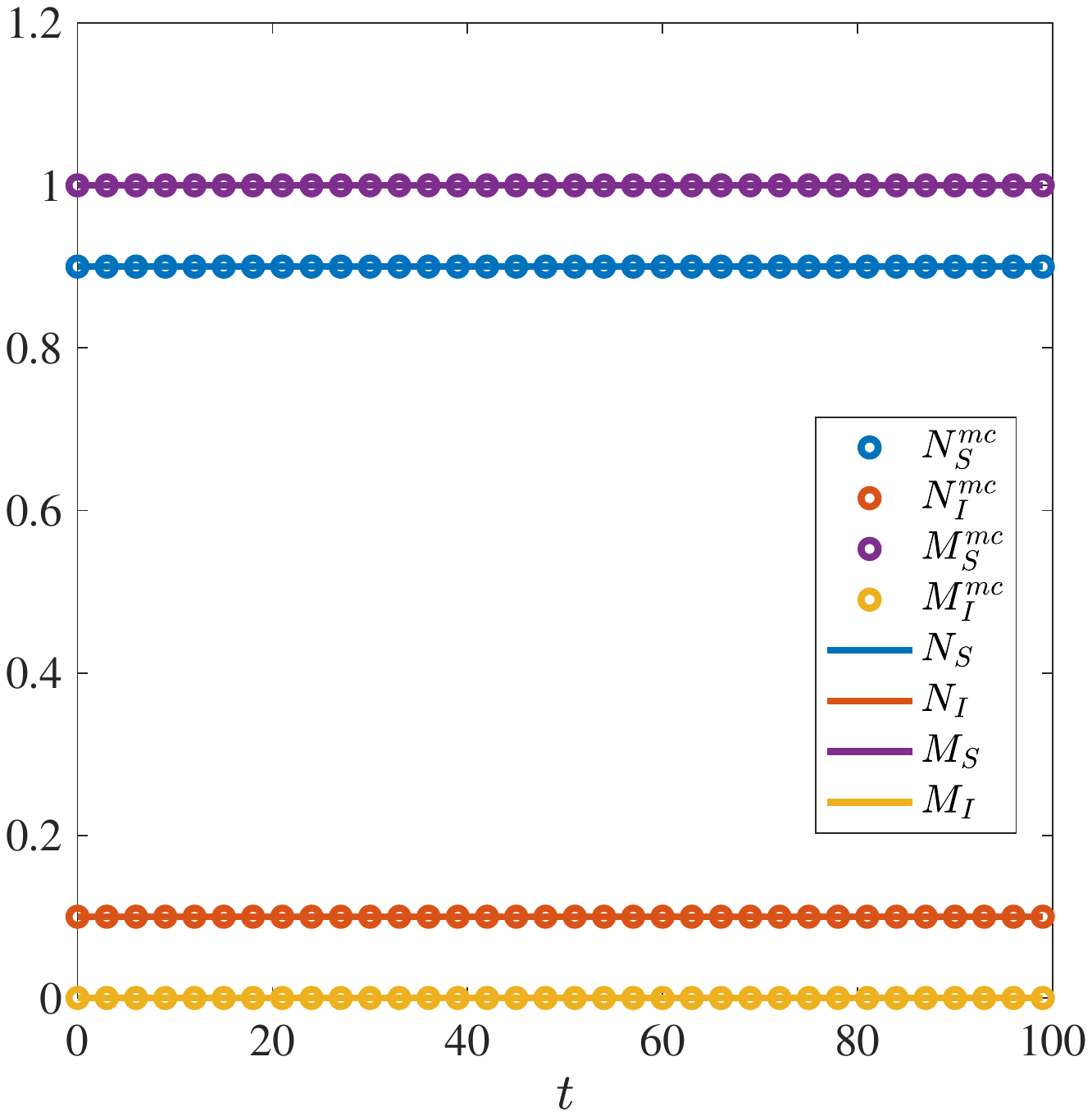}}
\subcaptionbox{\label{V4R7_prop3.2d}} 
{\includegraphics[trim={0 7cm 0 7cm},clip, width=.49\textwidth]{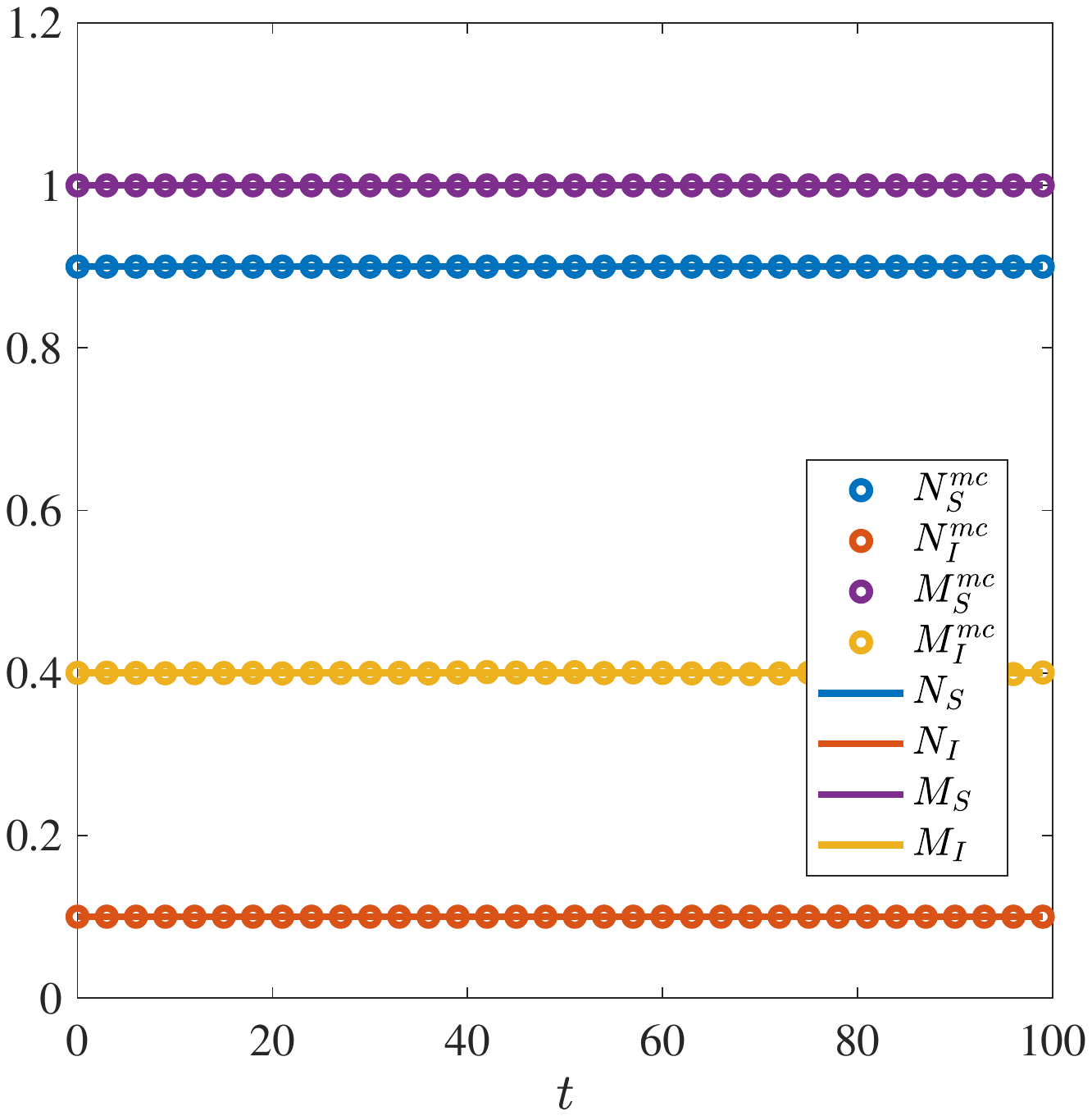}}
\caption{{\bf Numerical results under the assumptions of Proposition~\ref{prop1_analysis2}.} Dynamics of $N_S(t)$ (blue lines) and $N_I(t)$ (orange lines) obtained by solving numerically the macroscopic model defined via the ODE system~\eqref{eq:ODEsNSNIMSMI_analysis2red}, subject to initial conditions~\eqref{eq:ODEsNSNIMSMIIC_analysis2}, complemented with~\eqref{eq:MSMI_analysis2a}-\eqref{eq:MSMI_analysis2comp}. The corresponding values of $M_S(t)$ and $M_I(t)$ defined via~\eqref{eq:MSMI_analysis2a}-\eqref{eq:MSMI_analysis2comp} are also displayed (purple and yellow lines, respectively). The dynamics of the corresponding quantities obtained through Monte Carlo simulations of the individual-based model, subject to the compatibility conditions~\eqref{eq:MSMI_analysis2comp}, are highlighted by circular markers, while the black lines highlight the asymptotic values given by~\eqref{eq:asylim1a_analysis2a} (panels (a) and (b)). These results are obtained in the case where the function $q$ is defined via~\eqref{def:q_analysis1}, the kernel $Q$ is defined via~\eqref{def:Qunif}, and the kernels $K_S$ and $K_I$ satisfy assumptions~\eqref{ass:KSKI_analysis2} and are defined: either via~\eqref{GPDs} and \eqref{GPDi} with $\mu_{K_S}=0.2$ and $\mu_{K_I}=0.6$ (panel (a)); or via~\eqref{GPDs} and \eqref{GPDi} with $\mu_{K_S}=0.7$ and $\mu_{K_I}=0.4$ (panel (b)); or via~\eqref{GPDs2} and \eqref{GPDi2} (panel (c)); or via~\eqref{GPDs2} and \eqref{GPDi} with $\mu_{K_I}=0.4$.}
\end{figure}

\bibliographystyle{plain}
\bibliography{LtPeTa-SI_rv}

\begin{thebibliography}{10}

\bibitem{abi2021asymptotic}
L.~Abi~Rizk, J.-B. Burie, and A.~Ducrot.
\newblock Asymptotic speed of spread for a nonlocal evolutionary-epidemic
  system.
\newblock {\em Discrete Continuous Dyn. Syst.}, 41(10):4959--4985, 2021.

\bibitem{albi2022kinetic}
Giacomo Albi, Giulia Bertaglia, Walter Boscheri, Giacomo Dimarco, Lorenzo
  Pareschi, Giuseppe Toscani, and Mattia Zanella.
\newblock Kinetic modelling of epidemic dynamics: social contacts, control with
  uncertain data, and multiscale spatial dynamics.
\newblock In {\em Predicting Pandemics in a Globally Connected World, Volume 1:
  Toward a Multiscale, Multidisciplinary Framework through Modeling and
  Simulation}, pages 43--108. Springer, 2022.

\bibitem{albi2021control}
Giacomo Albi, Lorenzo Pareschi, and Mattia Zanella.
\newblock Control with uncertain data of socially structured compartmental
  epidemic models.
\newblock {\em J. Math. Biol.}, 82:1--41, 2021.

\bibitem{almeida2021final}
Lu{\'\i}s Almeida, Pierre-Alexandre Bliman, Gr{\'e}goire Nadin, Beno{\^\i}t
  Perthame, and Nicolas Vauchelet.
\newblock Final size and convergence rate for an epidemic in heterogeneous
  populations.
\newblock {\em Math. Models Methods Appl. Sci.}, 31(05):1021--1051, 2021.

\bibitem{atlante2020epigenetic}
Sandra Atlante, Alessia Mongelli, Veronica Barbi, Fabio Martelli, Antonella
  Farsetti, and Carlo Gaetano.
\newblock The epigenetic implication in coronavirus infection and therapy.
\newblock {\em Clin. Epigenetics}, 12:1--12, 2020.

\bibitem{banerjee2020immuno}
Malay Banerjee, Alexey Tokarev, and Vitaly Volpert.
\newblock Immuno-epidemiological model of two-stage epidemic growth.
\newblock {\em Math. Model. Nat. Phenom.}, 15:27, 2020.

\bibitem{barbarossa2015immuno}
Maria~Vittoria Barbarossa and Gergely R{\"o}st.
\newblock Immuno-epidemiology of a population structured by immune status: a
  mathematical study of waning immunity and immune system boosting.
\newblock {\em J. Math. Biol.}, 71:1737--1770, 2015.

\bibitem{bernardi2022effects}
Emanuele Bernardi, Lorenzo Pareschi, Giuseppe Toscani, and Mattia Zanella.
\newblock Effects of vaccination efficacy on wealth distribution in kinetic
  epidemic models.
\newblock {\em Entropy}, 24(2):216, 2022.

\bibitem{borra2013asymptotic}
Domenica Borra and Tommaso Lorenzi.
\newblock Asymptotic analysis of continuous opinion dynamics models under
  bounded confidence.
\newblock {\em Commun. Pure Appl. Anal.}, 12(3):1487--1499, 2013.

\bibitem{britton2020mathematical}
Tom Britton, Frank Ball, and Pieter Trapman.
\newblock A mathematical model reveals the influence of population
  heterogeneity on herd immunity to sars-cov-2.
\newblock {\em Science}, 369(6505):846--849, 2020.

\bibitem{burie2020asymptotic}
J.-B. Burie, R.~Djidjou-Demasse, and A.~Ducrot.
\newblock Asymptotic and transient behaviour for a nonlocal problem arising in
  population genetics.
\newblock {\em Eur. J. Appl. Math.}, 31(1):84--110, 2020.

\bibitem{burie2020slow}
J.-B. Burie, R.~Djidjou-Demasse, and A.~Ducrot.
\newblock Slow convergence to equilibrium for an evolutionary epidemiology
  integro--differential system.
\newblock {\em Discrete Continuous Dyn. Syst. Ser. B}, 25(6):2223, 2020.

\bibitem{burie2020concentration}
J.-B. Burie, A.~Ducrot, Q.~Griette, and Q.~Richard.
\newblock Concentration estimates in a multi-host epidemiological model
  structured by phenotypic traits.
\newblock {\em J. Differ. Equ.}, 269(12):11492--11539, 2020.

\bibitem{busenberg1991global}
S.~N. Busenberg, M.~Iannelli, and H.~R. Thieme.
\newblock Global behavior of an age-structured epidemic model.
\newblock {\em SIAM J. Math. Anal.}, 22(4):1065--1080, 1991.

\bibitem{chabas2018evolutionary}
H{\'e}l{\`e}ne Chabas, S{\'e}bastien Lion, Antoine Nicot, Sean Meaden, Stineke
  van Houte, Sylvain Moineau, Lindi~M Wahl, Edze~R Westra, and Sylvain Gandon.
\newblock Evolutionary emergence of infectious diseases in heterogeneous host
  populations.
\newblock {\em PLoS Biol.}, 16(9):e2006738, 2018.

\bibitem{delitala2012asymptotic}
Marcello Delitala and Tommaso Lorenzi.
\newblock Asymptotic dynamics in continuous structured populations with
  mutations, competition and mutualism.
\newblock {\em J. Math. Anal. Appl.}, 389(1):439--451, 2012.

\bibitem{della2021sir}
Rossella Della~Marca, Nadia Loy, and Andrea Tosin.
\newblock An {SIR}-like kinetic model tracking individuals' viral load.
\newblock {\em Netw. Heterog. Media}, 17(3):467--494, 2022.

\bibitem{della2022sir}
Rossella Della~Marca, Nadia Loy, and Andrea Tosin.
\newblock An {SIR} model with viral load--dependent transmission.
\newblock {\em J. Math. Biol.}, 86(61), 2023.

\bibitem{diekmann1978thresholds}
Odo Diekmann.
\newblock Thresholds and travelling waves for the geographical spread of
  infection.
\newblock {\em J. Math. Biol.}, 6(2):109--130, 1978.

\bibitem{diekmann2013mathematical}
Odo Diekmann, Hans Heesterbeek, and Tom Britton.
\newblock {\em {M}athematical {T}ools for {U}nderstanding {I}nfectious
  {D}isease {D}ynamics}, volume~7.
\newblock Princeton University Press, 2013.

\bibitem{diekmann2000mathematical}
Odo Diekmann and Johan Andre~Peter Heesterbeek.
\newblock {\em Mathematical Epidemiology of Infectious Diseases: model
  building, analysis and interpretation}, volume~5.
\newblock John Wiley \& Sons, 2000.

\bibitem{dimarco2021kinetic}
G~Dimarco, B~Perthame, G~Toscani, and M~Zanella.
\newblock Kinetic models for epidemic dynamics with social heterogeneity.
\newblock {\em J. Math. Biol.}, 83(1):4, 2021.

\bibitem{dimarco2020wealth}
Giacomo Dimarco, Lorenzo Pareschi, Giuseppe Toscani, and Mattia Zanella.
\newblock Wealth distribution under the spread of infectious diseases.
\newblock {\em Phys. Rev. E}, 102(2):022303, 2020.

\bibitem{djidjou2017steady}
R.~Djidjou-Demasse, A.~Ducrot, and F.~Fabre.
\newblock Steady state concentration for a phenotypic structured problem
  modeling the evolutionary epidemiology of spore producing pathogens.
\newblock {\em Math. Models Methods Appl. Sci.}, 27(02):385--426, 2017.

\bibitem{du2022within}
Zhanwei Du, Shuqi Wang, Yuan Bai, Chao Gao, Eric~HY Lau, and Benjamin~J
  Cowling.
\newblock Within-host dynamics of {SARS}-{C}o{V}-2 infection: {A} systematic
  review and meta-analysis.
\newblock {\em Transbound. Emerg. Dis.}, 69(6):3964--3971, 2022.

\bibitem{ducasse2022threshold}
Romain Ducasse.
\newblock Threshold phenomenon and traveling waves for heterogeneous integral
  equations and epidemic models.
\newblock {\em Nonlinear Anal.}, 218:112788, 2022.

\bibitem{elie2022source}
Baptiste Elie, Christian Selinger, and Samuel Alizon.
\newblock The source of individual heterogeneity shapes infectious disease
  outbreaks.
\newblock {\em Proc. Royal Soc. B}, 289(1974):20220232, 2022.

\bibitem{fraser2007variation}
Christophe Fraser, T~D{\'e}irdre Hollingsworth, Ruth Chapman, Frank de~Wolf,
  and William~P Hanage.
\newblock Variation in {HIV}-1 set-point viral load: epidemiological analysis
  and an evolutionary hypothesis.
\newblock {\em Proc. Natl. Acad. Sci. U.S.A.}, 104(44):17441--17446, 2007.

\bibitem{gandolfi2015epidemic}
Alberto Gandolfi, Andrea Pugliese, and Carmela Sinisgalli.
\newblock Epidemic dynamics and host immune response: a nested approach.
\newblock {\em J. Math. Biol.}, 70:399--435, 2015.

\bibitem{gomez2012epigenetics}
Elena G{\'o}mez-D{\'\i}az, Mireia Jorda, Miguel~Angel Peinado, and Ana Rivero.
\newblock Epigenetics of host--pathogen interactions: the road ahead and the
  road behind.
\newblock {\em PLoS Pathog.}, 8(11):e1003007, 2012.

\bibitem{gutierrez2012circulating}
Serafin Gutierrez, Michel Yvon, Elodie Pirolles, Eliza Garzo, Alberto Fereres,
  Yannis Michalakis, and St{\'e}phane Blanc.
\newblock Circulating virus load determines the size of bottlenecks in viral
  populations progressing within a host.
\newblock {\em PLoS Pathog.}, 8(11):e1003009, 2012.

\bibitem{hadjichrysanthou2016understanding}
Christoforos Hadjichrysanthou, Emilie Cau{\"e}t, Emma Lawrence, Carolin
  Vegvari, Frank De~Wolf, and Roy~M Anderson.
\newblock Understanding the within-host dynamics of influenza a virus: from
  theory to clinical implications.
\newblock {\em J. R. Soc. Interface}, 13(119):20160289, 2016.

\bibitem{hay2021estimating}
James~A Hay, Lee Kennedy-Shaffer, Sanjat Kanjilal, Niall~J Lennon, Stacey~B
  Gabriel, Marc Lipsitch, and Michael~J Mina.
\newblock Estimating epidemiologic dynamics from cross-sectional viral load
  distributions.
\newblock {\em Science}, 373(6552):eabh0635, 2021.

\bibitem{2}
H.~Hethcote.
\newblock The {M}athematics of {I}nfectious {D}iseases.
\newblock {\em SIAM Rev.}, 42(4):599--653, 2000.

\bibitem{hill1996genetics}
Adrian~VS Hill.
\newblock Genetics of infectious disease resistance.
\newblock {\em Curr. Opin. Genet. Dev.}, 6(3):348--353, 1996.

\bibitem{3}
A.~Huppert and G.~Katriel.
\newblock Mathematical modelling and prediction in infectious disease
  epidemiology.
\newblock {\em Clin. Microbiol. Infect.}, 19(11):999--1005, 2013.

\bibitem{iacono2012evolution}
Giovanni~Lo Iacono, Frank van~den Bosch, and Neil Paveley.
\newblock The evolution of plant pathogens in response to host resistance:
  factors affecting the gain from deployment of qualitative and quantitative
  resistance.
\newblock {\em J. Theor. Biol.}, 304:152--163, 2012.

\bibitem{iannelli2015introduction}
M.~Iannelli and A.~Pugliese.
\newblock {\em {An Introduction to Mathematical Population Dynamics: Along the
  Trail of Volterra and Lotka}}, volume~79.
\newblock Springer, 2015.

\bibitem{inaba1990threshold}
H.~Inaba.
\newblock Threshold and stability results for an age-structured epidemic model.
\newblock {\em J. Math. Biol.}, 28(4):411--434, 1990.

\bibitem{inaba2012new}
H.~Inaba.
\newblock On a new perspective of the basic reproduction number in
  heterogeneous environments.
\newblock {\em J. Math. Biol.}, 65(2):309--348, 2012.

\bibitem{inaba2017age}
H.~Inaba.
\newblock {\em {A}ge-{S}tructured {P}opulation {D}ynamics in {D}emography and
  {E}pidemiology}.
\newblock Springer, 2017.

\bibitem{jones2021viral}
Jennifer~E Jones, Valerie Le~Sage, and Seema~S Lakdawala.
\newblock Viral and host heterogeneity and their effects on the viral life
  cycle.
\newblock {\em Nat. Rev. Microbiol.}, 19(4):272--282, 2021.

\bibitem{karev2019trait}
Georgy~P Karev and Artem~S Novozhilov.
\newblock How trait distributions evolve in populations with parametric
  heterogeneity.
\newblock {\em Math. Biosci.}, 315:108235, 2019.

\bibitem{karlsson2014natural}
Elinor~K Karlsson, Dominic~P Kwiatkowski, and Pardis~C Sabeti.
\newblock Natural selection and infectious disease in human populations.
\newblock {\em Nat. Rev. Genet.}, 15(6):379--393, 2014.

\bibitem{kimberly2016heterogeneity}
L~Kimberly, O~Vanessa, et~al.
\newblock Heterogeneity in pathogen transmission: mechanisms and methodology.
\newblock {\em Funct. Ecol.}, 2016.

\bibitem{kwok2021host}
Andrew~J Kwok, Alex Mentzer, and Julian~C Knight.
\newblock Host genetics and infectious disease: new tools, insights and
  translational opportunities.
\newblock {\em Nat. Rev. Genet.}, 22(3):137--153, 2021.

\bibitem{lorenzi2021evolutionary}
Tommaso Lorenzi, Andrea Pugliese, Mattia Sensi, and Agnese Zardini.
\newblock Evolutionary dynamics in an {SI} epidemic model with
  phenotype-structured susceptible compartment.
\newblock {\em J. Math. Biol.}, 83(6-7):72, 2021.

\bibitem{lou2011reaction}
Y.~Lou and X.-Q. Zhao.
\newblock A reaction--diffusion malaria model with incubation period in the
  vector population.
\newblock {\em J. Math. Biol.}, 62(4):543--568, 2011.

\bibitem{loy2020boltzmann}
Nadia Loy and Andrea Tosin.
\newblock Boltzmann-type equations for multi-agent systems with label
  switching.
\newblock {\em Kinet. Relat. Models}, 14(5):867--894, 2020.

\bibitem{loy2021viral}
Nadia Loy and Andrea Tosin.
\newblock A viral load-based model for epidemic spread on spatial networks.
\newblock {\em Math. Biosci. Eng.}, 18(5):5635--5663, 2021.

\bibitem{metz2014dynamics}
Johan~A Metz and Odo Diekmann.
\newblock {\em {T}he {D}ynamics of {P}hysiologically {S}tructured
  {P}opulations}, volume~68.
\newblock Springer, 2014.

\bibitem{novozhilov2012epidemiological}
Artem~S Novozhilov.
\newblock Epidemiological models with parametric heterogeneity: Deterministic
  theory for closed populations.
\newblock {\em Math. Model. Nat. Phenom.}, 7(3):147--167, 2012.

\bibitem{nowak2000virus}
Martin Nowak and Robert~M May.
\newblock {\em {V}irus {D}ynamics: {M}athematical {P}rinciples of {I}mmunology
  and {V}irology}.
\newblock Oxford University Press, UK, 2000.

\bibitem{pareschi2013BOOK}
L.~Pareschi and G.~Toscani.
\newblock {\em Interacting {M}ultiagent {S}ystems: {K}inetic equations and
  {M}onte {C}arlo methods}.
\newblock Oxford University Press, 2013.

\bibitem{peng2012reaction}
R.~Peng and X.-Q. Zhao.
\newblock A reaction--diffusion {SIS} epidemic model in a time-periodic
  environment.
\newblock {\em Nonlinearity}, 25(5):1451, 2012.

\bibitem{perthame2006transport}
Beno{\^\i}t Perthame.
\newblock {\em {T}ransport {E}quations in {B}iology}.
\newblock Springer Science \& Business Media, 2006.

\bibitem{puhach2023sars}
Olha Puhach, Benjamin Meyer, and Isabella Eckerle.
\newblock {SARS}-{C}o{V}-2 viral load and shedding kinetics.
\newblock {\em Nat. Rev. Microbiol.}, 21(3):147--161, 2023.

\bibitem{quintana2019human}
Lluis Quintana-Murci.
\newblock Human immunology through the lens of evolutionary genetics.
\newblock {\em Cell}, 177(1):184--199, 2019.

\bibitem{thieme2009spectral}
H.~R. Thieme.
\newblock Spectral bound and reproduction number for infinite-dimensional
  population structure and time heterogeneity.
\newblock {\em SIAM J. Appl. Math.}, 70(1):188--211, 2009.

\bibitem{thieme1977model}
HR~Thieme.
\newblock A model for the spatial spread of an epidemic.
\newblock {\em J. Math. Biol.}, 4(4):337--351, 1977.

\bibitem{tkachenko2021time}
Alexei~V Tkachenko, Sergei Maslov, Ahmed Elbanna, George~N Wong, Zachary~J
  Weiner, and Nigel Goldenfeld.
\newblock Time-dependent heterogeneity leads to transient suppression of the
  covid-19 epidemic, not herd immunity.
\newblock {\em Proc. Natl. Acad. Sci. U.S.A.}, 118(17):e2015972118, 2021.

\bibitem{trauer2019importance}
James~M Trauer, Peter~J Dodd, M~Gabriela~M Gomes, Gabriela~B Gomez, Rein~MGJ
  Houben, Emma~S McBryde, Yayehirad~A Melsew, Nicolas~A Menzies, Nimalan
  Arinaminpathy, Sourya Shrestha, et~al.
\newblock The importance of heterogeneity to the epidemiology of tuberculosis.
\newblock {\em Clin. Infect. Dis.}, 69(1):159--166, 2019.

\bibitem{veliov2005effect}
Vladimir~M Veliov.
\newblock On the effect of population heterogeneity on dynamics of epidemic
  diseases.
\newblock {\em J. Math. Biol.}, 51(2):123--143, 2005.

\bibitem{wang2011nonlocal}
W.~Wang and X.-Q. Zhao.
\newblock A nonlocal and time-delayed reaction-diffusion model of dengue
  transmission.
\newblock {\em SIAM J. Appl. Math.}, 71(1):147--168, 2011.

\bibitem{wang2012basic}
W.~Wang and X.-Q. Zhao.
\newblock Basic reproduction numbers for reaction-diffusion epidemic models.
\newblock {\em SIAM J. Appl. Dyn. Syst.}, 11(4):1652--1673, 2012.

\bibitem{woolhouse1997heterogeneities}
Mark~EJ Woolhouse, C~Dye, J-F Etard, T~Smith, JD~Charlwood, GP~Garnett,
  P~Hagan, JL~xK Hii, PD~Ndhlovu, RJ~Quinnell, et~al.
\newblock Heterogeneities in the transmission of infectious agents:
  implications for the design of control programs.
\newblock {\em Proc. Natl. Acad. Sci. U.S.A.}, 94(1):338--342, 1997.

\bibitem{yates2006pathogen}
Andrew Yates, Rustom Antia, and Roland~R Regoes.
\newblock How do pathogen evolution and host heterogeneity interact in disease
  emergence?
\newblock {\em Proc. Royal Soc. B}, 273(1605):3075--3083, 2006.

\bibitem{zanella2022kinetic}
Mattia Zanella.
\newblock Kinetic models for epidemic dynamics in the presence of opinion
  polarization.
\newblock {\em Bull. Math. Biol.}, 85(36), 2023.

\bibitem{zhang2022monitoring}
Yunjun Zhang, Tom Britton, and Xiaohua Zhou.
\newblock Monitoring real-time transmission heterogeneity from incidence data.
\newblock {\em PLoS Comput. Biol.}, 18(12):e1010078, 2022.

\end{thebibliography}

\end{document}